\documentclass[conference]{IEEEtran}
\ifCLASSINFOpdf
\else
\fi

\usepackage{cite}
\usepackage{amsmath,amssymb,amsfonts,amsthm}
\usepackage{mathabx}
\usepackage{algorithmic}
\usepackage{graphicx}
\usepackage{textcomp}
\usepackage{booktabs}
\usepackage[hidelinks]{hyperref}

\usepackage{orcidlink}
\usepackage[most]{tcolorbox}
\usepackage{xcolor,soul}
\usepackage{tikz}
\usepackage{subcaption}
\usepackage{bm}
\usepackage{float}
\usepackage{stfloats}
\usepackage[cal=cm]{mathalfa}
\usepackage{pifont}
\usepackage{verbatim}
\usepackage{multirow}
\usepackage{enumitem}
\usepackage{arydshln}
\usepackage{url}
\allowdisplaybreaks[4]
\definecolor{MyColor}{RGB}{182,222,192} % Green background
\AtBeginDocument{%
  }

\newtheorem{definition}{Definition}
\newtheorem{theorem}{Theorem}

\newcommand{\grey}{\textcolor[rgb]{0.5, 0.5, 0.5}}
\newcommand{\blue}{\textcolor[rgb]{0.0, 0.0, 1.0}}

\newcommand{\violet}{\textcolor[rgb]{0.58, 0.0, 0.83}}

\newtcolorbox{mybox}[2][]{%
  sidebyside align = top seam,
  colback      = black!5!white,
  colframe     = black!75!black,
  colbacktitle = gray!85!white,
  title        = #2,#1,
  enhanced,
}

\begin{document}
\title{PriSrv+: Privacy and Usability-Enhanced \\
Wireless Service Discovery with Fast and
 Expressive Matchmaking Encryption}

% \author{\em Anonymous Authors}

\author{\IEEEauthorblockN{Yang Yang\IEEEauthorrefmark{1},
Guomin Yang\IEEEauthorrefmark{1}, 
Yingjiu Li\IEEEauthorrefmark{2},
Pengfei Wu\IEEEauthorrefmark{1}, 
Rui Shi\IEEEauthorrefmark{3},
Minming Huang\IEEEauthorrefmark{1},\\
Jian Weng\IEEEauthorrefmark{4}, 
HweeHwa Pang\IEEEauthorrefmark{1}, 
Robert H. Deng\IEEEauthorrefmark{1}
}

\IEEEauthorblockA{\IEEEauthorrefmark{1}Singapore Management University, Singapore\\ (\texttt{\{yyang, gmyang, pfwu, mmhuang, hhpang, robertdeng\}@smu.edu.sg)}}
\IEEEauthorblockA{\IEEEauthorrefmark{2}University of Oregon, USA (\texttt{yingjiul@uoregon.edu})}
\IEEEauthorblockA{\IEEEauthorrefmark{3}Hainan University, China (\texttt{shir@hainanu.edu.cn})}
\IEEEauthorblockA{\IEEEauthorrefmark{4}Jinan University, Guangzhou, China (\texttt{cryptjweng@gmail.com})}
}


% \IEEEoverridecommandlockouts
% \makeatletter\def\@IEEEpubidpullup{6.5\baselineskip}\makeatother
% \IEEEpubid{\parbox{\columnwidth}{
% 		Network and Distributed System Security (NDSS) Symposium 2026\\
% 		23-27 February 2026, San Diego, CA, USA\\
% 		ISBN 979-8-9919276-8-0\\
% 		https://dx.doi.org/10.14722/ndss.2026.230087\\
% 		www.ndss-symposium.org
% }
% \hspace{\columnsep}\makebox[\columnwidth]{}}

\IEEEoverridecommandlockouts
\makeatletter\def\@IEEEpubidpullup{3\baselineskip}\makeatother
\IEEEpubid{\parbox{\columnwidth}{
    This is the full version of the research work published in \textit{Network and Distributed System Security Symposium (NDSS) 2026}.\\
    https://dx.doi.org/10.14722/ndss.2026.230087\\
}
\hspace{\columnsep}\makebox[\columnwidth]{}}

\maketitle

\begin{abstract}
Service discovery is a fundamental process in wireless networks, enabling devices to find and communicate with services dynamically, and is critical for the seamless operation of modern systems like 5G and IoT. This paper introduces PriSrv+, an advanced privacy and usability-enhanced service discovery protocol for modern wireless networks and resource-constrained environments. PriSrv+ builds upon PriSrv (NDSS'24), by addressing critical limitations in expressiveness, privacy, scalability, and efficiency, while maintaining compatibility with widely-used wireless protocols such as mDNS, BLE, and Wi-Fi. 

A key innovation in PriSrv+ is the development of Fast and Expressive Matchmaking Encryption (FEME), the first matchmaking encryption scheme capable of supporting expressive access control policies with an unbounded attribute universe, allowing any arbitrary string to be used as an attribute. FEME significantly enhances the flexibility of service discovery while ensuring robust message and attribute privacy. Compared to PriSrv, PriSrv+ optimizes cryptographic operations, achieving 7.62$\times$ faster for encryption and 6.23$\times$ faster for decryption, and dramatically reduces ciphertext sizes by 87.33$\%$. In addition, PriSrv+ reduces communication costs by 87.33$\%$ for service broadcast and 86.64$\%$ for anonymous mutual authentication compared with PriSrv. Formal security proofs confirm the security of FEME and PriSrv+. Extensive evaluations on multiple platforms demonstrate that PriSrv+ achieves superior performance, scalability, and efficiency compared to existing state-of-the-art protocols.
\end{abstract}


\IEEEpeerreviewmaketitle

% Contents
% \makefirstuc{we present}  make first letter upper case

% \newcommand{\name}{AKMA+}

% \newcommand{\namep}{AKMA'} 

%------------Scheme---------------------

\newcommand{\FEME}{\mathcal{FEME}} 

%------------Entity---------------------

\newcommand{\snd}{\textsf{snd}} 

\newcommand{\rcv}{\textsf{rcv}}

%------------Algo---------------------

\newcommand{\Setup}{\textsf{Setup}} 

\newcommand{\EKGen}{\textsf{EKGen}} 

\newcommand{\DKGen}{\textsf{DKGen}} 

\newcommand{\PolGen}{\textsf{PolGen}} 

\newcommand{\Enc}{\textsf{Enc}} 

\newcommand{\Dec}{\textsf{Dec}} 

\newcommand{\KeyGen}{\textsf{KeyGen}}       % MAC algo

\newcommand{\MAC}{\textsf{MAC}}             % MAC algo

\newcommand{\Verify}{\textsf{Verify}}       % MAC algo
 
%------------Param---------------------

\newcommand{\mpk}{\textsf{mpk}} 

\newcommand{\msk}{\textsf{msk}} 

\newcommand{\EK}{\textsf{EK}} 

\newcommand{\ek}{\textsf{ek}} 

\newcommand{\DK}{\textsf{DK}} 

\newcommand{\dk}{\textsf{dk}} 

\newcommand{\SK}{\textsf{SK}} 

\newcommand{\sk}{\textsf{sk}} 

\newcommand{\msg}{\textsf{msg}} 

\newcommand{\MSG}{\textsf{MSG}} 

\newcommand{\CT}{\textsf{CT}} 

\newcommand{\ct}{\textsf{ct}} 

\newcommand{\cred}{\textsf{cred}}

\newcommand{\pair}{\textmd{pair}} 

\newcommand{\A}{\textbf{A}} 

\newcommand{\M}{\textbf{M}} 

\newcommand{\sfv}{\textbf{v}} 

\newcommand{\sfy}{\textbf{y}}

%------------Complex Symbols---------------------

\newcommand{\Ssnd}{\mathcal_\textsf{snd}}

\newcommand{\gm}[1]{\textcolor{blue}{[guomin: #1]}}
\newcommand{\yy}[1]{\textcolor{orange}{#1}}

\section{Introduction}
\label{sec:Intro}

Service discovery (SD) protocols, including Wi-Fi~\cite{WiFi}, AirDrop~\cite{AirDrop}, and BLE~\cite{BLE}, are integral to modern wireless networks but lack robust privacy safeguards. This exposes them to tracking, linkability, and identity exposure attacks, where adversaries monitor device presence, track movements, and link sessions, leading to profiling and privacy breaches~\cite{cassola2015authenticating, fawaz2016protecting, stute2021disrupting}. Despite the adoption of protocols like DNS-SD~\cite{DNS-SD}, mDNS~\cite{mDNS}, SSDP~\cite{SSDP}, and UPnP~\cite{UPnP}, their use of cleartext advertisements and lack of authentication lead to spoofing, MitM, and DoS attacks~\cite{bai2016staying, wang2019looking, venkatnarayan2020leveraging}.

Existing privacy-enhancing protocols such as PrivateDrop~\cite{heinrich2021privatedrop} and WTSB~\cite{wu2016privacy} still fall short, as they lack policy-controlled access and attribute hiding, leaving users vulnerable to tracking and impersonation~\cite{stute2019billion}. Similarly, CBN~\cite{cassola2015authenticating} provides anonymous client authentication but fails to protect service providers from spoofing.

These limitations underscore the need for privacy-preserving SD mechanisms aligned with global standards. Regulations such as RFC 7258~\cite{rfc7258}, ISO/IEC 29184~\cite{iso29184}, and GDPR~\cite{gdpr} mandate confidentiality and privacy-by-design; NIST SP 800-63B~\cite{nist80063b} and ETSI TS 103 465~\cite{etsi103465} advocate bilateral access control; and RFC 6973~\cite{rfc6973}, the NIST Privacy Framework~\cite{nistprivacy}, and OECD Guidelines~\cite{oecdprivacy} emphasize bilateral anonymity. Sender authentication is equally essential, mandated by NIST SP 800-53~\cite{nist80053}, ISO/IEC 29115~\cite{iso29115}, and ETSI EN 303 645~\cite{etsi303645}.

Recently, Yang et al.~\cite{yang2024prisrv} proposed PriSrv (NDSS'24), addressing major privacy and usability challenges in service discovery. Its dual-layer architecture enables only authorized clients to discover services, protecting sensitive information during interactions. Unlike protocols such as AirDrop and BLE that lack strong privacy guarantees, PriSrv supports bilateral policy control through anonymous credential-based matchmaking encryption (ACME), providing mutual authentication and defending against MitM attacks, tracking, and profiling.

% Long Version
% Recently, Yang et al.~\cite{yang2024prisrv} developed PriSrv in NDSS'24, offering a comprehensive solution to many of the privacy and usability challenges present in the existing SD protocols. By leveraging its dual-layer architecture, PriSrv ensures that services are only discoverable by authorized clients, mitigating the risks of exposing sensitive information during device interaction.
% In contrast to protocols like AirDrop and BLE, which lack robust privacy mechanisms, PriSrv enhances mutual authentication and shields private attributes of involved parties from unauthorized access. This mechanism allows PriSrv to counteract common vulnerabilities in SD protocols, such as MitM attacks, user tracking, and device profiling, by applying anonymous credential-based matchmaking encryption (ACME) that supports bilateral policy control for secure communication.

However, PriSrv inherits limitations from ACME. It reveals public attributes in the outer layer, potentially enabling tracking. Its binary attribute vector model (each attribute can only represent 1 or 0) restricts expressiveness and increases computation with large attribute sets. ACME’s small-universe design requires system rebuilds to add new attributes, hindering scalability. Additionally, large ciphertexts lead to high communication overhead, and pre-issued anonymous credentials introduce management complexity.

% Long Version
% However, despite its strengths, PriSrv has limitations due to its reliance on ACME. One major drawback is the leakage of public attributes in the outer layer, which can still enable tracking attacks.
% The use of binary attribute vectors (each attribute can only represent 1 or 0) significantly constrains the expressiveness of attributes, and increases the computational costs during encryption and decryption due to large vector size in case of numerous attributes. Additionally, ACME’s small-universe construction limits the total number of attributes in the system. Adding new attributes requires a full system rebuild, reducing flexibility for large-scale deployments. Furthermore, the large size of the service discovery broadcast ciphertext in PriSrv results in high transmission overhead and delays, particularly constraining its usage in low-bandwidth environments such as BLE and mDNS. The requirement for pre-issued anonymous credentials in PriSrv also adds complexity to credential management, potentially affecting system efficiency.

To overcome these issues, we propose a fast and expressive matchmaking encryption (FEME) scheme, which is also of \textit{independent interest} for advancing matchmaking encryption (ME).
% Long Version
% Aiming to address these issues, we propose a fast and expressive matchmaking encryption (FEME) scheme, which is also of \textit{independent interest} in advancing matchmaking encryption (ME) techniques. FEME ensures high efficiency and robust privacy. 
Unlike prior Identity-Based ME (IBME) schemes~\cite{ateniese2019match,chen2022identity} that support only equality policies, FEME enables expressive policies with arbitrary strings and solves an open problem posed in CRYPTO'19 and ASIACRYPT'22.
% Long Version
% Unlike earlier ME schemes, such as Ateniese et al.'s Identity-Based Matchmaking Encryption (IBME) instantiation \cite{ateniese2019match} and Chen et al.'s IBME scheme \cite{chen2022identity}, which are constrained to basic equality policies, FEME supports expressive policies with arbitrary attribute values. It \textit{solves the open problem} highlighted by Ateniese et al. in CRYPTO'19 and Chen et al. in ASIACRYPT'22, calling for the development of matchmaking encryption 
% that balances policy expressiveness with privacy.
Additionally, FEME is significantly faster in encryption and decryption than the existing ME scheme supporting expressive policy control~\cite{yang2024prisrv}. 

To further ensure robust authentication and privacy, FEME introduces a novel \textit{double re-randomization and binding technique}, which prevents encryption key extraction, thwarts ciphertext forgery and component-mixing attacks, and conceals sensitive attribute values. FEME achieves bilateral access control, bilateral anonymity, and sender authentication in the context of expressive policies with high efficiency. It adopts a \textit{partially hidden access structure}~\cite{lai2012expressive}, where only attribute names are exposed while attribute values remain concealed, enabling efficient policy matching without revealing sensitive information. Combined with a \textit{randomness splitting technique}~\cite{meng2024fease}, FEME offers a practical and privacy-preserving solution for expressive matchmaking encryption.

Building on the strengths of FEME, PriSrv+ overcomes the limitations of its predecessor, PriSrv, and introduces new capabilities. It enhances usability by eliminating the reliance on anonymous credentials and the associated overhead of credential issuance and revocation. PriSrv+ supports expressive bilateral policy control and flexible attribute representation, lifting the constraints of binary vectors and small-universe designs in PriSrv. It also significantly reduces communication overhead, shrinking broadcast sizes by up to 87.33$\%$, which boosts scalability and performance in low-bandwidth settings. Additionally, PriSrv+ improves privacy by concealing all attribute values during service discovery, providing a robust and efficient solution for privacy-preserving service discovery.

% Long Version
% Building on the strengths of FEME, PriSrv+ overcomes the shortcomings of its predecessor, PriSrv, while introducing new features. PriSrv+ enhances the usability of service discovery by eliminating its dependency on anonymous credentials and the management overhead associated with credential issuance and revocation. PriSrv+ achieves high expressiveness in bilateral policy control and attribute representation, breaking the constraints of binary attribute vectors and a small attribute universe, as seen in PriSrv. PriSrv+ significantly reduces communication overhead by decreasing the size of broadcast messages by up to 87.33$\%$, thereby improving its scalability and efficiency, which is crucial in low-bandwidth environments. Furthermore, it enhances privacy by ensuring that outer layer attributes in PriSrv remain protected during the service discovery process, offering a comprehensive solution for private service discovery.

The key contributions of PriSrv+ are outlined as follows.

$\bullet$ \textit{Fast and Expressive Matchmaking Encryption (FEME)}. At the core of PriSrv+, FEME is the first matchmaking encryption scheme capable of supporting expressive access control policies with an unbounded attribute universe, allowing any arbitrary string to be used as an attribute. 
% This solves an open problem identified by Ateniese et al. in CRYPTO'19 and by Chen et al. in ASIACRYPT'21. 
FEME offers up to 7.62$\times$ faster encryption and 6.23$\times$ faster decryption compared with ACME, making PriSrv+ suitable for wireless environments.

$\bullet$ \textit{Enhanced Protocol Scalability and Flexibility}. PriSrv+ significantly improves the scalability over PriSrv by supporting unrestricted attribute space. Attributes in PriSrv+ can be arbitrary strings, such as postal addresses, eliminating the restriction of rigid binary vectors used in PriSrv. This enhancement provides greater flexibility in service discovery and access control management, enabling the applicability of PriSrv+ across diverse real-world settings while maintaining low computation and communication overheads.

$\bullet$ \textit{Optimized Performance and Scalability}. By reducing ciphertext size and optimizing cryptographic operations, PriSrv+ significantly lowers packet transmission overhead, leading to up to 7.17$\times$ faster service broadcast and 3.32$\times$ faster anonymous mutual authentication compared to PriSrv. This positions PriSrv+ as a more efficient and scalable protocol, particularly suitable for bandwidth-limited and latency-sensitive networks.

$\bullet$ \textit{Interoperability with Existing Protocols}. PriSrv+ maintains compatibility with widely-used wireless protocols such as mDNS, BLE, EAP, AirDrop, and Wi-Fi, while addressing scalability issues in PriSrv. In comparison to PriSrv, for instance, PriSrv+ reduces the packet size in mDNS by 88.89$\%$, in BLE by 87.73$\%$, and in Wi-Fi by 86.64$\%$, which makes it more suitable for low-bandwidth environments.

$\bullet$ \textit{Versatile Implementation across Platforms}. PriSrv+ has been tested on a range of platforms, including desktops, laptops, mobile devices, and IoT systems like Raspberry Pi. Experimental results indicate that PriSrv+ reduces delays in both privacy-preserving service broadcast and mutual authentication, delivering immediate responses even in resource-constrained environments. 
% On average, PriSrv+ performs 7.17$\times$ faster for broadcast and 3.32$\times$ faster for anonymous mutual authentication than its predecessor PriSrv, demonstrating its practical efficiency in real-world scenarios. 
% Our source code is available at \cite{PriSrv+}.

$\bullet$ \textit{Formal Security and Privacy Guarantees}. Rigorous formal security proofs demonstrate that FEME satisfies confidentiality, anonymity, and authenticity. PriSrv+ is proven to be a secure service discovery protocol with bilateral anonymity, offering superior protection compared to other state-of-the-art protocols.

These contributions establish PriSrv+ as an efficient, secure, and scalable solution for wireless networks, offering robust security and privacy guarantees and adaptability to the evolving demands of modern communication systems.

% \subsection{Roadmap}

% The remainder of this paper is organized as follows. In Section \ref{sec:Relatedwork}, we analyze the related works of service discovery protocols and matchmaking encryption schemes.
% In Section \ref{sec:Preliminary}, we introduce the preliminaries highly related to our scheme and protocol. Section \ref{sec:FEME} presents the core scheme, FEME, including its technical roadmap, construction, and comparisons with existing ME schemes. Section \ref{sec:PriSrv+} constructs PriSrv+ for privacy and usability-enhanced wireless service discovery. We benchmark the performance of PriSrv+ in Section \ref{sec:Implementation} and conclude the paper in Section \ref{sec:Conclusion}. 
\section{Related Work} 
\label{sec:Relatedwork}

\subsection{Service Discovery Protocols}
\label{subsec:SrvDiscovery}

Service discovery (SD) protocols such as Wi-Fi~\cite{WiFi}, AirDrop~\cite{AirDrop}, and Bluetooth Low Energy (BLE)~\cite{BLE} facilitate the automatic detection and advertisement of services and devices in dynamic networks, streamlining device interactions. However, these protocols pose significant privacy risks, particularly for users wishing to safeguard sensitive or identifying information. Studies show that about 90$\%$ of users view the exposure of device names as a privacy threat~\cite{konings2013device}, enabling adversaries to infer personal data such as location, mobility, and user profiles~\cite{wu2016privacy, zhou2019discovering, stute2019billion, stute2021disrupting}. For example, device names in public Wi-Fi can allow Internet Service Providers (ISPs) to track users~\cite{cassola2015authenticating}, while attackers in IoT networks can analyze service data to reveal user routines~\cite{fawaz2016protecting}.

% Long Version
% Service discovery (SD) protocols such as Wi-Fi \cite{WiFi}, AirDrop \cite{AirDrop}, and Bluetooth Low Energy (BLE) \cite{BLE} are essential for facilitating the detection and advertisement of services and devices within dynamic network environments. These protocols enable devices to discover available services automatically, streamlining device interactions. However, SD protocols present considerable privacy risks, especially for users seeking to protect identifying or sensitive information during their interactions. Research indicates that approximately 90$\%$ of users consider the exposure of device names to be a privacy concern \cite{konings2013device}, as it allows adversaries to infer personal details, including location, mobility patterns, and user profiles \cite{wu2016privacy, zhou2019discovering, stute2019billion, stute2021disrupting}. For instance, in public Wi-Fi networks, device names may provide Internet Service Providers (ISPs) with tracking capabilities \cite{cassola2015authenticating}, and in IoT networks, attackers can extract information from smart devices to deduce users' routines \cite{fawaz2016protecting}.

Most existing SD protocols, including DNS-SD~\cite{DNS-SD}, mDNS~\cite{mDNS}, SSDP~\cite{SSDP}, and UPnP~\cite{UPnP}, lack strong privacy safeguards, leaving them vulnerable to man-in-the-middle (MitM), spoofing, and denial-of-service (DoS) attacks~\cite{bai2016staying, wang2019looking}. These risks are exacerbated by the use of cleartext broadcasts in Wi-Fi and BLE, which expose device identifiers and enable adversarial tracking and profiling~\cite{venkatnarayan2020leveraging}. Although protocols like 
CBN~\cite{cassola2015authenticating} support anonymous client authentication, they offer insufficient protection for service providers, who remain exposed to impersonation and MitM attacks.

% Long Version
% Most existing SD protocols, such as DNS-SD \cite{DNS-SD}, mDNS \cite{mDNS}, SSDP \cite{SSDP}, and UPnP \cite{UPnP}, lack strong privacy safeguards, making them susceptible to attacks like man-in-the-middle (MitM), spoofing, and denial-of-service (DoS) \cite{bai2016staying, wang2019looking}. These vulnerabilities are worsened by the widespread use of cleartext broadcasts in Wi-Fi and BLE, which exposes device identifiers, allowing adversaries to track and profile users \cite{venkatnarayan2020leveraging}. While some protocols, like the CBN scheme \cite{cassola2015authenticating}, offer a level of anonymous client authentication, they do not provide adequate protections for service providers, leaving them exposed to impersonation and MitM attacks.

Protocols such as AirDrop~\cite{AirDrop}, PrivateDrop~\cite{heinrich2021privatedrop}, and WTSB~\cite{wu2016privacy} introduce encryption and authentication to enhance service discovery privacy. While improving mutual authentication and anonymity, they still suffer from tracking, MitM, and DoS vulnerabilities due to incomplete privacy features, such as selective attribute disclosure and multi-show unlinkability~\cite{stute2019billion, bai2016staying}. For example, reliance on certificates in AirDrop and PrivateDrop may allow attackers to link sessions and track users~\cite{heinrich2021privatedrop}.

% Long Version
% Protocols like AirDrop \cite{AirDrop}, PrivateDrop \cite{heinrich2021privatedrop}, and WTSB \cite{wu2016privacy} have introduced encryption and authentication techniques to enhance privacy during service discovery. These protocols strengthen mutual authentication and anonymity, yet they still exhibit weaknesses to tracking, MitM, and DoS attacks due to incomplete privacy features, such as selective attribute disclosure and multi-show unlinkability \cite{stute2019billion, bai2016staying}. For instance, AirDrop and PrivateDrop’s reliance on client and server certificates can be exploited by attackers to link sessions and track users across interactions \cite{heinrich2021privatedrop}.

Yang et al. introduced PriSrv~\cite{yang2024prisrv}, a private SD protocol that allows service providers and clients to define fine-grained access control policies, enabling mutual authentication while concealing private information. PriSrv leverages Anonymous Credential-based Matchmaking Encryption (ACME) to support bilateral policy control, selective attribute disclosure, and multishow unlinkability. However, its large message size results in high transmission overhead and reception delays, limiting its effectiveness in low-bandwidth networks like BLE and congested Wi-Fi. Additionally, the exposure of public attributes may lead to tracing and profiling attacks.

% Long Version
% Yang et al. introduced PriSrv\cite{yang2024prisrv}, a private service discovery protocol that enables service providers and clients to specify fine-grained access control policies, ensuring mutual authentication while concealing private information. PriSrv incorporates Anonymous Credential-based Matchmaking Encryption (ACME) to achieve bilateral policy control, selective attribute disclosure, and multishow unlinkability. PriSrv’s primary limitation is its large message size, which leads to high transmission overhead and reception delays, particularly constraining its usage in slower networks like BLE and congested Wi-Fi channels. This scalability challenge affects the protocol’s efficiency and performance. Another limitation is its leakage of public attributes, which may result in tracing and profiling attacks.

Therefore, there is a critical need for private service discovery protocols that provide stronger privacy protections and enhanced usability, a gap that PriSrv+ is designed to fill.

\subsection{Matchmaking Encryption (ME)}
\label{subsec:ME-Review}

Matchmaking Encryption (ME) was introduced by Ateniese et al.~\cite{ateniese2019match} in CRYPTO'19 as a new encryption paradigm enabling both sender and receiver to specify policies that must be mutually satisfied for successful decryption. In ME, the sender with identity or attribute $\sigma$ defines a policy $\mathbb{R}$, and the receiver with $\rho$ defines $\mathbb{S}$; decryption succeeds only if $\sigma$ satisfies $\mathbb{S}$ and $\rho$ satisfies $\mathbb{R}$. Ateniese et al. also instantiated Identity-Based Matchmaking Encryption (IBME) in the random oracle model, where equality-based identities are used, and sender authentication is achieved via embedded encryption keys.

% Long Version
% Matchmaking encryption (ME) was first introduced by Ateniese et al. \cite{ateniese2019match} in CRYPTO'19 as a new encryption paradigm where both the sender and receiver can define policies that must be satisfied for the message to be decrypted. In ME \cite{ateniese2019match}, a sender, associated with an identity or attribute $\sigma$, specifies a target identity or policy $\mathbb{R}$ during encryption, while a receiver, associated with identity or attribute $\rho$, sets a target identity or policy $\mathbb{S}$. Successful decryption occurs only if both the sender's $\sigma$ satisfies the receiver's $\mathbb{S}$ and the receiver's $\rho$ satisfies the sender's $\mathbb{R}$. Ateniese et al. also instantiated Identity-Based Matchmaking Encryption (IBME) in the random oracle model, where sender and receiver identities are specified as equality-based policies, and a sender's encryption key is embedded in the ciphertext to authenticate the sender.

Francati et al.~\cite{Francati2021identity} extended IBME to the standard model using non-standard assumptions and NIZK proofs, while Chen et al.~\cite{chen2022identity} constructed IBME under standard assumptions. Despite providing data privacy and authenticity, these schemes are limited to equality-based policies and 1-to-1 data sharing.

% Long Version
% Later, Francati et al. \cite{Francati2021identity} in INDOCRYPT’21 developed an IBME under non-standard assumptions in the standard model using non-interactive zero-knowledge (NIZK) proofs.
% Chen et al. \cite{chen2022identity} in ASIACRYPT'22 further proposed an IBME scheme under standard assumptions in the standard model. 
% Although these IBME schemes ensure data privacy and authenticity, they are limited to equality-based policies and one-to-one data sharing. 

To support one-to-many data sharing in ME, Sun et al.~\cite{sun2023privacy} and Yang et al.~\cite{yang2023lightweight} in TIFS'23 proposed privacy-aware ME (PSME) and certificateless ME (CLME), respectively—extending IBME to multi-user settings via identity-based broadcast encryption. Wu et al.~\cite{wu2023fuzzy} introduced fuzzy IBME (FBME), enabling decryption when the overlap between sender and receiver attributes exceeds a threshold. However, FBME's threshold-based policies have limited expressiveness and incur high decryption costs.

% Long Version
% To enable one-to-many secure data sharing in ME, Sun et al. \cite{sun2023privacy} in TIFS'23 proposed a privacy-aware ME scheme (PSME), and Yang et al. \cite{yang2023lightweight} in TIFS'23 introduced a certificateless ME scheme (CLME), both of which extend IBME to multi-user scenarios using identity-based broadcast encryption. Additionally, Wu et al. \cite{wu2023fuzzy} in TIFS'23 presented a fuzzy IBME (FBME) that supports fuzzy bilateral access control, allowing decryption if the overlap between the sender's and receiver's attribute sets exceeds a threshold. However, FBME supports threshold-based policy matching with limited expressiveness and incurs high computation costs during decryption. 

Recently, Yang et al.~\cite{yang2024prisrv} in NDSS'24 developed ACME, an anonymous credential-based ME scheme with flexible bilateral policy control. Despite its utility, ACME suffers from large ciphertext size and a small-universe construction that requires binary attribute vectors—leading to large vector sizes and increased computation. In contrast, FEME supports monotonic Boolean policies with an unrestricted attribute universe, allowing arbitrary strings as attributes. It also improves performance, reducing ciphertext size by 87.33\% and achieving up to 7.62$\times$ faster encryption and 6.23$\times$ faster decryption.

% Long Version
% More recently, Yang et al. \cite{yang2024prisrv} in NDSS'24 developed an anonymous credential-based matchmaking encryption (ACME) scheme that supports flexible bilateral policy control. While ACME offers valuable functionality, it is limited by large ciphertext size and a small-universe construction. Since the attribute sets used in ACME are denoted by binary attribute vectors, ACME requires large vectors to represent attributes in a wide range, leading to increased computation costs in both encryption and decryption. In comparison, FEME not only matches ACME's ability to support monotonic Boolean formulas as policies but also extends expressivity to an unrestricted attribute space, enabling any arbitrary string to be used as an attribute. Moreover, FEME significantly improves efficiency, reducing ciphertext size by 87.33$\%$, and achieving encryption and decryption speeds of 7.62$\times$ and 6.23$\times$ faster, respectively.

\section{Preliminary} \label{sec:Preliminary}

We present notations, bilinear pairing, access structure, linear secret sharing scheme, and partially hidden access structure, for constructing FEME and PriSrv+.

\subsection{Notation and Bilinear Pairing}

Let integers $m$ and $n$ satisfy $m<n$, with $[m,n]$ representing the set $\{m,m+1,...,n\}$, and $[n]$ denoting the set $\{1,...,n\}$. For a prime $p$, define $\mathbb{Z}_p$ as the set $\{0,1,...,p-1\}$, where addition and multiplication are performed modulo $p$. The set $\mathbb{Z}_p^*$ excludes 0 from $\mathbb{Z}_p$. The security parameter is denoted by $\lambda$. We use bold lowercase letters for vectors and bold uppercase letters for matrices. A vector $\mathbf{v}$ denotes a column vector by default, and $\mathbf{v}_k$ represents its $k$-th element. For a matrix $\mathbf{M}$, $\mathbf{M}_i$ is the $i$-th row, and $\mathbf{M}_{i,j}$ denotes the element at position $(i,j)$.

The notation $s \stackrel{\$}{\leftarrow} S$ indicates that $s$ is uniformly sampled from set $S$. The notation $y\leftarrow\textsf{Algo}(x)$ refers to the output $y$ after running algorithm $\textsf{Algo}$ on input $x$. An algorithm is probabilistic polynomial time (PPT) if it runs in polynomial time with respect to the input length. We assume a master public key is an implicit input to all algorithms.
A bilinear group with Type-III pairings is defined as $\mathcal{BG}=(\mathbb{G}_1,\mathbb{G}_2,\mathbb{G}_T,e,p)$, where there is no efficiently computable isomorphism between $\mathbb{G}_1$ and $\mathbb{G}_2$. For any $g_1\in\mathbb{G}_1$ and $g_2\in\mathbb{G}_2$, the pairing $e(g_1,g_2)$ maps to $\mathbb{G}_T$. For $a,b\stackrel{\$}{\leftarrow}\mathbb{Z}_p^*$, one has $e(g_1^a,g_2^b)=e(g_1,g_2)^{ab}$. 
% The proposed FEME scheme utilizes bilinear maps, with security proven in the generic group model, rather than relying on specific hardness assumptions in pairing-based cryptography.

% \begin{definition}[Computational Bilinear Diffie-Hellman (CBDH) Assumption]
%     Let $g_1$, $g_2$ be the generators of groups $\mathbb{G}_1$, $\mathbb{G}_2$, respectively. The CBDH assumption holds if for all PPT adversary $\mathcal{A}$, the advantage function $\textsf{Adv}_{\mathcal{A}}^{CBDH}:=\textsf{Pr}[\mathcal{A}(g_1,g_1^a,g_1^b,g_1^c,g_2,g_2^a,g_2^b,g_2^c)=e(g_1,g_2)^{abc}]$ is negligible, where $a,b,c\stackrel{\$}{\leftarrow}\mathbb{Z}_p^*$. 
% \end{definition}

\subsection{Access Structure}

\begin{definition}[Access Structure\cite{beimel1996secure}] 
Let $\{P_1,\cdots,P_n\}$ be a set of parties. A collection $\mathbb{A}\subseteq 2^{\{P_1,\cdots,P_n\}}$ is monotone if $\forall B,C$: if $B\in\mathbb{A}$ and $B\subseteq C$, then $C\in\mathbb{A}$. An access structure (respectively, monotone access structure) is a collection (respectively, monotone collection) $\mathbb{A}$ of non-empty subsets of $\{P_1,\cdots,P_n\}$, i.e., $\mathbb{A}\subseteq 2^{\{P_1,\cdots,P_n\}}\backslash\{\emptyset\}$. The sets in $\mathbb{A}$ are called authorized sets, and the sets not in $\mathbb{A}$ are called unauthorized sets.
\end{definition}\vspace{-1.5mm}

An access structure is said to be monotone if, for any two sets $S$ and $T$ of attributes, $S\subseteq T$ and $S$ being authorized imply that $T$ is also authorized. It ensures that any user possessing a set of attributes that satisfies the access policy continues to have access if additional attributes are granted.

% Our construction is designed for monotone access structures. General access structures can be supported by representing the negation of an attribute as a distinct attribute.

\subsection{Linear Secret Sharing Scheme (LSSS)}

\begin{definition}[Linear Secret Sharing Scheme (LSSS)\cite{beimel1996secure}] 
A secret sharing scheme $\Pi$ over a set of parties $\mathcal{P}$ is called linear (over $\mathbb{Z}_p$) if (1) the shares of each party form a vector over $\mathbb{Z}_p$. (2) there exists a matrix $\mathbf{A}$ with $m$ rows and $n$ columns called the share-generating matrix for $\Pi$. For all $i=1,\cdots,m$, the $i$-th row of $\mathbf{A}$ is labeled by a party $\rho(i)$ ($\rho$ is a function from $\{1,\cdots,m\}$ to $\mathcal{P}$). When we consider the column vector $v=(s,r_2,\cdots,r_n)$, where $s\in\mathbb{Z}_p$ is the secret to be shared, and $r_2,\cdots,r_n\in\mathbb{Z}_p$ are randomly chosen, then $\mathbf{A}v$ is the vector of $m$ shares of the secret $s$ according to $\Pi$. The share $(\mathbf{A}v)_i$ belongs to party $\rho(i)$.
\end{definition}\vspace{-1.5mm}

LSSS possesses the linear reconstruction property~\cite{beimel1996secure}. Let $\Pi$ be an LSSS for the access structure $\mathbb{A}$, and $S\in\mathbb{A}$ be an authorized set with $I\subset\{1,\cdots,m\}$, where $I=\{i|\rho(i)\in S\}$. There exists a set of constants $\{\omega_i\in\mathbb{Z}_p\}_{i\in I}$ such that, given any valid shares $\{\lambda_i\}$ of a secret $s$ in $\Pi$, the relationship $\sum\nolimits_{i\in I}\omega_i\lambda_i=s$ holds. Let $A_i$ be the $i$-th row of $\mathbf{A}$, we similarly have $\sum\nolimits_{i\in I}\omega_iA_i=(1,0,\cdots,0)$. These constants $\{\omega_i\}$ are computable in time polynomial \cite{beimel1996secure} in the size of $\textbf{A}$. Notably, constants $\{\omega_i\}$ cannot be constructed for unauthorized sets.

\textbf{Boolean Formulas}. Boolean formulae are a common way
to model access control. LSSS is a more general class of functions and
include Boolean formulas. Using established methods \cite{beimel1996secure}, any monotone Boolean formula can be transformed into an LSSS format. Such a formula can be structured as an access tree, where an access tree with $m$ nodes yields an LSSS matrix of $m$ rows.

% \textbf{Monotone span programs (MSP)} (or linear secret sharing schemes) are a more general method for expressing access structures than Boolean formulas \cite{lewko2011unbounded}. An access structure is encoded by a policy $(\mathbf{M},\pi)$, where $\mathbf{M}$ is an $\ell \times n$ matrix over $\mathbb{Z}_p$ and $\pi:\{1,...,\ell\}\rightarrow\mathcal{U}$ is a mapping function. Any (monotone) Boolean formula $F$ can be converted into an MSP $(\mathbf{M},\pi)$, where each row of $\mathbf{M}$ corresponds to an input in $F$, and the number of columns equals the number of \textbf{AND} gates in $F$. Each entry in $\mathbf{M}$ is either 0, 1, or -1.

% Given an attribute set $\mathcal{S}=\{\Psi_i\}_{i\in[m]}\subseteq\mathcal{U}$ containing $m$ attributes, define $I=\{i|i\in\{1,...,\ell\},\pi(i)\in\mathcal{S}\}$ as the set of rows in $\mathbf{M}$ corresponding to attributes in $\mathcal{S}$. The policy $(\mathbf{M},\pi)$ is satisfied by $\mathcal{S}$ if there is a linear combination of rows in $I$ that results in $(1,0,...,0)$. That is, there exist constants $\gamma_i\in\mathbb{Z}_p$ for $i\in I$ such that $\sum_{i\in I}\gamma_i\M_i=(1,0,...,0)$. These constants can be computed in polynomial time relative to the size of $\mathbf{M}$.

\subsection{Partially Hidden Access Structure}
\label{subsec:partialhidden}

In a partially hidden access structure~\cite{lai2012expressive}, attributes are divided into attribute names and attribute values, where only attribute names are exposed, while attribute values remain hidden. For example, consider an access policy that requires ``Role: Admin \textbf{AND} Department: Research \textbf{OR} Level: Confidential" to access certain data, where ``Role", ``Department", and ``Level" are attribute names, and ``Admin", ``Research", and ``Confidential" are attribute values. In a partially hidden access structure, the policy is transformed to ``Role \textbf{AND} Department \textbf{OR} Level", revealing attribute names only. 
This contrasts with a traditional access structure, where attributes are exposed in the policy.

We first define the structures of an attribute set and an access policy. Let the attribute set be $\mathcal{S}=\{u_i\}_{i\in[\ell]}$ containing $\ell$ attributes, where each attribute belongs to a unique category. Each attribute is denoted as $u_i = \langle n_i, v_i\rangle$, with $n_i$ representing the attribute name and $v_i$ the attribute value. An access policy is defined as $\mathbb{A} = (\M, \pi, \mathcal{T})$, where $\M$ is an $m \times n$ access control matrix, $\M_i$ is the $i$-th row of $\M$, and $\pi$ is a mapping function that associates each row $\M_i$ with an attribute $\pi(i)$. The policy $\mathcal{T}$ is expressed as $(\Psi_{\pi(1)}, \cdots, \Psi_{\pi(m)})$, where each $\Psi_{\pi(i)} = \langle n_{\pi(i)}, v_{\pi(i)}\rangle$ consists of a name $n_{\pi(i)}$ and value $v_{\pi(i)}$.

In a partially hidden attribute set, attribute values $v_i$ are concealed, leaving only attribute names $n_i$ visible. The resulting attribute set is modified to $\mathcal{S}_{\text{partial}}=\{n_i\}_{i\in[\ell]}$. Similarly, in a partially hidden access policy, the attribute values $v_{\pi(i)}$ are removed from $\mathcal{T}$, exposing only the attribute names. The modified policy is represented as $\mathbb{A}_{\text{partial}} = (\M, \pi, \mathcal{T}_{\text{name}})$, where $\mathcal{T}_{\text{name}} = (n_{\pi(1)}, \cdots, n_{\pi(m)})$. It conceals attribute values to enhance privacy while using attribute names for efficient policy matching. We define partial satisfaction $\mathcal{S}_{\text{partial}}\models\mathbb{A}_{\text{partial}}$ if the attribute names in $\mathcal{S}_{\text{partial}}$ match those in $\mathbb{A}_{\text{partial}}$. Full satisfaction ($\mathcal{S}\models\mathbb{A}$) requires matching both names and values, whereas partial satisfaction only matches attribute names.

\section{Fast and Expressive Matchmaking Encryption (FEME)} 
\label{sec:FEME}

We construct FEME, a fast and expressive matchmaking encryption scheme, as the core component of PriSrv+. It is also of independent interest for advancing ME techniques.

\subsection{Technical Roadmap}

In an ME system, both sender and receiver, each possessing a set of attributes, define access policies that the other must meet to decrypt any message. FEME features privacy-preserving policy matching and user anonymity. We leverage Attribute-Based Encryption (ABE)~\cite{bethencourt2007ciphertext,chase2007multi} with expressive access policies to enable bilateral matching of the policies of both sender and receiver. ABE is available in two forms: ciphertext-policy ABE (CP-ABE) and key-policy ABE (KP-ABE), both essential for building FEME.

\begin{figure}[ht]
  \centering  \includegraphics[width=.98\linewidth]{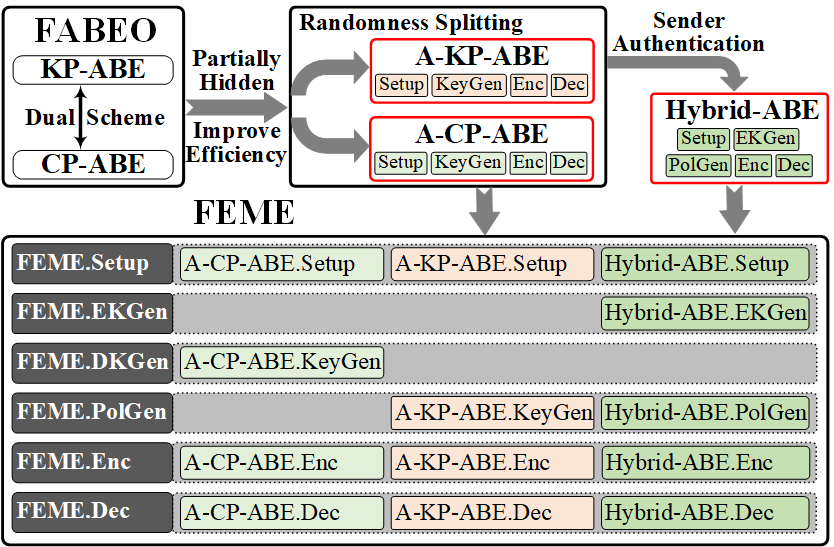}  
  \caption{Technical Roadmap of FEME} 
  \label{fig:FEME-Roadmap}
\end{figure}

The design of FEME, as shown in Fig.~\ref{fig:FEME-Roadmap}, follows a structured, multi-stage roadmap that enhances existing ABE schemes to address privacy and efficiency challenges. We build on FABEO~\cite{riepel2022fabeo}, a dual-form KP-ABE and CP-ABE scheme\footnote{Both schemes control access by matching attributes to policies, but they reverse the roles of the ciphertext and key in defining access control. They share the same design mechanism and common parameters.} that supports expressive policies without restrictions on policy type or attribute range. However, while FABEO excels in policy expressiveness, it lacks privacy-preserving policy matching or anonymity. FABEO’s CP-ABE exposes access policies with plaintext attribute values, and its KP-ABE reveals attribute values in attribute sets. Moreover, FABEO's decryption incurs high computation overhead due to pairing and exponentiation operations that scale with policy complexity.

FEME addresses these limitations in three distinct stages. Stage 1 enhances FABEO's KP-ABE and CP-ABE schemes, creating anonymous versions (A-KP-ABE in Fig.~\ref{fig:A-KP-ABE} and A-CP-ABE in Fig.~\ref{fig:A-CP-ABE}), that hide attribute values in attribute sets and access policies, greatly improving computational efficiency.
Stage 2 introduces Hybrid-ABE (Fig.~\ref{fig:Hybrid-ABE}), bridging the gap between ME with CP-ABE/KP-ABE and supporting bilateral policy-matching and \textit{sender authentication}. Stage 3 integrates A-KP-ABE, A-CP-ABE, and Hybrid-ABE schemes to create FEME, an ME that enhances both privacy and efficiency.

\subsection{Novelty of FEME}

\textbf{Distinct from Existing ABE-Based Solutions}.
FEME achieves bilateral access control, bilateral anonymity, and sender authentication—features not simultaneously supported by existing ABE-based schemes such as FABEO~\cite{riepel2022fabeo}, FEASE~\cite{meng2024fease}, and FABESA~\cite{meng2024fabesa}. These works optimize ABE efficiency or enable unilateral anonymity for searchable encryption but lack bilateral policy matching and sender authentication, both of which are critical for privacy-preserving service discovery.
% as required by standards like GDPR~\cite{gdpr}, ISO/IEC 29184~\cite{iso29184}, and RFC 7258~\cite{rfc7258}.

\textbf{Technical Challenges and Innovations}.
Designing an ME scheme for real-time service discovery presents significant challenges beyond traditional ABE systems. Existing ABE schemes, including combinations of CP-ABE and KP-ABE, cannot enforce \textit{sender authentication}—a critical requirement in ME. Without sender authentication, malicious entities can forge ciphertexts with fabricated attributes, enabling impersonation, spoofing, and injection attacks that threaten both security and availability. Prior techniques such as partial policy hiding and randomness splitting, used in works like FEASE and FABESA, fall short of defending against these advanced threats in bilateral settings.

To overcome these limitations, FEME introduces a novel \textit{double re-randomization and binding} technique that ensures both efficiency and robust sender authentication. The \textit{first-level re-randomization} randomizes encryption key components using a shared factor, preventing adversaries from extracting or reusing the sender’s encryption key. The \textit{second-level re-randomization} applies additional independent randomness to specific ciphertext components while enforcing a constraint across them. This enforces \textit{binding} between encryption key-derived components and other ciphertext elements, ensuring they cannot be mixed or tampered with to create forged messages. Moreover, the second-level re-randomization conceals sensitive attribute values to resist attribute guessing attacks. Together, these mechanisms provide strong protection against impersonation, ciphertext injection, and attribute guessing attacks—enabling secure, private, and authenticated service discovery in adversarial settings.

\subsection{Technical Details}

Following the above roadmap, we transform FABEO into a privacy-preserving and efficient ME scheme.

% As illustrated in Fig. \ref{fig:FEME-Roadmap}, the design of FEME proceeds in multiple stages. We begin with FABEO\cite{riepel2022fabeo}, a fast KP-ABE scheme and a fast CP-ABE scheme designed in a dual form\footnote{Both schemes control access by matching attributes to policies, but they reverse the roles of the ciphertext and key in defining access control. They share the same design mechanism and common parameters.}, which supports expressive policies without restrictions on policy type or attributes. However, the FABEO schemes do not inherently provide privacy-preserving policy matching or anonymity. In FABEO CP-ABE, the access policy is explicitly embedded in ciphertext, while in its KP-ABE, the attribute set is exposed. 
% Moreover, the decryption process in FABEO incurs significant computational overhead due to its pairing and exponentiation operations, which scale linearly with the complexity of the involved access policy and random vector.

% This stage-by-stage breakdown illustrates the technical roadmap for FEME (in Fig. \ref{fig:FEME-Roadmap}). The construction begins with transforming FABEO KP-ABE and CP-ABE into efficiency-enhanced anonymous variants A-CP-ABE and A-KP-ABE in Stage 1. In Stage 2, Hybrid-ABE bridges the gap between ME and CP-ABE/KP-ABE. Finally, Stage 3 culminates with the design of FEME, integrating the three schemes.

\textit{\textbf{Stage 1}}. We create A-CP-ABE and A-KP-ABE as anonymous variants of FABEO's CP-ABE and KP-ABE, respectively, using the following techniques. 

\textit{(1) Partially Hidden Access Structure}. To balance privacy and efficiency, we adopt a \textit{partially hidden access structure} ($\S$\ref{subsec:partialhidden}) that separates each attribute into a visible \textit{attribute name} and a concealed \textit{attribute value}, protecting sensitive information. Since attribute names—visible in A-KP-ABE or A-CP-ABE ciphertexts—are typically less sensitive, this design enables significant efficiency gains. Our A-CP-ABE and A-KP-ABE constructions minimize costly pairing and exponentiation operations, accelerating policy matching and improving suitability for resource-constrained environments.

\begin{figure}[h]
\scriptsize
\centering
    \begin{mybox}[colback = white, width = .98\linewidth]{A-KP-ABE: Anonymous KP-ABE}
      \underline{$\textsf{Setup}(1^{\lambda})\rightarrow(\textsf{mpk},\textsf{msk})$}.
      Generate $\mathcal{G}:=(p,\mathbb{G}_1,\mathbb{G}_2,\mathbb{G}_T,e,g_1,$ $g_2)$. Pick $\alpha,b_1,b_2\stackrel{\$}{\leftarrow}\mathbb{Z}_p^*$ and a hash function $H:\{0,1\}^*\rightarrow\mathbb{G}_1$. Compute $Z=e(g_1,g_2)^{\alpha}$, $\delta_1=g_2^{b_1}$, $\delta_2=g_2^{b_2}$.
      Output the master public key
      $\textsf{mpk}:=(\mathcal{G},H,Z,\delta_1,\delta_2)$ and master secret key $\textsf{msk}:=(\alpha,b_1,b_2)$.
      \vspace{1pt}
      
      \underline{$\textsf{KeyGen}(\textsf{msk},\mathbb{A}=(\textbf{A},\rho,\{\Psi_{\rho(i)}\}_{i\in[m]}))\rightarrow\textsf{SK}_{\mathbb{A}}$}.       
      Remind that $\{\Psi_{\rho(i)}\}_{i\in[m]}=\{\langle n_{\rho(i)},v_{\rho(i)}\rangle\}_{i\in[m]}$.
      Pick $r'\stackrel{\$}{\leftarrow}\mathbb{Z}_p^*$ and $\textbf{y}\stackrel{\$}{\leftarrow}\mathbb{Z}_p^{n-1}$. Compute      $\textsf{sk}_1=g_2^{r'}$, $\textsf{sk}_{2,i}=\big(g_1^{\textbf{A}_i(\alpha||\textbf{y})^{\top}}\cdot H(\Psi_{\rho(i)})^{r'}\big)^{\frac{1}{b_1}}$,      $\textsf{sk}_{3,i}=\big(g_1^{\textbf{A}_i(\alpha||\textbf{y})^{\top}}\cdot H(\Psi_{\rho(i)})^{r'}\big)^{\frac{1}{b_2}}$.      
      Output 
      $\textsf{SK}_{\mathbb{A}}:=((\textbf{A},\rho,\{n_{\rho(i)}\}_{i\in[m]}),\textsf{sk}_1,\{\textsf{sk}_{2,i},\textsf{sk}_{3,i}\}_{i\in[m]})$. 
      \vspace{2pt}

      \underline{$\textsf{Enc}(\mathcal{S}=\{u_i\}_{i\in[\ell]}=\{\langle n_i,v_i\rangle\}_{i\in[\ell]},\textsf{msg})\rightarrow\textsf{CT}_{\mathcal{S}}$}. Pick $s',s''\stackrel{\$}{\leftarrow}\mathbb{Z}_p^*$. Let $s=s'+s''$. Compute
      $\textsf{ct}_0=e(g_1,g_2)^{\alpha s}\cdot\textsf{msg}$, $\textsf{ct}_{1,i}=H(u_i)^{s}$, $\textsf{ct}_2=\delta_1^{s'}$, $\textsf{ct}_3=\delta_2^{s''}$.      
      Output 
      $\textsf{CT}_{\mathcal{S}}:=(\{n_i\}_{i\in[\ell]},\textsf{ct}_0,\{\textsf{ct}_{1,i}\}_{i\in[\ell]},\textsf{ct}_2,\textsf{ct}_3).$ 
      
      \underline{$\textsf{Dec}(\textsf{SK}_{\mathbb{A}},\textsf{CT}_{\mathcal{S}})\rightarrow \textsf{msg}/\bot$}. If there is any subset $I$ that matches the attribute names $\{n_i\}_{i\in[\ell]}$ in $\CT$ with $(\textbf{A},\rho,\{n_{\rho(i)}\}_{i\in[m]})$ in $\SK$, there exist constants $\{\omega_i\}_{i\in I}$ s.t. $\sum_{i\in I} \omega_i\textbf{A}_i=(1,0,\cdots,0)$. Output  $$\textsf{msg}=\frac{\textsf{ct}_0\cdot e(\prod_{i\in I}(\textsf{ct}_{1,\rho(i)})^{\omega_i},\textsf{sk}_1)}{e(\prod\nolimits_{i\in I}(\textsf{sk}_{2,i})^{\omega_i},\textsf{ct}_2)\cdot e(\prod\nolimits_{i\in I}(\textsf{sk}_{3,i})^{\omega_i},\textsf{ct}_3)}.$$
    \end{mybox}
\caption{A-KP-ABE Scheme} 
\label{fig:A-KP-ABE}
\end{figure}

\begin{figure}[h]
\scriptsize
\centering
    \begin{mybox}[colback = white, width = .98\linewidth]{A-CP-ABE: Anonymous CP-ABE}
      \underline{$\textsf{Setup}(1^{\lambda})\rightarrow(\textsf{mpk},\textsf{msk})$}.
      Generate $\mathcal{G}:=(p,\mathbb{G}_1,\mathbb{G}_2,\mathbb{G}_T,e,g_1,$ $g_2)$. Pick $\alpha\stackrel{\$}{\leftarrow}\mathbb{Z}_p^*$, $h\stackrel{\$}{\leftarrow}\mathbb{G}_1$ and a hash function $H:\{0,1\}^*\rightarrow\mathbb{G}_1$. Compute $Z=e(g_1,g_2)^{\alpha}$.
      Output the master public key
      $\textsf{mpk}:=(\mathcal{G},H,Z,h)$ and master secret key $\textsf{msk}:=\alpha$.
      
     \underline{$\KeyGen(\textsf{msk},\mathcal{S}=\{u_i\}_{i\in[\ell]}=\{\langle n_i,v_i\rangle\}_{i\in[\ell]})\rightarrow\textsf{SK}_{\mathcal{S}}$}. Pick $r\stackrel{\$}{\leftarrow}\mathbb{Z}_p^*$. Computes 
     $\textsf{sk}_1=g_1^{\alpha}h^r$, $\textsf{sk}_3=g_2^r$, $\textsf{sk}_{2,i}=H(u_i)^r\text{~for~}i\in[\ell]$.
      Output 
      $\textsf{SK}_{\mathcal{S}}:=(\{n_i\}_{i\in[\ell]},\textsf{sk}_1,\{\textsf{sk}_{2,i}\}_{i\in\ell},\textsf{sk}_3)$. 
      \vspace{2pt}

      \underline{$\textsf{Enc}(\mathbb{A}=(\textbf{M},\pi,\{\Psi_{\pi(i)}\}_{i\in[m]}),\textsf{msg})\rightarrow\textsf{CT}_{\mathbb{A}}$}.       
      Remind that $\{\Psi_{\pi(i)}\}_{i\in[m]}=\{\langle n_{\pi(i)},v_{\pi(i)}\rangle\}_{i\in[m]}$.   
      Pick $s_1,s'\stackrel{\$}{\leftarrow}\mathbb{Z}_p^*$ and vector $\textbf{v}\stackrel{\$}{\leftarrow}\mathbb{Z}_p^{n-1}$. Compute ciphertext
      $\textsf{ct}_0=e(g_1,g_2)^{\alpha s_1}\cdot\textsf{msg}$, $\textsf{ct}_1=g_2^{s_1}$, $\textsf{ct}_2=g_2^{s'}$, $\textsf{ct}_{3,i}=h^{\textbf{M}_i(s_1||\textbf{v})^{\top}}\cdot H(\Psi_{\pi(i)})^{s'}$. 
      
      Output
      $\textsf{CT}:=((\textbf{M},\pi,\{n_{\pi(i)}\}_{i\in[m]}),\textsf{ct}_0,\textsf{ct}_1,\textsf{ct}_2,\{\textsf{ct}_{3,i}\}_{i\in[m]}).$
      \vspace{2pt}
      
      \underline{$\textsf{Dec}(\textsf{SK}_{\mathcal{S}},\textsf{CT}_{\mathbb{A}})\rightarrow \textsf{msg}/\bot$}. If there is any subset $I$ that matches the attribute names $\{n_i\}_{i\in[\ell]}$ in $\SK$ with $(\textbf{M},\pi,\{n_{\pi(i)}\}_{i\in[m]})$ in $\CT$, there exist constants $\{\gamma_i\}_{i\in I}$ s.t. $\sum_{i\in I} \gamma_i\textbf{M}_i=(1,0,\cdots,0)$. Output  $$\textsf{msg}=\frac{\textsf{ct}_0\cdot e(\prod_{i\in I}(\textsf{ct}_{3,i})^{\gamma_i},\textsf{sk}_3)}{e(\textsf{sk}_1,\textsf{ct}_1)\cdot e(\prod\nolimits_{i\in I}(\textsf{sk}_{2,\pi(i)})^{\gamma_i},\textsf{ct}_2)}.$$
    \end{mybox}
\caption{A-CP-ABE Scheme} 
\label{fig:A-CP-ABE}
\end{figure}

\begin{figure}[h]
\scriptsize
\centering
    \begin{mybox}[colback = white, width = .98\linewidth]{Hybrid-ABE: Bridging CP-ABE and KP-ABE}
      \underline{$\textsf{Setup}(1^{\lambda})\rightarrow(\textsf{mpk},\textsf{msk})$}.
      Generate $\mathcal{G}:=(p,\mathbb{G}_1,\mathbb{G}_2,\mathbb{G}_T,e,g_1,$ $g_2)$. Pick $x,\mu,b_1,b_2\stackrel{\$}{\leftarrow}\mathbb{Z}_p^*$, $h\stackrel{\$}{\leftarrow}\mathbb{G}_1$ and a hash function $H:\{0,1\}^*\rightarrow\mathbb{G}_1$. Compute $Y=e(g_1,g_2)^{x\mu}$, $\delta_0=g_2^{\mu}$, $\delta_1=g_2^{b_1}$, $\delta_2=g_2^{b_2}$.
      Output the master public key
      $\textsf{mpk}:=(\mathcal{G},H,Y,h,\delta_0,\delta_1,\delta_2)$ and master secret key $\textsf{msk}:=(x,\mu,b_1,b_2)$.
      
     \underline{$\EKGen(\textsf{msk},\mathcal{S}_{\snd}=\{u_i\}_{i\in[\ell]}=\{\langle n_i,v_i\rangle\}_{i\in[\ell]})\rightarrow\textsf{EK}_{\mathcal{S}_{\snd}}$}. Pick $\tau\stackrel{\$}{\leftarrow}\mathbb{Z}_p^*$. Compute    $\textsf{ek}_1=g_1^xh^{\tau}$, $\textsf{ek}_3=\delta_1^{\tau}$, $\textsf{ek}_4=\delta_2^{\tau}$, $\textsf{ek}_{2,i}=H(u_i)^{\tau}$.
      Output 
      $\textsf{EK}_{\mathcal{S}_{\snd}}:=(\{n_i\}_{i\in[\ell]},\textsf{ek}_1,\{\textsf{ek}_{2,i}\}_{i\in\ell},\textsf{ek}_3,\textsf{ek}_4)$. 
      \vspace{2pt}

      \underline{$\PolGen(\mathbb{A}_{\rcv}=(\msk,\textbf{A},\rho,\{\Psi_{\rho(i)}\}_{i\in[m]}))\rightarrow\textsf{SK}_{\mathbb{A}_{\rcv}}$}.       
      % Remind that $\{\Psi_{\rho(i)}\}_{i\in[m]}=\{n_{\rho(i)},v_{\rho(i)}\}_{i\in[m]}$.   
      Pick $r'\stackrel{\$}{\leftarrow}\mathbb{Z}_p^*$ and $\textbf{y}\stackrel{\$}{\leftarrow}\mathbb{Z}_p^{n-1}$. Compute $\textsf{sk}_1=g_2^{r'}$, $\textsf{sk}_{2,i}=\big(h^{\textbf{M}_i(s_1||\textbf{y})^{\top}}\cdot H(\Psi_{\rho(i)})^{r'}\big)^{\frac{1}{b_1}}$, $\textsf{sk}_{3,i}=\big(h^{\textbf{M}_i(s_1||\textbf{y})^{\top}}\cdot H(\Psi_{\rho(i)})^{r'}\big)^{\frac{1}{b_2}}$.     
      Output
      $\textsf{SK}_{\mathbb{A}_{\rcv}}:=((\textbf{A},\rho,\{n_{\rho(i)}\}_{i\in[m]}),\textsf{sk}_1,\{\textsf{sk}_{2,i},\textsf{sk}_{3,i}\}_{i\in[m]}).$

      \underline{$\textsf{Enc}(\textsf{EK}_{\mathcal{S}_{\snd}},\textsf{msg})\rightarrow\textsf{CT}_{\snd}$}.    
      Pick $\tau',s',s''\stackrel{\$}{\leftarrow}\mathbb{Z}_p^*$. Let $s=s'+s''$. Compute       $\textsf{ct}_0=Y^{s}\cdot\textsf{msg}$, $\textsf{ct}_{1,i}=(\ek_{1,i}\cdot H(u_i)^{\tau'})^{s},$     
      $\ct_2=(\ek_2\cdot \delta_1^{\tau'})^{s'}$, $\ct_3=(\ek_2\cdot \delta_2^{\tau'})^{s''}$, $\ct_4=(\ek_4\cdot h^{\tau'})^{s}$.
      Output
      $\textsf{CT}_{\snd}:=(\{n_i\}_{i\in[\ell]},\textsf{ct}_0,\{\textsf{ct}_{1,i}\}_{i\in[\ell]},\textsf{ct}_2,\textsf{ct}_3,\textsf{ct}_4).$
      
      \underline{$\textsf{Dec}(\textsf{SK}_{\mathbb{A}_{\rcv}},\textsf{CT}_{\snd})\rightarrow \textsf{msg}/\bot$}. If there is any subset $I$ that matches $\{n_i\}_{i\in[\ell]}$ in $\CT_{\snd}$ with $(\A,\rho,\{n_{\rho(i)}\}_{i\in[m_2]})$ in $\SK_{\mathbb{A}_{\rcv}}$, there exist constants $\{\omega_i\}_{i\in I}$ s.t. $\sum_{i\in I} \omega_i\textbf{A}_i=(1,0,\cdots,0)$. Output  $$\textsf{msg}=\ct_0\cdot\frac{e(\prod_{i\in I}(\sk_{2,i})^{\omega_i},\ct_2)e(\prod_{i\in I}(\sk_{3,i})^{\omega_i},\ct_3)}{e(\ct_4,\delta_0)e(\prod\nolimits_{i\in I}(\ct_{1,\rho(i)})^{\omega_i},\sk_1)}.$$ 
    \end{mybox}
\caption{Hybrid-ABE Scheme} 
\label{fig:Hybrid-ABE}
\end{figure}

\textit{(2) Randomness Splitting Technique}. To address the vulnerability of attribute guessing attacks in the FABEO KP-ABE scheme, we implement a \textit{randomness splitting technique}. In the original FABEO KP-ABE scheme (see Fig. 1 in \cite{riepel2022fabeo}), the reuse of a single random value $s$ across ciphertext components $\ct_{1,u}=H(u)^s$ and $\ct_2=g_2^s$ makes it possible for an attacker to deduce an attribute $u$ by testing the equality $e(\ct_{1,u}, g_2)=e(H(u),\ct_2)$. To mitigate this risk, our A-KP-ABE scheme (see Fig. \ref{fig:A-KP-ABE}) splits the randomness $s$ into two independent values, $s'$ and $s''$, such that $s=s'+s''$. This adjustment modifies the ciphertext components as follows: $\ct_{1,i}=H(u_i)^{s}$, $\ct_2=\delta_1^{s'}$, and $\ct_3=\delta_2^{s''}$, where $\delta_1=g_2^{b_1}$ and $\delta_2=g_2^{b_2}$. To cancel the exponentiation $b_1$ and $b_2$, the decryption key includes components $\sk_{2,i}$ and $\sk_{3,i}$, which use exponentiation by $\frac{1}{b_1}$ and $\frac{1}{b_2}$. It ensures that the attribute set remains concealed, and the modified components $\ct_{1,i}$, $\ct_2$, and $\ct_3$ reveal no information for any attacker to infer attributes, thus effectively preventing attribute guessing attacks.

\begin{figure*}[bp]
\normalsize
% \small
\centering
\begin{mybox}[colback = white, width = \linewidth]{FEME: Fast and Expressive Matchmaking Encryption}
  \underline{1. $\textsf{Setup}(1^{\lambda})\rightarrow(\mpk,\msk)$. \grey{// System Setup}}
  
  This algorithm takes in the security parameter $1^{\lambda}$ and generates a bilinear pairing $\mathcal{G}:=(p,\mathbb{G}_1,\mathbb{G}_2,\mathbb{G}_T,e,g_1,g_2)$. The algorithm picks random numbers $\alpha,x,\mu,b_1,b_2\stackrel{\$}{\leftarrow}\mathbb{Z}_p^*$,  $h\stackrel{\$}{\leftarrow}\mathbb{G}_1$, hash functions $H:\{0,1\}^*\rightarrow\mathbb{G}_1$, $\hat{H}:\mathbb{G}_T\rightarrow\{0,1\}^{l_0}$, and a polynomial-time computable padding function $\phi:\{0,1\}^n\rightarrow\{0,1\}^{\ell_0}$. It computes  $Z=e(g_1,g_2)^{\alpha}, Y=e(g_1,g_2)^{x\mu}, \delta_0=g_2^{\mu}, \delta_1=g_2^{b_1}, \delta_2=g_2^{b_2}.$
  It outputs the master public key as
  $\textsf{mpk}:=(\mathcal{G},H,\hat{H},\phi,Z,Y,h,\delta_0,\delta_1,\delta_2)$, and the master secret key as $\msk:=(\alpha,x,\mu,b_1,b_2)$.
  \vspace{1mm}

  \underline{2. $\EKGen(\msk,\mathcal{S}_{\snd}=\{u_i\}_{i\in[\ell_1]}=\{\langle n_i,v_i\rangle\}_{i\in[\ell_1]})\rightarrow\EK_{\mathcal{S}_{\snd}}$. \grey{// Attribute Encryption Key Generation}}
  
  This algorithm generates the sender's attribute encryption key $\EK_{\mathcal{S}_{\snd}}$ for attributes $\mathcal{S}_{\snd}=\{u_i\}_{i\in[\ell_1]}=\{\langle n_i,v_i\rangle\}_{i\in[\ell_1]}$, where $n_i$ denotes the attribute name and $v_i$ the attribute value. It picks a random number $\tau\stackrel{\$}{\leftarrow}\mathbb{Z}_p^*$ and computes as follows:     $\ek_{1,i}= H(u_i)^{\tau}\text{~for~}i\in[\ell_1],~\ek_2=\delta_1^{\tau},~\ek_3=\delta_2^{\tau},~\ek_4=g_1^xh^{\tau}.$
  It outputs the sender attribute encryption key   $\EK_{\mathcal{S}_{\snd}}:=(\{n_i\}_{i\in[\ell_1]},\{\ek_{1,i}\}_{i\in[\ell_1]},\ek_2,\ek_3,\ek_4)$. 
  \vspace{1mm}

  \underline{3. $\DKGen(\msk,\mathcal{S}_{\rcv}=\{u_i\}_{i\in[\ell_2]}=\{\langle n_i,v_i\rangle\}_{i\in[\ell_2]})\rightarrow\DK_{\mathcal{S}_{\rcv}}$. \grey{// Attribute Decryption Key Generation}}

  To generate the receiver's attribute decryption key $\DK_{\mathcal{S}_{\rcv}}$ for attributes $\mathcal{S}_{\rcv}=\{u_i\}_{i\in[\ell_2]}=\{\langle n_i,v_i\rangle\}_{i\in[\ell_2]}$, this algorithm picks a random number $r\stackrel{\$}{\leftarrow}\mathbb{Z}_p^*$ and computes as follows:    $\dk_1=g_1^{\alpha}h^r,~\dk_{2,i}=H(u_i)^r, ~\dk_3=g_2^r.$
  It outputs the receiver attribute decryption key 
  $\DK_{\mathcal{S}_{\rcv}}:=(\{n_i\}_{i\in[\ell_2]},\dk_1,\{\dk_{2,i}\}_{i\in[\ell_2]},\dk_3)$. 
  \vspace{5pt}

  \underline{4. $\PolGen(\msk,\mathbb{A}_{\rcv}=(\A,\rho,\{\Psi_{\rho(i)}\}_{i\in[m_2]}))\rightarrow\SK_{\mathbb{A}_{\rcv}}$. \grey{// Policy Decryption Key Generation}}

  This receiver's policy decryption key generation algorithm generates the secret key $\sk_{\mathbb{A}_{\rcv}}$ with receiver's monotone span policy $\mathbb{A}_{\rcv}=(\A,\rho,\{\Psi_{\rho(i)}\}_{i\in[m_2]}))$, where $\A$ is an $m_2 \times n_2$ access control matrix, $\{\Psi_{\rho(i)}\}_{i\in[m_2]}=\{\langle n_{\rho(i)},v_{\rho(i)}\rangle\}_{i\in[m_2]}$, $n_{\rho(i)}$ denotes attribute name and $v_{\rho(i)}$ attribute value. It picks a random number $r'\stackrel{\$}{\leftarrow}\mathbb{Z}_p^*$, a random vector $\sfy\stackrel{\$}{\leftarrow}\mathbb{Z}_p^{n_2-1}$ and computes as follows:  
  $$\sk_1=g_2^{r'},~~~\sk_{2,i}=(g_1^{\A_i(\alpha||\sfy)^{\top}}\cdot H(\Psi_{\rho(i)})^{r'})^{\frac{1}{b_1}},~~~\sk_{3,i}=(g_1^{\A_i(\alpha||\sfy)^{\top}}\cdot H(\Psi_{\rho(i)})^{r'})^{\frac{1}{b_2}},$$
  $$\sk_{4,i}=(h^{\A_i(\mu||\sfy)^{\top}}\cdot H(\Psi_{\rho(i)})^{r'})^{\frac{1}{b_1}},~~~\sk_{5,i}=(h^{\A_i(\mu||\sfy)^{\top}}\cdot H(\Psi_{\rho(i)})^{r'})^{\frac{1}{b_2}},~\text{for~each~row}~i\in[m_2].$$  
  
  It outputs the receiver policy decryption key 
  $\SK_{\mathbb{A}_{\rcv}}:=((\A,\rho,\{n_{\rho(i)}\}_{i\in[m_2]}),\sk_1,\{\sk_{2,i},\sk_{3,i},\sk_{4,i},\sk_{5,i}\}_{i\in[m_2]})$. 
  \vspace{5pt}
  
  \underline{5. $\Enc(\EK_{\mathcal{S}_{\snd}},\mathbb{A}_{\snd}=(\M,\pi,\{\Psi_{\pi(i)}\}_{i\in[m_1]}),\msg)\rightarrow\CT_{\snd}$. \grey{// Encrypt}}

  This algorithm encrypts a message $\msg\in\{0,1\}^n$ with sender's monotone span policy $\mathbb{A}_{\snd}=(\M,\pi,\{\Psi_{\pi(i)}\}_{i\in[m_1]})$ and sender attribute encryption key   $\EK_{\mathcal{S}_{\snd}}$, where $\{\Psi_{\pi(i)}\}_{i\in[m_1]}=\{\langle n_{\pi(i)},v_{\pi(i)}\rangle\}_{i\in[m_1]}$ and matrix $\M\in\mathbb{Z}^{m_1\times n_1}$. It selects $s_1,s_2',s_2'',s_3',s_3'',\tau'\stackrel{\$}{\leftarrow}\mathbb{Z}_p^*$, a vector $\sfv\stackrel{\$}{\leftarrow}\mathbb{Z}_p^{n_1-1}$. Let $s_2=s_2'+s_2''$ and $s_3=s_3'+s_3''$. It computes as follows:
  $$V=Z^{s_1+s_2}\cdot Y^{s_3},~~~\ct_0=\phi(\msg)\oplus\hat{H}(V),~~~\ct_1=g_2^{s_1},~~~\ct_2=g_2^{s_3},$$  $$\ct_{3,i}=h^{\M_i(s_1||\sfv)^{\top}}\cdot H(\Psi_{\pi(i)})^{s_3}~\text{for~each~row}~i\in[m_1],~~~\ct_{4,1}=\delta_1^{s_2'},~~\ct_{4,2}=\delta_2^{s_2''},$$  $$~~~\ct_{5,i}=H(u_i)^{s_2},~~\ct_{6,i}=(\ek_{1,i}\cdot H(u_i)^{\tau'})^{s_3}\text{~for~}i\in[\ell_1],~~~\ct_7=(\ek_2\cdot \delta_1^{\tau'})^{s_3'},~~\ct_8=(\ek_3\cdot \delta_2^{\tau'})^{s_3''},~~\ct_9=(\ek_4\cdot h^{\tau'})^{s_3}.$$  
  It outputs the ciphertext 
  $$\CT_{\snd}:=((\M,\pi,\{n_{\pi(i)}\}_{i\in[m_1]}), \{n_i\}_{i\in[\ell_1]}, \ct_0,\ct_1,\ct_2,\{\ct_{3,i}\}_{i\in[m_1]},\ct_{4,1},\ct_{4,2},\{\ct_{5,i},\ct_{6,i}\}_{i\in[\ell_1]},\ct_7,\ct_8,\ct_9).$$ 
  
  \underline{6. $\Dec(\DK_{\mathcal{S}_{\rcv}},\SK_{\mathbb{A}_{\rcv}},\CT_{\snd})\rightarrow \msg/\bot$: \grey{// Decrypt}}
  
  This algorithm decrypts a given ciphertext $\CT_{\snd}$ using $\DK_{\mathcal{S}_{\rcv}}$ and $\SK_{\mathbb{A}_{\rcv}}$. If $\mathcal{S}_{\rcv}\models\mathbb{A}_{\snd}$ (denoting that $\mathcal{S}_{\rcv}$ satisfies $\mathbb{A}_{\snd}$), there exist constants $\{\gamma_i\}_{i\in I_1}$ s.t. $\sum_{i\in I_1} \gamma_i\M_i=(1,0,\cdots,0)$. If $\mathcal{S}_{\snd}\models\mathbb{A}_{\rcv}$ (denoting that $\mathcal{S}_{\snd}$ satisfies $\mathbb{A}_{\rcv}$), there exist constants $\{\omega_i\}_{i\in I_2}$ s.t. $\sum_{i\in I_2} \omega_i\A_i=(1,0,\cdots,0)$. This algorithm recovers $V$ by computing    
  $$\normalsize{V=\frac{e(\dk_1,\ct_1) e(\prod\nolimits_{i\in I_1}(\dk_{2,\pi(i)})^{\gamma_i},\ct_2)}{e(\prod_{i\in I_1}(\ct_{3,\pi(i)})^{\gamma_i},\dk_3)}\cdot \frac{e(\prod\nolimits_{i\in I_2}(\sk_{2,\rho(i)})^{\omega_i},\ct_{4,1}) e(\prod\nolimits_{i\in I_2}(\sk_{3,\rho(i)})^{\omega_i},\ct_{4,2})}{e(\prod_{i\in I_2}(\ct_{5,\rho(i)})^{\omega_i},\sk_1)}}$$  $$\cdot\frac{e(\ct_9,\delta_0)e(\prod\nolimits_{i\in I_2}(\ct_{6,\rho(i)})^{\omega_i},\sk_1)}{e(\prod_{i\in I_2}(\sk_{4,\rho(i)})^{\omega_i},\ct_7)e(\prod_{i\in I_2}(\sk_{5,\rho(i)})^{\omega_i},\ct_8)}.~~~~~~~~~~~~$$
  It computes $\phi(\msg)=\ct_0\oplus\hat{H}(V)$. If the padding is valid, this algorithm returns $\msg$. Otherwise, it returns $\bot$.
  % $$\scriptsize{\frac{\ct_0\cdot e(\prod_{i\in I_1}(\ct_{3,\pi(i)})^{\gamma_i},\dk_3)\cdot e(\prod_{i\in I_2}(\ct_{5,\rho(i)})^{\omega_i},\sk_1)\cdot e(\prod_{i\in I_2}(\sk_{4,\rho(i)})^{\omega_i},\ct_7)\cdot e(\prod_{i\in I_2}(\sk_{5,\rho(i)})^{\omega_i},\ct_8)}{e(\dk_1,\ct_1) e(\prod\nolimits_{i\in I_1}(\dk_{2,\pi(i)})^{\gamma_i},\ct_2)\cdot e(\prod\nolimits_{i\in I_2}(\sk_{2,\rho(i)})^{\omega_i},\ct_{4,1})\cdot e(\prod\nolimits_{i\in I_2}(\sk_{3,\rho(i)})^{\omega_i},\ct_{4,2})\cdot e(\ct_9,\delta_0)\cdot e(\prod\nolimits_{i\in I_2}(\ct_{6,\rho(i)})^{\omega_i},\sk_1)}.}$$ 
\end{mybox}
\caption{FEME: Fast and Expressive Matchmaking Encryption Scheme} 
\label{fig:FEME}
\end{figure*}

\textit{(3) Scalability and Efficiency Enhancement}. 
We take several steps to make FEME highly scalable and efficient.

% Long version
First, we construct large-universe A-KP-ABE and A-CP-ABE schemes that allow FEME to handle an extensive and potentially unbounded set of attributes dynamically. This approach reduces the need for pre-defining and managing fixed attribute sets, enhancing scalability and minimizing the overhead associated with system updates and attribute expansions. By replacing the term $H(|\mathcal{U}|+1)$ in FABEO with an element $h\stackrel{\$}{\leftarrow}\mathbb{G}_1$ in the master public key $\mpk$, A-KP-ABE and A-CP-ABE become independent of the size of the attribute universe $|\mathcal{U}|$. This modification eliminates dependence on a fixed set of attributes, thereby improving both scalability and efficiency, and allowing for a more flexible and adaptable system.

Second, we address the main efficiency bottleneck in FABEO CP-ABE, which stems from its ciphertext components. Specifically, its ciphertext components $\ct_{2,j}=g_2^{s'[j]}$ and $\ct_{3,i}=H(|\mathcal{U}|+1)^{\textbf{M}_i(s_1||\textbf{v})^{\top}}\cdot H(\Psi_{\pi(i)})^{s'[\zeta(i)]}$ involve a random vector $\vec{s'}\stackrel{\$}{\leftarrow}\mathbb{Z}_p^{\tau}$, where $\tau$ represents vector size and $\zeta(i):=|\{z|\pi(z)=\pi(i),z\leq i\}|$. This setup leads to decryption involving $\tau$ pairing operations and approximately $\tau I$ exponentiations, which becomes computationally expensive, where $I$ represents the number of attributes required to satisfy an access policy. To mitigate this, we propose modifications to the ciphertext components, simplifying them to $\textsf{ct}_2=g_2^{s'}$ and $\textsf{ct}_{3,i}=h^{\textbf{M}i(s_1||\textbf{v})^{\top}}\cdot H(\Psi_{\pi(i)})^{s'}$ in our developed A-CP-ABE (Fig. \ref{fig:A-CP-ABE}), where $s'\stackrel{\$}{\leftarrow}\mathbb{Z}_p$. This significantly reduces the decryption workload to a single pairing operation and $I$ exponentiations, resulting in a more efficient decryption.

Third, recognizing the dual structure between FABEO KP-ABE and FABEO CP-ABE, we apply the above optimization technique to A-KP-ABE (Fig. \ref{fig:A-KP-ABE}). We replace its secret key components $\textsf{sk}_{1,j}=g_2^{r'[j]}$, $\textsf{sk}_{2,i}=g_1^{\textbf{A}_i(\alpha||\textbf{y})^{\top}}\cdot H(\Psi_{\rho(i)})^{r'[\eta(i)]}$, where $\vec{r'}\stackrel{\$}{\leftarrow}\mathbb{Z}_p^{\tau}$, with $\textsf{sk}_1=g_2^{r'}$ and $\textsf{sk}_{2,i}=g_1^{\textbf{A}_i(\alpha||\textbf{y})^{\top}}\cdot H(\Psi_{\rho(i)})^{r'}$, using $r'\stackrel{\$}{\leftarrow}\mathbb{Z}_p$ in our developed A-KP-ABE. Here, $\eta(i)$ is defined as $|\{z|\rho(z)=\rho(i),z\leq i\}|$. Then, we utilize the \textit{randomness splitting technique} to split $\textsf{sk}_{2,i}$ into two terms $\textsf{sk}_{2,i}=\big(g_1^{\textbf{A}_i(\alpha||\textbf{y})^{\top}}\cdot H(\Psi_{\rho(i)})^{r'}\big)^{\frac{1}{b_1}}$,      
$\textsf{sk}_{3,i}=\big(g_1^{\textbf{A}_i(\alpha||\textbf{y})^{\top}}\cdot H(\Psi_{\rho(i)})^{r'}\big)^{\frac{1}{b_2}}$. These optimizations effectively reduce the computational overhead in both A-CP-ABE and A-KP-ABE, making them more suitable for real-world applications, particularly those in resource-constrained environments.

\renewcommand{\arraystretch}{1.05}
\begin{table*}[thbp]\centering 
	\footnotesize
	\setlength{\tabcolsep}{2.3mm}
	\begin{tabular}{|l|c|c|c|c|c|c|c|c|c|c|c|}
	\hline 
            \multirow{3}{*}{\textbf{Scheme}}  & \multicolumn{3}{c|}{\textbf{Expressiveness}} &  \multicolumn{3}{c|}{\textbf{Security and Privacy}} & \multicolumn{2}{c|}{\textbf{Usability}}\\
        \cline{2-12}
        & Monotonic & Arbitrary & Large  		 
		&  Data & Data &  Attribute & No Pre-registration & No Additional\\
		  &  Policy &  Attribute & Universe 	 
		 & Privacy & Authenticity & Privacy & Pairing & Component \\
		\hline
		IBME\cite{ateniese2019match} (Crypto'19) & $\times$ & $\times$ & $\surd$  & $\surd$ & $\surd$ & $\surd$ & $\times$ & $\surd$\\ 
        \hline
        IBME\cite{Francati2021identity} (IndoCrypt'21) & $\times$ & $\times$ & $\surd$  & $\surd$ & $\times/\surd$ & $\surd$ & $\times$ & $\surd$\\ 
        \hline
        IBME\cite{chen2022identity} (AsiaCrypt'21) & $\times$ & $\times$ & $\surd$ & $\surd$ & $\surd$ & $\surd$ & $\times$ & $\surd$\\ 
        \hline
        % IBME\cite{boyen2023identity} (ESORICS'23) & $\times$ & $\times$ & $\surd$ & $\surd$ & $\surd$ & $\surd$ & $\surd$ & $\times$ & $\surd$\\ 
        % \hline
        FBME\cite{wu2023fuzzy} (TIFS'23) & $\times$ & $\times$ & $\times$ & $\surd$ & $\surd$ & $\surd$ & $\times$ & $\surd$\\ 
        \hline
        PSME\cite{sun2023privacy} (TIFS'23) & $\times$ & $\times$ & $\times$ & $\surd$ & $\surd$ & $\surd$ & $\times$ & $\surd$\\ 
        \hline
        CLME\cite{yang2023lightweight} (TIFS'23) & $\times$ & $\times$ & $\surd$ & $\surd$ & $\surd$ & $\surd$ & $\times$ & $\surd$\\ 
        \hline
	ACME\cite{yang2024prisrv} (NDSS'24)  & $\surd$ & $\times$ & $\times$ &$\surd$  & $\surd$ & $\times$ & $\surd$ & $\times$\\
        \hline
        FEME  & $\surd$ & $\surd$& $\surd$  & $\surd$ & $\surd$ & $\surd$ & $\surd$ & $\surd$\\
		\hline
	\end{tabular}
	\caption{Comparison of Matchmaking Encryption (ME) Schemes}		
	\label{Tab:CompareME}	
\end{table*}

\textit{\textbf{Stage 2}}. To address the gap between ME and A-CP-ABE/A-KP-ABE, we develop a Hybrid-ABE scheme (Fig. \ref{fig:Hybrid-ABE}). This gap stems from the \textit{sender authentication} requirement in ME, which ensures that only authorized senders (with a valid encryption key $\EK$ tied to their attributes) can generate legitimate ciphertexts. However, FABEO CP-ABE and KP-ABE do not inherently support sender authentication, as senders only use the master public key $\mpk$ and an attribute set (in KP-ABE) or access policy (in CP-ABE) to derive ciphertext.

To close this gap, Hybrid-ABE integrates the frameworks of A-CP-ABE and A-KP-ABE. Hybrid-ABE comprises the following algorithms: $\Setup$, $\EKGen$, $\PolGen$, $\Enc$ and $\Dec$. The $\EKGen$ algorithm generates the sender's attribute encryption key $\EK_{\mathcal{S}_{\snd}}$, drawing from A-CP-ABE's $\KeyGen$ algorithm and incorporating the \textit{randomness splitting technique} from A-KP-ABE's $\Enc$ algorithm.
$\PolGen$ produces the receiver's policy decryption key $\SK_{\mathbb{A}_{\rcv}}$, utilizing the exponentiation tricks from A-KP-ABE's $\KeyGen$.

During encryption, the sender's encryption key $\EK_{\mathcal{S}_{\snd}}$ is used to generate ciphertext $\CT_{\snd}$, which wraps message $\msg$ in ciphertext component $\ct_0=Y^s\cdot\msg$. The sender's encryption key is re-randomized into ciphertext components 
$\textsf{ct}_{1,i}$,
$\textsf{ct}_2$, $\textsf{ct}_3$, $\textsf{ct}_4$, using a nonce $\tau'\stackrel{\$}{\leftarrow}\mathbb{Z}_p^*$ and two split randomness values $s'$ and $s''$ (where $s=s'+s''$). The decryption algorithm unifies the processes of A-CP-ABE and A-KP-ABE, requiring only 4 pairings and $3I$ exponentiations (where $I$ represents the number of attributes required to satisfy an access policy), ensuring high efficiency.

\textit{\textbf{Stage 3}}. We construct FEME as shown in Fig. \ref{fig:FEME} based on A-CP-ABE, A-KP-ABE, and Hybrid-ABE developed in Stages 1 and 2. FEME consists of the following algorithms: $\Setup$, $\EKGen$, $\DKGen$, $\PolGen$, $\Enc$ and $\Dec$. The $\Setup$ algorithm initializes the master public key and the master secret key. $\EKGen$ generates the sender's attribute encryption key following the process outlined in Hybrid-ABE. $\DKGen$ produces the receiver's attribute decryption key, as $\KeyGen$ does from A-CP-ABE. Meanwhile, $\PolGen$ generates the receiver's policy decryption key following the $\KeyGen$ method from A-KP-ABE for producing its components $(\sk_1,\{\sk_{2,i},\sk_{3,i}\}_{i\in[m_2]})$. Additionally, $\PolGen$ adopts the Hybrid-ABE framework for producing its components $(\{\sk_{4,i},\sk_{5,i}\}_{i\in[m_2]})$, where $m_2$ denotes the number of rows in a receiver's access matrix $\textbf{A}$.

During encryption, a message $\msg$ is encapsulated in ciphertext component $\ct_0=\phi(\msg)\oplus\hat{H}(V)$ with $V=Z^{s_1+s_2}\cdot Y^{s_3}$, which combines elements from the $\ct_0$ components of all three schemes A-CP-ABE, A-KP-ABE, and Hybrid-ABE, where $\hat{H}$ is a hash function, and $\phi$ is a polynomial-time computable and efficiently invertible padding function to realize authenticated encryption. The encryption process generates $(\ct_1,\ct_2,\{\ct_{3,i}\}_{i\in[m_2]})$ based on $\Enc$ in A-CP-ABE, $(\ct_{4,1},\ct_{4,2},\{\ct_{5,i}\}_{i\in[m_1]})$ from A-KP-ABE, and $(\{\ct_{6,i}\}_{i\in[\ell_1]},\ct_7,\ct_8,\ct_9)$ based on Hybrid-ABE. Here, $m_1$ denotes the number of rows in sender's access matrix $\textbf{M}$, and $\ell_1$ the number of attributes in sender's attribute set.

FEME $\Dec$ incorporates the decryption processes of A-CP-ABE, A-KP-ABE, and Hybrid-ABE, with its three decryption fractions corresponding to $\Dec$ in each scheme.

\subsection{FEME Construction}

Following the above technical details, we describe the construction of FEME in Fig. \ref{fig:FEME}. The security models for FEME, detailed in Appendix \ref{app:SecModel}, outline its confidentiality, anonymity, and authenticity properties.
% The correctness proof of FEME is shown in Appendix \ref{app:FEMECorrectness}.

In FEME, a key generation center (KGC) runs $\Setup$ to generate the master public and secret keys, incorporating an efficiently computable and invertible padding function $\phi$ that enables integrity checks, thereby ensuring authenticated message encryption and robustness against unauthorized modifications~\cite{ateniese2019match}. To enable secure communication, the KGC executes $\EKGen$ to create the sender's attribute encryption key, and $\DKGen$/$\PolGen$ to generate the receiver's attribute decryption key and policy decryption key.
During $\Enc$, the sender specifies an access policy that the receiver must meet to access the message. FEME ensures that decryption is only possible if the sender's and receiver's attributes match their respective policies, guaranteeing \textit{sender authenticity} by certifying sender attributes through the attribute encryption key to prevent forged ciphertexts.

% The $\Setup$ phase in FEME includes a padding function $\phi$ that ensures authenticated message encryption. It should be efficiently computable and invertible, allowing for polynomial-time verification of correct padding\cite{ateniese2019match}. This guarantees that all encrypted messages can be checked for integrity, preventing unauthorized modifications and enhancing the overall security of the encryption process.

A core innovation of FEME is its \textit{double re-randomization and binding technique}, which ensures secure and efficient sender authentication. \textit{First re-randomization} applies a shared random value $\tau^{\prime}$ to encryption key components $(\textsf{ek}_{1,i},\textsf{ek}_2,\textsf{ek}_3,\textsf{ek}_4)$, generating ciphertext components $(\textsf{ct}_{6,i},\textsf{ct}_7,\textsf{ct}_8,\textsf{ct}_9)$, preventing adversaries from extracting valid encryption keys. \textit{Second re-randomization} utilizes random values $s_3,s_3^{\prime},s_3^{\prime\prime}$, ensuring attackers cannot generate new valid ciphertexts via mimicry. The same $s_3$ \textit{binds} encryption key-derived components with ciphertext elements $(V,\textsf{ct}_2,\textsf{ct}_{3,i})$, preventing attackers from mixing components from different ciphertexts. To conceal sensitive attribute values and prevent attribute guessing attacks, $(s_3^{\prime},s_3^{\prime\prime})$ are applied to further re-randomize ciphertext components $(\textsf{ct}_7,\textsf{ct}_8)$.

FEME's partially hidden access structure enhances decryption efficiency by revealing only attribute names while concealing values. This allows the receiver to pre-filter unmatched ciphertexts without computation. In $\Dec$, the receiver checks whether the sender’s attribute names satisfy its policy and vice versa. If either check fails, decryption aborts with output $\bot$. Otherwise, full decryption proceeds to verify attribute values and padding integrity. The message is returned only if both policies are satisfied and the padding is valid; otherwise, $\bot$ ensures ciphertext integrity.

\begin{theorem}
    \label{theo:FEME-privacy}
    FEME satisfies confidentiality under the Generic Group Model (GGM) by modeling the hash function $H$ as a random oracle.
\end{theorem}

\begin{theorem}
    \label{theo:FEME-anon}
    FEME satisfies anonymity under GGM by modeling the hash function $H$ as a random oracle.
\end{theorem}

\begin{theorem}
    \label{theo:FEME-authn}
    FEME satisfies authenticity under GGM by modeling the hash function $H$ as a random oracle.
\end{theorem}
\vspace{-0.5mm}

The proofs of Theorems \ref{theo:FEME-privacy}-\ref{theo:FEME-authn} are deferred to Appendix \ref{app:proof-FEME}.

% Full Version
% The proofs of Theorems \ref{theo:FEME-privacy}-\ref{theo:FEME-authn} are deferred to Appendices \ref{app:proof-FEME-privacy}-\ref{app:proof-FEME-authn}, respectively.

% \textit{Note on Potential Applications}. FEME is suited for domains requiring secure, policy-based data sharing. It enables bilateral access control in secure communication systems, permitting only authorized senders and receivers to exchange information based on matching attributes or policies, crucial for privacy-preserving service discovery in wireless networks. FEME also supports cross-organizational data sharing, ensuring confidentiality and policy compliance. Additionally, FEME aids decentralized platforms like blockchain systems by enabling fine-grained, policy-controlled data encryption and access, enhancing transaction privacy and security.
% This versatility makes FEME a critical tool for robust data protection and efficient attribute-based data dissemination.

\subsection{Comparative Advantages of FEME}

Table \ref{Tab:CompareME} compares the existing ME schemes in terms of expressiveness, security and privacy, and usability.

\textbf{Expressiveness.} Both FEME and ACME~\cite{yang2024prisrv} support monotonic Boolean formula-based access structures, unlike other schemes limited to identity-based matching. However, ACME's small-universe design restricts it to a fixed attribute set and relies on binary vectors, leading to longer vectors and higher computational costs for expressive policies. FEME supports an unrestricted attribute universe, allowing any arbitrary string as an attribute.

% Long Version
% In terms of expressiveness, FEME and ACME \cite{yang2024prisrv} support monotonic Boolean formula-based access structures, whereas other schemes focus on identity-based matching (either single, multiple, or fuzzy identities). However, ACME’s small-universe design constrains the system to a fixed set of attributes, necessitating a complete rebuild when expanding the attribute set. Furthermore, ACME uses binary attribute vectors (where each element is either 1 or 0) to denote attribute sets, necessitating larger vectors to represent a wide range of attributes, which in turn significantly increases the computational cost during encryption and decryption. In contrast, FEME supports an unrestricted attribute universe, enabling any arbitrary string to function as an attribute.

\textbf{Security and Privacy.} All schemes ensure data confidentiality. Francati et al.~\cite{Francati2021identity} provide an IBME without authenticity and another using NIZK for authenticity. Except for ACME, all schemes preserve attribute privacy or identity privacy. ACME reveals outer-layer public attributes due to its dual-layer design.

% Long Version
% In terms of security and privacy, all ME schemes ensure data privacy. Regarding data authenticity, Francati et al. \cite{Francati2021identity} proposed one IBME scheme without authenticity and another IBME scheme with authenticity using NIZK. Regarding attribute privacy, ACME exposes the attributes in the outer layer as it employs a dual-layer matching mechanism. Except for ACME, all other ME schemes in comparison preserve attribute/identity privacy. 

\textbf{Usability.} IBME~\cite{ateniese2019match,Francati2021identity,chen2022identity}, FBME~\cite{wu2023fuzzy}, PSME~\cite{sun2023privacy}, and CLME~\cite{yang2023lightweight} require \textit{pre-registration pairing}, where the sender must know the receiver’s identity or attributes beforehand. This tight coupling restricts flexible and real-time service discovery, as these schemes rely on identity-based or broadcast encryption. ACME~\cite{yang2024prisrv} avoids such pairing but incurs extra overhead due to its dependence on anonymous credentials for sender authentication. In contrast, FEME \textit{eliminates} both pre-registration pairing and external credential management by integrating sender authentication directly into ciphertexts, offering higher usability and scalability.

% Long Version
% In terms of usability, IBME \cite{ateniese2019match,Francati2021identity,chen2022identity}, FBME \cite{wu2023fuzzy}, PSME \cite{sun2023privacy}, and CLME \cite{yang2023lightweight} require pre-registration pairing, in that the sender must know the receiver's identity before encryption. ACME \cite{yang2024prisrv} relies on anonymous credentials as an additional component for sender authentication, adding overhead due to credential issuance, management, and revocation. In contrast, FEME avoids pre-registration pairing and additional components for sender authentication, achieving higher usability among the ME schemes.

In summary, FEME stands out as the only ME scheme that achieves expressive bilateral access control, robust security and privacy, and advantageous usability.

\section{PriSrv+ Protocol} 
\label{sec:PriSrv+}

PriSrv+ is a private service discovery protocol that leverages FEME to enable privacy-preserving service broadcasts and mutual authentication between a service provider and a client. Building on its predecessor PriSrv\cite{yang2024prisrv}, PriSrv+ replaces the core ACME scheme with FEME to eliminate the need for issuing, managing, and revoking credentials.

The overall workflow of PriSrv+ is as follows. (1) During the system setup phase, a KGC generates the attribute encryption key, attribute decryption key, and policy decryption key for the service provider ($S$) and the client ($C$) according to the FEME scheme. Both parties require a complete set of encryption and decryption keys as they act as both sender and receiver during interactions. (2) During the broadcast phase, the service provider announces an encrypted broadcast message using FEME. This message includes a service policy (in the partially hidden structure), service details, the provider's Diffie-Hellman (DH) public key, and a MAC key for authentication. Consider a private journalist network operated by an NGO as an example, where the provider's policy is ``(Journalist Type: Investigative \textbf{AND} Focus Area: Government Corruption \textbf{AND} Journalist Affiliation: Independent Media) \textbf{OR} (Role: Whistleblower \textbf{AND} Level: High Threat)," which is transformed to its partially hidden form ``(Journalist Type \textbf{AND} Focus Area \textbf{AND} Journalist Affiliation) \textbf{OR} (Role \textbf{AND} Level)" in the broadcast.

(3) In the service discovery phase, the client checks if its attribute names match the service policy and if the provider's attribute names (included in the broadcast message) meet its own policy. For the same example, the client may set its connection policy as ``(Network Type: Investigative \textbf{AND} Affiliation: NGO-Backed) \textbf{OR} (Jurisdiction: EU \textbf{AND} Support: Protection Available)".
If both checks pass, the client executes the FEME decrypt algorithm to verify whether the hidden attribute values of both parties satisfy those in each other's policies. If it succeeds, the client generates a response, and sends a FEME encrypted reply with its policy in its partially hidden form (which is ``(Network Type \textbf{AND} Affiliation) \textbf{OR} (Jurisdiction \textbf{AND} Support)" in the above example) and an authentication tag, including its DH public key and MAC key, back to the provider. (4) The service provider decrypts and verifies the client's response, then sends a confirmation message with an authentication tag to the client. (5) Both parties independently compute a shared session key using their respective DH secret keys, ensuring mutual authentication and maintaining privacy for both the client and the provider.

\subsection{Security and Threat Models of PriSrv+}

\textbf{Security Model}. The security model of PriSrv+ follows that of PriSrv (Appendix C in full version)\cite{yang2024prisrv}, which defines service discovery security and bilateral anonymity (i.e., both anonymity of service provider and anonymity of client). 
\textit{Service discovery security} ensures privacy-preserving service advertisement and anonymous mutual authentication with bilateral policy control, protecting sessions from adversarial exposure. It maintains confidentiality and authentication, allowing only authorized clients and service providers to establish secure communications.
The \textit{bilateral anonymity} property implies that neither the PPT service provider nor the PPT client can learn anything about the other participant's attribute values unless they satisfy each other's access policies.

\textbf{Threat Model}. Similar to PriSrv, PriSrv+ assumes a fully trusted Key Generation Center (KGC) for key distribution, which does not participate in service discovery. Service providers and clients are considered \textit{untrustworthy} and may attempt to gain unauthorized information or disrupt the protocol. Malicious providers may impersonate the other providers, track clients, or inject forged ciphertexts. Malicious clients may impersonate users or launch excessive requests to overwhelm providers. The mitigation of excessive requests is further discussed in Section~\ref{subsec:PriSrv+Discussions}.

Following the Canetti-Krawczyk model for authenticated key exchange (AKE)~\cite{canetti2001analysis,canetti2002security} and the service discovery model in~\cite{wu2016privacy}, PriSrv+ considers a strong adversary capable of controlling public communications—eavesdropping, injecting, modifying, replaying, or interleaving messages across sessions. The adversary, which may be external, a rogue service provider, or a compromised client, can launch attacks including spoofing, impersonation, MitM, and DoS. Their goals include breaking authenticated key exchange and exposing sensitive information for tracking and inference.

\subsection{PriSrv+ Construction}

\begin{figure*}[htbp]
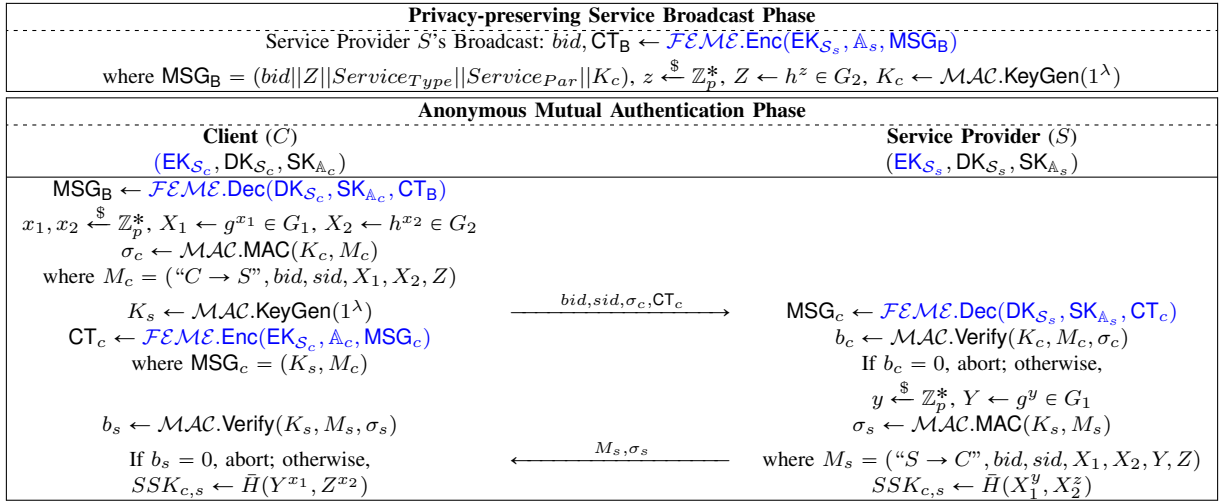

	\centering
	\footnotesize
	\renewcommand{\arraystretch}{1.09}
	\begin{tabular}{|ccc|}
		\hline
		\multicolumn{3}{|c|}{\textbf{Privacy-preserving Service Broadcast Phase}}\\
		\cdashline{1-3}[1.5pt/2pt]
		\multicolumn{3}{|c|}{Service Provider $S$'s Broadcast: $bid, \CT_\textsf{B}\leftarrow\blue{\FEME.\Enc(\EK_{\mathcal{S}_{s}},\mathbb{A}_{s},\MSG_\textsf{B})}$}\\
		\multicolumn{3}{|c|}{where $\MSG_\textsf{B}=(bid||Z||Service_{Type}||Service_{Par}||K_c)$, $z\xleftarrow{\$} \mathbb{Z}_p^*$, $Z\leftarrow h^z\in G_2$, $K_c\leftarrow\mathcal{MAC}.\KeyGen(1^{\lambda})$}\\
		\hline\hline
		\multicolumn{3}{|c|}{\textbf{Anonymous Mutual Authentication Phase}}\\
		\cdashline{1-3}[1.5pt/2pt]
		\textbf{Client} ($C$)  && \textbf{Service Provider}  ($S$)\\	
		$\blue{(\EK_{\mathcal{S}_{c}}},\DK_{\mathcal{S}_{c}},\SK_{\mathbb{A}_{c}})$ && $(\blue{\EK_{\mathcal{S}_{s}}},\DK_{\mathcal{S}_{s}},\SK_{\mathbb{A}_{s}})$\\
		\hline
		$\MSG_\textsf{B}\leftarrow \blue{\FEME.\Dec(\DK_{\mathcal{S}_c},\SK_{\mathbb{A}_c},\CT_\textsf{B})}$&&\\
		$x_1,x_2\xleftarrow{\$} \mathbb{Z}_p^*$, $X_1\leftarrow g^{x_1}\in G_1$, $X_2\leftarrow h^{x_2}\in G_2$  &&\\		
		$\sigma_c\leftarrow\mathcal{MAC}.\MAC(K_c,M_c)$&&\\
		
		where $M_c=(``C\rightarrow S",bid,sid,X_1,X_2,Z)$&&\\	
		
		$K_s\leftarrow\mathcal{MAC}.\KeyGen(1^{\lambda})$&$\xrightarrow{~~~~~~
			bid,sid,\sigma_c,\CT_c~~~~~}
		$		
		&$\MSG_c\leftarrow \blue{\FEME.\Dec(\DK_{\mathcal{S}_s},\SK_{\mathbb{A}_s},\CT_c)}$\\	
		
		$\CT_{c}\leftarrow\blue{\FEME.\Enc(\EK_{\mathcal{S}_{c}},\mathbb{A}_{c},\MSG_c)}$&&$b_c\leftarrow\mathcal{MAC}.\Verify(K_c,M_c,\sigma_c)$\\	
		
		where $\MSG_c=(K_s,M_c)$&& If $b_c=0$, abort; otherwise,\\

		%&&$K_s\leftarrow\mathcal{MAC}.\KeyGen(1^{\lambda})$\\

		&&$y\xleftarrow{\$} \mathbb{Z}_p^*$, $Y\leftarrow g^y\in G_1$\\
		
		$b_s\leftarrow\mathcal{MAC}.\Verify(K_s,M_s,\sigma_s)$	
		&&  $\sigma_s\leftarrow\mathcal{MAC}.\MAC(K_s,M_s)$\\

		If $b_s=0$, abort; otherwise,
		&$\xleftarrow{~~~~~~~~~~
			M_s,\sigma_s~~~~~~~~~~}$
		&where $M_s=(``S\rightarrow C",bid,sid,X_1,X_2,Y,Z)$\\
		
		$SSK_{c,s}\leftarrow \bar{H}(Y^{x_1},Z^{x_2})$
		&&$SSK_{c,s}\leftarrow \bar{H}(X_1^{y},X_2^{z})$\\
		\hline
	\end{tabular}	
	\caption{PriSrv+ Protocol} 
	\label{Fig:PriSrv+} 
\end{figure*}

Fig.~\ref{Fig:PriSrv+} presents the PriSrv+ protocol, comprising a privacy-preserving service broadcast phase and an anonymous mutual authentication phase. Key differences from PriSrv (highlighted in blue) include: (1) PriSrv+ uses the sender’s encryption key $\EK$ for authentication, replacing PriSrv’s anonymous credentials; and (2) PriSrv+ adopts FEME for bilateral policy control, whereas PriSrv employs ACME.

% Long Version
% Fig. \ref{Fig:PriSrv+} illustrates the PriSrv+ protocol, which operates through two phases: a privacy-preserving service broadcast phase and an anonymous mutual authentication phase. The differences between PriSrv+ and PriSrv, highlighted in blue, are as follows: (1) PriSrv+ employs the sender's encryption key, $\EK$, for sender authentication, whereas PriSrv uses anonymous credentials for this purpose. (2) PriSrv+ utilizes FEME to achieve bilateral policy control, while PriSrv relies on ACME for this functionality.

A unique broadcast identifier ($bid$) is assigned to each broadcast cycle, and a session identifier ($sid$) is assigned to each session. The broadcast cycle has a lifetime (e.g., 30 seconds), with the timestamp included in $bid$. Clients verify the timestamp upon decryption to ensure message freshness. Let $\mathcal{FEME}=(\Setup,\EKGen,\DKGen,\PolGen,\Enc,\Dec)$ be a FEME scheme, $\mathcal{MAC}=(\Setup, \KeyGen, \MAC, \Verify)$ be a message authentication code (MAC) scheme, and  $\bar{H}: \{0,1\}^* \rightarrow \mathcal{K}$ be a hash function where $\mathcal{K}$ represents the secret session key space. 

\textit{Service Broadcast Phase}: To initiate a broadcast, $S$ selects an access policy $\mathbb{A}_s$ that $C$ should satisfy. $S$ chooses an ephemeral Diffie-Hellman (DH) exponent $z\xleftarrow{\$} \mathbb{Z}_p^*$ and computes $Z = h^z$. It generates a MAC key $K_c \leftarrow \mathcal{MAC}.\textsf{KeyGen}(1^{\lambda})$. The broadcast message $\MSG_B = (bid || Z || Service_{Type} || Service_{Par} || K_c)$ includes the broadcast identifier, service type, parameters, and the MAC key. $S$ encrypts it into a ciphertext $\textsf{CT}_B$ using $\mathcal{FEME}.\textsf{Enc}$, and broadcasts $bid$ and $\textsf{CT}_B$ publicly.

\textit{Anonymous Mutual Authentication Phase}: This phase establishes a session key ($SSK_{c,s}$) between $C$ and $S$.

(1) Client Response: $C$ checks whether its attribute name set $\{n_i\}_{i\in[\ell_2]}$ satisfies $S$'s policy $(\M,\pi,\{n_{\pi(i)}\}_{i\in[m_1]})$, and whether $S$'s attribute name set $\{n_i\}_{i\in[\ell_1]}$ satisfies $C$'s policy $(\A,\rho,\{n_{\rho(i)}\}_{i\in[m_2]})$. If either test fails, $C$ discards the broadcast ciphertext without decryption. Otherwise, $C$ decrypts $\textsf{CT}_B$ using its decryption keys $(\textsf{DK}_{\mathcal{S}_c}, \textsf{DK}_{\mathbb{A}_c})$. If the decryption succeeds, $C$ generates DH values $X_1 = g^{x_1}$ and $X_2 = h^{x_2}$, where $x_1, x_2 \xleftarrow{\$}\mathbb{Z}_p^*$. $C$ then computes an authentication tag $\sigma_c$ for the message $M_c = (``C\rightarrow S", bid, sid, X_1, X_2, Z)$ using $K_c$ from the broadcast message $\MSG_B$. $C$ defines a policy $\mathbb{A}_c$ for $S$, encrypts its message $\MSG_c=(K_s,M_c)$ to ciphertext $\CT_c$ using $\mathcal{FEME}.\textsf{Enc}$, and sends $(bid,sid,\sigma_c,\textsf{CT}{c})$ to $S$.

(2) Service Provider Response: $S$ decrypts $C$'s ciphertext, verifies $\sigma_c$, and generates its own DH value $Y \leftarrow g^y$ using a random exponent $y \xleftarrow{\$} \mathbb{Z}_p^*$. It creates a message $M_s = (``S\rightarrow C", bid, sid, X_1, X_2, Y, Z)$ and a tag $\sigma_s$ using the MAC key $K_s$ from $MSG_c$. $S$ then computes a session key $SSK{c,s} \leftarrow \bar{H}(X_1^y, X_2^z)$ and sends $(M_s, \sigma_s)$ to $C$.

(3) Client Finalization: Upon receiving $(M_s, \sigma_s)$, $C$ verifies $\sigma_s$. If valid, $C$ computes a session key $SSK_{c,s} \leftarrow \bar{H}(Y^{x_1}, Z^{x_2})$ using its secret DH exponents $(x_1, x_2)$. As $X_1^y = Y^{x_1} = g^{x_1y}$ and $X_2^z = Z^{x_2} = h^{x_2z}$, both $C$ and $S$ derive the same session key $SSK_{c,s}$.

\begin{theorem}
\label{Theo:PriSrv+Security}
Suppose that the DDH assumption holds, $\mathcal{FEME}$ is secure, $\mathcal{MAC}$ is unforgeable, and $H$ is a random oracle, then PriSrv+ is a secure service discovery protocol and satisfies bilateral anonymity.
\end{theorem}
\vspace{-1mm}

\textit{Proof Sketch}. The security proof of PriSrv+ in Theorem \ref{Theo:PriSrv+Security} parallels that of PriSrv, as the main distinction between the protocols lies in replacing ACME (in PriSrv) with FEME (in PriSrv+).
The service discovery security is proved based on the confidentiality and authenticity of FEME in Theorems \ref{theo:FEME-privacy} and \ref{theo:FEME-authn}. The bilateral anonymity is proved based on FEME's anonymity in Theorem \ref{theo:FEME-anon}.
% Thus, the proof for PriSrv extends to PriSrv+, confirming it as a secure service discovery protocol and ensuring bilateral anonymity. 
% The complete proof of Theorem  \ref{Theo:PriSrv+Security} is shown in the full version of this paper \cite{PriSrv+}.
\vspace{-2mm}

\renewcommand{\arraystretch}{1.05}
\begin{table*}[thbp]
    \centering 
    \footnotesize
    \setlength{\tabcolsep}{1.35mm}
    \begin{tabular}{|c|c|c|c|c|c|c|c|c|c|c|c|c|c|c|c|c|c|c|}		
		\hline
		\multirow{3}{*}{\textbf{Protocol}} & \multicolumn{3}{c|}{\textbf{Expressiveness}} & \multicolumn{4}{c|}{\textbf{Security and Privacy}}
        & \multicolumn{5}{c|}{\textbf{Usability}}  \\
		\cline{2-13}		
		& Monotonic & Arbitrary  & Large & Privacy & Mutual  & Bilateral & Pub. Attri. & No Pre-reg. & \multicolumn{2}{c|}{No 3rd-party} & \multicolumn{2}{c|}{No In-advance} \\
		& Policy & Attribute & Universe & Broadcast & Authn. & Anon. & Hidden & Pairing & \multicolumn{2}{c|}{Dependence} & \multicolumn{2}{c|}{ID Issuance}  \\
		\hline
	PriSrv\cite{yang2024prisrv} & $\surd$ & $\times$ & $\times$ & $\surd$ & $\surd$ & $\surd$ & $\times$ & $\surd$ & \multicolumn{2}{c|}{$\surd$} & \multicolumn{2}{c|}{$\times$} \\
    \hline
    PriSrv+ & $\surd$ & $\surd$ & $\surd$ & $\surd$ & $\surd$ & $\surd$ & $\surd$ & $\surd$ & \multicolumn{2}{c|}{$\surd$} & \multicolumn{2}{c|}{$\surd$} \\
    \hline        
    \end{tabular}
    \caption{Comparison of Private Service Discovery Protocols with Bilateral Policy Control}
    \label{tab:comparePriSrv}
\end{table*}

\subsection{Comparative Advantages of PriSrv+}

In Table I of \cite{yang2024prisrv}, PriSrv is shown to be the only protocol among 10 SD protocols—including 7 standard ones (DNS-SD\cite{DNS-SD}, mDNS\cite{mDNS}, SSDP\cite{SSDP}, UPnP\cite{UPnP}, Wi-Fi\cite{WiFi}, BLE\cite{BLE}, AirDrop\cite{AirDrop}) and 3 privacy-preserving ones (PrivateDrop\cite{heinrich2021privatedrop}, CBN\cite{cassola2015authenticating}, WTSB\cite{wu2016privacy})—that achieves both high privacy and usability. Rather than revisiting prior comparisons, we directly compare PriSrv+ with PriSrv in terms of expressiveness, security and privacy, and usability. Table \ref{tab:comparePriSrv} summarizes the results.

% Long Version
% In Table I of \cite{yang2024prisrv}, PriSrv distinguishes itself as the only privacy-preserving service discovery protocol that enhances privacy while maintaining high usability, based on a thorough comparison with 7 standard service discovery protocols (DNS-SD\cite{DNS-SD}, mDNS\cite{mDNS}, SSDP\cite{SSDP}, UPnP\cite{UPnP}, Wi-Fi\cite{WiFi}, BLE\cite{BLE}, AirDrop\cite{AirDrop}) and 3 privacy-preserving SD protocols (PrivateDrop\cite{heinrich2021privatedrop}, CBN\cite{cassola2015authenticating}, WTSB\cite{wu2016privacy}). Instead of revisiting these earlier comparisons, we focus on comparing PriSrv+ directly with PriSrv in terms of expressiveness, security and privacy, and 
% usability. Table \ref{tab:comparePriSrv} summarizes comparison results.

In terms of \textit{expressiveness}, PriSrv+ supports LSSS-defined policies with a large-universe construction, allowing any arbitrary string (e.g., a postal address) as an attribute. In contrast, PriSrv’s small-universe design restricts it to a fixed attribute set, requiring a full system rebuild to add new attributes. Moreover, PriSrv represents attributes as binary vectors (1 or 0), resulting in longer vectors and increased encryption/decryption overhead for broader attribute sets.

% Long Version
% In terms of \textit{expressiveness}, PriSrv+ supports expressive policies defined by LSSS and uses a large-universe construction to accommodates an unrestricted attribute set, allowing any arbitrary string, such as a postal address, to serve as an attribute. In contrast, ACME’s small-universe design restricts PriSrv to a fixed set of attributes, necessitating a full system rebuild to add new attributes, making it less adaptable for large, flexible deployments. PriSrv also employs binary attribute vectors, where each attribute is represented as 1 or 0, leading to longer vectors and increased computational overheads for broader attribute ranges during encryption and decryption.

In terms of \textit{security and privacy}, both protocols provide similar privacy protections. However, PriSrv discloses public attributes and access policies in its outer layer due to its dual-layer design. PriSrv+ mitigates this exposure by using a partially hidden access structure that separates attribute names from values, revealing only names during matching while keeping values confidential.

% Long Version
% In terms of \textit{security and privacy}, both protocols provide similar privacy protections. However, PriSrv discloses a set of public attributes and a public access policy in its outer layer during the matching process due to its dual-layer architecture. PriSrv+ addresses this attribute exposure by employing a partially hidden access structure, which separates each attribute into an attribute name and an attribute value. Only the attribute names are revealed during matching, while the attribute values remain confidential.

In terms of \textit{usability}, neither protocol requires pre-registered pairings or third-party support for service discovery. However, PriSrv depends on anonymous credential issuance, adding management overhead. PriSrv+ eliminates this requirement by directly leveraging the sender’s encryption key, simplifying deployment without compromising functionality.

% Long Version
% In terms of \textit{usability}, neither protocol requires pre-registered pairings or third-party dependencies for service discovery. However, PriSrv requires the issuance of anonymous credentials in advance, which adds overhead for credential management. PriSrv+ eliminates this requirement by using the sender's encryption key directly, thereby avoiding the need for credential management.

This comparison shows that PriSrv+ achieves greater expressiveness, improved security and privacy, and usability.

\subsection{Discussions}
\label{subsec:PriSrv+Discussions}

\textbf{Justifying Attribute Concealment in PriSrv+}. Although PriSrv could be modified to conceal all attributes and policies within its inner layer, this would substantially increase overhead. Its dual-layer architecture relies on public attributes in the outer layer to efficiently filter mismatched services without decryption. Concealing all attributes would force clients to decrypt every broadcast, incurring significant latency. PriSrv+ addresses this limitation using FEME, which reveals only attribute names while concealing values. This preserves fast filtering and eliminates tracking risks, leading to up to 3.55× faster encryption and 6.23× faster decryption compared to PriSrv, even with its outer-layer filtering mechanism.

\textbf{Need for Expressive Attributes}.
Modern service discovery in IoT, smart cities, enterprise systems, and 5G networks requires expressive attributes—e.g., building names, project codes, or service categories—not supported by fixed identifier sets. PriSrv+ accommodates these scenarios through support for arbitrary string attributes, enabling fine-grained policy enforcement beyond traditional service discovery.

\textbf{Mitigating DoS Attacks}. To resist Denial-of-Service (DoS) attacks from excessive attribute submissions, PriSrv+ can incorporate a \textit{flexible Proof-of-Work (PoW) mechanism}. The service provider embeds a difficulty level in the broadcast identifier ($bid$), and clients must compute a session ID ($sid$) such that $H(bid,sid,\sigma_c,\textsf{CT}_c)$ satisfies this level (e.g., number of leading zeros). This imposes minimal overhead on legitimate users and service providers while significantly raising the attack cost. Dynamic difficulty adjustment and complementary measures—such as attribute limits and request throttling—further enhance resilience without harming usability.

% \blue{More detailed discussions are available in Appendix \ref{app:Discussion}.}
\section{Implementation and Evaluation} 
\label{sec:Implementation}

For a fair comparison, we used the same benchmarks, crypto library, elliptic curves, and parameters as PriSrv\cite{yang2024prisrv} when implementing PriSrv+ and PriSrv in C/C++.
We utilized (i) MIRACL library\footnote{MIRACL: multiprecision integer and rational arithmetic c/c++ library.
https://github.com/miracl/MIRACL.} for FEME and PriSrv+ implementations, (ii) three elliptic curves, including MNT159 (80-bit security), MNT201 (90-bit security), and BN256 (100-bit security) to evaluate different security levels\footnote{Pairing-Friendly Curves. https://datatracker.ietf.org/doc/draft-irtf-cfr
g-pairing-friendly-curves.}, 
(iii) SHA-256 for the the hash function, and (iv) MAC-GGM\cite{chase2014algebraic} for the MAC scheme. Our source code is available at\cite{PriSrv+}.
% Our source code, approximately \red{11,290} lines, is available at\cite{PriSrv+}.

 \subsection{Evaluation of FEME and ACME}
 \label{subsec:Eval-FEME}

Table \ref{Tab:PerFEME} presents the computation cost (comp.) and communication (comm.) cost of FEME and ACME for various algorithms on
a desktop (Intel Core i9-7920X, 12 cores, 16GB RAM). The experimental setup ensures that FEME and ACME operate under equivalent conditions as used in~\cite{yang2024prisrv}, with the same attribute numbers and policies. For FEME, parameters are $\ell_1=\ell_2=4$, $m_1=m_2=2$, and $I_1=I_2=4$, representing the sender’s and receiver’s attribute numbers, access matrix rows, and the number of attributes satisfying the access policy. ACME parameters are set similarly for comparability: $n=10$, $\ell_1=\ell_2=4$, $k=2$, and $I_1=I_2=4$, 
where $n$ is the system's total attribute number (which is fixed in the small-universe setting), and $k$ is the number of access matrix rows. 
This alignment allows a fair comparison of both schemes.

% The experiment setting ensures that FEME and ACME share the same attribute number, access matrix and policy: the parameters
% for FEME are $l_1=l_2=4$, $m_1=m_2=2$, and $I_1=I_2=4$, where $n$ is the total number of attributes in the system, $k$ is the matrix parameter and $\hat{m}$ is the circuit size. For ACME, the parameters are $n = 10$, $k = 2$, $\hat{m} = 9$ and $I_1=I_2=4$, where $l_1/l_2$ denotes the sender/receiver's attribute number, $m_1/m_2$ is the row number of the sender/receiver's access matrix, and $I_1/I_2$ is the sender/receiver's attribute number to satisfy the access policy.

% In the performance comparison, FEME and ACME share the same attribute number $l_1=l_2=|x_i=1|=4$, the same row number of access matrix $m_1=m_2=k=2$ and the same coefficient number for access structure $I_1=I_2=\hat{m}=4$.

\renewcommand{\arraystretch}{1.05}
\begin{table}[thbp]\centering 
	\footnotesize
	\setlength{\tabcolsep}{0.7mm}
    \setlength\dashlinedash{1.5pt}     % dash line setting: dash length
    \setlength\dashlinegap{1pt}        % dash line setting: gap length    
	\begin{tabular}{|c:c||c:c||c:c||c:c|}
		\hline  \multicolumn{2}{|c||}{\multirow{2}{*}{Curves}} &  \multicolumn{2}{c||}{MNT159} & \multicolumn{2}{c||}{MNT201} & \multicolumn{2}{c|}{BN256} \\
           \multicolumn{2}{|c||}{~}&  \multicolumn{2}{c||}{(80-bit Security)} & \multicolumn{2}{c||}{(90-bit Security)} & \multicolumn{2}{c|}{(100-bit Security)} \\
		\hline  \multicolumn{2}{|c||}{Schemes}& ACME & \blue{FEME} & ACME & \blue{FEME} & ACME & \blue{FEME}  \\
        \hline\hline
        \multicolumn{2}{|c||}{Algorithms}&\multicolumn{6}{c|}{Computation Costs (ms)} \\
		\hline
		\multicolumn{2}{|c||}{$\Setup$} & 20.526 & \blue{8.411} & 26.882 & \blue{9.699} & 33.344 & \blue{11.402}\\
		\multicolumn{2}{|c||}{$\EKGen$} & 35.451 & \blue{8.124} & 41.365 & \blue{9.393} & 48.485
  & \blue{10.031}\\
		\multicolumn{2}{|c||}{$\DKGen$} & 21.630 & \blue{4.150} & 18.640 & \blue{3.836} &	15.750 & \blue{4.787}\\
		\multicolumn{2}{|c||}{$\PolGen$} & 359.807 & \blue{2.998} & 327.796 & \blue{3.026} &	237.675 & \blue{3.697}\\
		\multicolumn{2}{|c||}{$\Enc$} & 146.931 & \blue{19.660} & 167.337 & \blue{20.302} &	187.822 & \blue{18.275}\\
		\multicolumn{2}{|c||}{$\Dec$} & 123.772 & \blue{20.109} & 188.346 & \blue{28.832} &	231.214 & \blue{27.283}\\
		\hline
		\hline   \multirow{1}{*}{Algo.}& \multirow{1}{*}{Param.}&
		\multicolumn{6}{c|}{Communication Costs (KB)}\\
		\hline
		$\Setup$ &	$|\mpk|$ & 1.044 & \blue{0.344} & 1.332 & \blue{0.428} &	4.128 & \blue{1.071}\\
		$\Setup$ &	$|\msk|$ & 1.2 & \blue{0.065} & 1.36 & \blue{0.074} & 1.6 & \blue{0.083}\\	
		$\EKGen$ & $|\EK|$ & 0.172 & \blue{0.320} & 0.220 & \blue{0.417} &	0.544 & \blue{1.184}\\	
		$\DKGen$ & $|\DK|$ & 0.86 & \blue{0.235} & 1.1 & \blue{0.306} &	2.72 & \blue{0.912}\\	
		$\PolGen$ & $|\SK|$ & 13.932 & \blue{0.127} & 17.82 & \blue{0.165} &	44.064 & \blue{1.169}\\
		$\Enc$ & $|\CT|$ & 164.34 & \blue{23.287} & 212.964 & \blue{29.599} &	537.984 & \blue{63.104}\\	
		\hline
	\end{tabular}
	\caption{Performance of FEME and ACME (on Desktop)}	
	\label{Tab:PerFEME}
\end{table}

Table \ref{Tab:PerFEME} shows that FEME's computation (upper part) and communication costs (lower part) are substantially lower than ACME's. While both schemes see performance declines as security levels increase from 80-bit to 100-bit, FEME's advantage generally grows, except for $\DKGen$ and $\PolGen$, where the performance gap narrows.

In $\Setup$, ACME's computation cost is $(n+3)k^2\cdot \exp_1+k\cdot\exp_T$, while FEME's is only $3\exp_2+2\exp_T$, where $\exp_1$, $\exp_2$, and $\exp_T$ represent the exponentiation costs in $\mathbb{G}_1$, $\mathbb{G}_2$, and $\mathbb{G}_T$, respectively. ACME averages 26.917 ms for system setup, whereas FEME takes 9.837 ms, making it 1.74$\times$ faster. FEME also reduces the sizes of the master public key and master private key by 71.68$\%$ and 94.66$\%$, respectively.

In $\EKGen$, FEME's time cost is $(l_1+2)\exp_1+2\exp_2$, while ACME's time cost is $(n+2)\exp_1+7\exp_2+2 ~pair$, where $pair$ is the computation time for a bilinear pairing operation. FEME averages 9.183 ms versus ACME's 41.767 ms, making it 3.55$\times$ faster, though FEME's $\EK$ size is larger at 0.640 KB compared to ACME’s 0.312 KB.

In $\DKGen$ and $\PolGen$, FEME shows significant improvements: $\DKGen$ time drops from 18.673 ms to 4.258 ms, and $\PolGen$ time from 308.426 ms to 3.240 ms, achieving 3.39$\times$ and 94.19$\times$ speedups, respectively. FEME's $\DK$ and $\SK$ sizes are also reduced by 68.97$\%$ and 98.07$\%$, respectively.
FEME's $\Enc$ and $\Dec$ are 7.62$\times$ and 6.23$\times$ faster than ACME's, with times reduced from 167.383 ms to 19.412 ms and from 181.111 ms to 25.041 ms, respectively. Additionally, FEME's ciphertext size $|\CT|$ is 87.33$\%$ smaller on average.

Overall, FEME significantly outperforms ACME, except in the $|\EK|$ size of $\EKGen$.

\subsection{Evaluation of PriSrv+ and PriSrv}
\label{subsec:Eval-PriSrv+}

\renewcommand{\arraystretch}{1.03}
\begin{table*}[thbp]\centering 
	\footnotesize
	\setlength{\tabcolsep}{1.6mm}
    \setlength\dashlinedash{1.5pt}     % dash line setting: dash length
    \setlength\dashlinegap{1pt}        % dash line setting: gap length    
	\begin{tabular}{|c|c||c:c|c:c||c:c|c:c||c:c|c:c|c|c|c|c|c|c|c|c|c|c|}
		\hline  \multicolumn{2}{|c||}{Curves} &  \multicolumn{4}{c||}{MNT159 (80-bit Security)} & \multicolumn{4}{c||}{MNT201 (90-bit Security)} & \multicolumn{6}{c|}{BN256 (100-bit Security)} \\
		\hline  \multicolumn{2}{|c||}{Costs}& \multicolumn{2}{c|}{Comp. (ms)} & \multicolumn{2}{c||}{Comm. (KB)} & \multicolumn{2}{c|}{Comp. (ms)} & \multicolumn{2}{c||}{Comm. (KB)} & \multicolumn{2}{c|}{Comp. (ms)} & \multicolumn{2}{c|}{Comm. (KB)} \\
		\hline  \multicolumn{2}{|c||}{Protocols}& PriSrv & \blue{PriSrv+} & PriSrv & \blue{PriSrv+} & PriSrv & \blue{PriSrv+} & PriSrv & \blue{PriSrv+} & PriSrv & \blue{PriSrv+} & PriSrv & \blue{PriSrv+} \\
        \hline\hline
        \multirow{1}{*}{No.}& \multirow{1}{*}{Devices}&\multicolumn{12}{c|}{Privacy-preserving Service Broadcast} \\
		\hline
		1 & Desktop & 158.931 & \blue{22.891} & 164.34 & \blue{23.29} &	180.337 & \blue{23.512} & 212.96 & \blue{29.602} &	202.822 & \blue{20.674} & 537.98 & \blue{63.107}\\
		2&	Laptop &	216.493 & \blue{31.168} & 164.34 & \blue{23.29} & 261.059 & \blue{34.035} & 212.96 & \blue{29.602} &	287.287 & \blue{29.281} & 537.98 & \blue{63.107}\\
		3&	Phone &	385.553 & \blue{55.531} & 164.34 & \blue{23.29} & 443.686 & \blue{57.853} & 212.96 & \blue{29.602} & 482.725 & \blue{49.202} & 537.98 & \blue{63.107}\\
		4&	Raspberry Pi&	638.259 & \blue{91.927} & 164.34 & \blue{23.29} & 880.868 & \blue{114.832} & 212.96 & \blue{29.602} &	1188.392 & \blue{121.139} & 537.98 & \blue{63.107}\\		
		\hline
		\hline   \multirow{1}{*}{No.}& \multirow{1}{*}{Devices}&
		\multicolumn{12}{c|}{Anonymous Mutual Authentication}\\
		\hline
		1 &	Desktop &	429.282 & \blue{97.067}	& 164.45 & \blue{24.69} &517.512 & \blue{119.389} & 213.09 & \blue{31.364} & 673.039 & \blue{158.003}	& 538.83 & \blue{66.385}\\	
		2& Laptop &	576.161 & \blue{130.268} & 164.45 & \blue{24.69} &686.054 & \blue{158.271}& 213.09 & \blue{31.364} &	854.177 & \blue{201.518} & 538.83 & \blue{66.385}\\	
		3&	Phone &	727.572 & \blue{164.512}& 164.45 & \blue{24.69} &892.712 & \blue{205.952}& 213.09 & \blue{31.364} &	972.163 & \blue{228.222} & 538.83 & \blue{66.385}\\
		4&	Raspberry Pi &	1224.365 & \blue{276.851}& 164.45 & \blue{24.69} &1832.187 & \blue{422.671}& 213.09 & \blue{31.364} & 2711.013 & \blue{636.443} & 538.83 & \blue{66.385}\\		
		\hline
	\end{tabular}
	\caption{Performance of PriSrv+ and PriSrv (on Four Platforms)}	
	\label{Tab:PerPriSrv}
\end{table*}

To comprehensively evaluate PriSrv+ and PriSrv~\cite{yang2024prisrv}, we conducted tests on four platforms (Table~\ref{Tab:PerPriSrv}) using the same parameters as Table~\ref{Tab:PerFEME} across three elliptic curves. The platforms include a desktop (Intel Core i9-7920X), a laptop (Intel Core i5-10210U), a smartphone (ARM Cortex @2.4GHz), and a Raspberry Pi (ARM Cortex @1.5GHz), covering both high-performance and mobile environments relevant to wireless service discovery.

% Long Version
% To comprehensively evaluate PriSrv+ and PriSrv \cite{yang2024prisrv}, we conducted performance tests on four platforms as shown in Table \ref{Tab:PerPriSrv}, using the same parameter as in Table \ref{Tab:PerFEME} across three elliptic curves. The platforms include a desktop (Intel Core i9-7920X, 12 cores, 16GB RAM), a laptop (Intel Core i5-10210U, 4 cores, 8GB RAM), a smartphone (ARM Cortex @2.4GHz, 4 cores, 4GB RAM), and a Raspberry Pi (ARM Cortex @1.5GHz, 4 cores, 2GB RAM). The desktop setup matched our earlier evaluations of FEME and ACME, while the other platforms were chosen for mobile device relevance to wireless service discovery.
 
PriSrv+ consistently outperforms PriSrv in computation overhead. In the broadcast phase (upper Table~\ref{Tab:PerPriSrv}), PriSrv+ averages 54.336 ms across platforms, compared to PriSrv’s 443.868 ms—7.17$\times$ faster. In the mutual authentication phase (lower Table), PriSrv+ averages 233.181 ms vs. PriSrv’s 1008.02 ms, achieving a 3.32$\times$ improvement. The total computation time of PriSrv+ is 287.517 ms, well under one second and perceived as an ``immediate response''~\cite{heinrich2021privatedrop,yang2024prisrv}, while PriSrv takes 1451.88 ms, making PriSrv+ 4.05$\times$ faster overall.

% Long Version
% PriSrv+ outperforms PriSrv in computation overhead (comp.) across all scenarios. During the privacy-preserving broadcast phase (upper section of Table \ref{Tab:PerPriSrv}), PriSrv+ averages 54.336 ms across four platforms and three elliptic curves, whereas PriSrv averages 443.868 ms, making PriSrv+ 7.17$\times$ faster.
% In the anonymous mutual authentication phase (lower section), PriSrv+ averages 233.181 ms, compared to 1008.02 ms for PriSrv, showing a 3.32$\times$ improvement. PriSrv+ achieves an average total computation time of 287.517 ms for both phases, which is well below one second and thus can be perceived as an ``immediate response" by users~\cite{heinrich2021privatedrop,yang2024prisrv}. In comparison, PriSrv's total computation time is 1451.88 ms, making PriSrv+ 4.05$\times$ faster than PriSrv.

Table \ref{Tab:PerPriSrv} indicates similar communication costs between the two phases within each protocol, dominated by the ciphertext size. FEME’s efficiency results in PriSrv+ reducing broadcast phase communication by 87.33$\%$ and authentication phase by 86.64$\%$ compared to PriSrv on average.

\begin{figure}[htp]        
	\color{black}\begin{center}
		\subfloat[\footnotesize{Computation Overheads}]{			\includegraphics[width=3.25in]{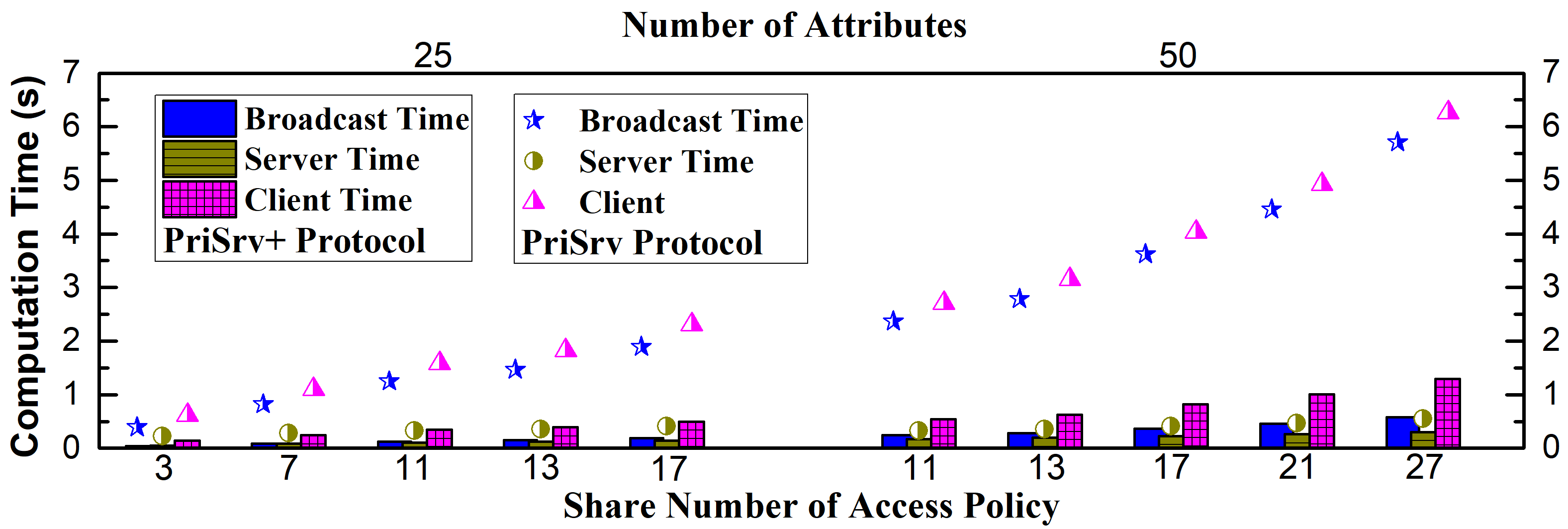}}\\
        \vspace{2mm}
		\subfloat[\footnotesize
        {Communication Overheads}]{			\includegraphics[width=3.25in]{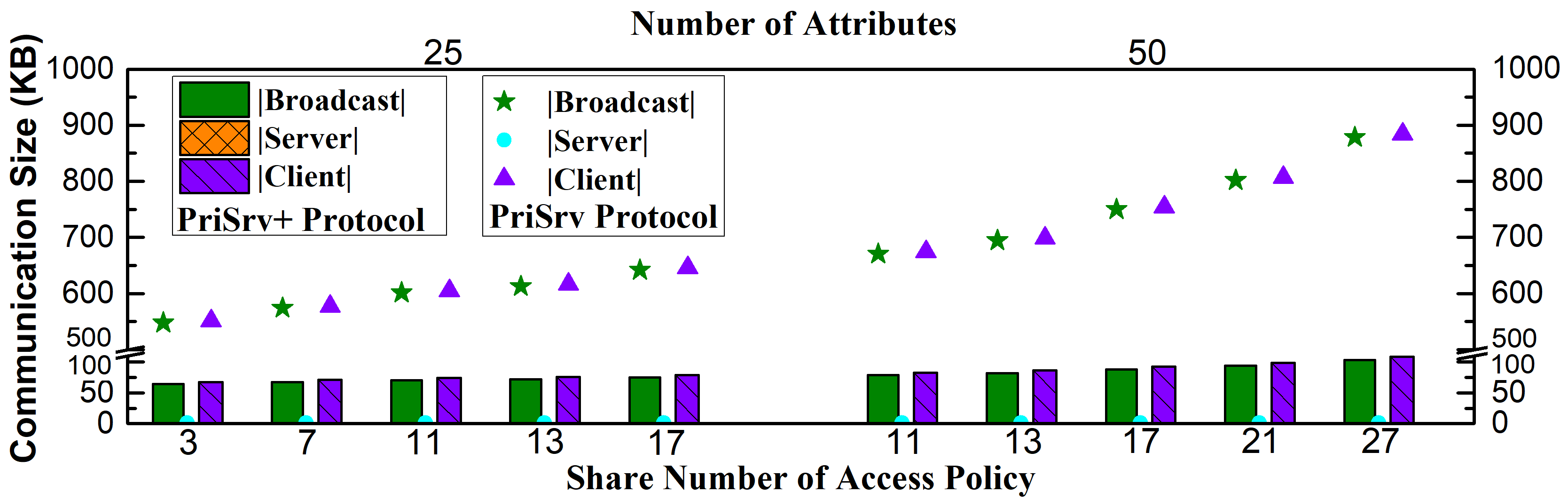}}
	\end{center}
 \vspace{-2mm}
\caption{Comparison of PriSrv+ and PriSrv}
	\label{Fig:PriSrvPerf}
\end{figure}

We implemented two protocols in a real-world wireless environment using an open-source Wi-Fi Alliance project~\cite{WiFi} supporting IEEE 802.1X, with \textit{wpa$\_$supplicant} for the client and \textit{hostapd} for the service provider. Experiments were run on two laptops using BN256.

Fig.~\ref{Fig:PriSrvPerf}(a) compares the computation costs of PriSrv+ and PriSrv under complex policies with more attributes. Bars show PriSrv+’s broadcast time ($T_B$), server time ($T_S$), and client time ($T_C$) in the anonymous authentication phase; symbols represent the same for PriSrv. We vary the number of attributes (top x-axis) and policy share number (bottom x-axis), where the latter reflects the number of shares required to reconstruct a secret~\cite{beimel1996secure} (i.e., the edge number in an equivalent access tree). Both protocols are tested with identical attribute and policy sizes to ensure fairness.

% Long Version
% Fig. \ref{Fig:PriSrvPerf}(a) presents the computation costs of two protocols with more complex policies and larger attribute numbers, where the bars represent PriSrv+'s broadcast time ($T_B$), server's computation time ($T_S$), and client's computation time ($T_C$) during the anonymous authentication phase, while the symbols represent those for PriSrv.
% We compare them by varying the number of attributes (top x-axis) and the share number of access policy (bottom x-axis), where the latter is the number of shares needed to reconstruct a shared secret based on the access policy\cite{beimel1996secure} (i.e., the edge number in an equivalent access tree). The experiment ensures that they share the same number of attributes and policies for a fair comparison.

In Fig. \ref{Fig:PriSrvPerf}(a), we vary the attribute number between $\{25,50\}$ and the share number of the access policy among $\{3,7,11,13,17,21,27\}$ to test performance with expressive policies. PriSrv+ shows significantly lower computation times than PriSrv. Average times for PriSrv+ are $T_B$ = 0.252 s, $T_S$ = 0.166 s, and $T_C$ = 0.592 s, while for PriSrv, $T_B$ = 2.478 s, $T_S$ = 0.374 s, and $T_C$ = 2.853 s, making PriSrv+ 8.83$\times$, 1.25$\times$, and 3.82$\times$ faster, respectively.

Fig. \ref{Fig:PriSrvPerf}(b) shows communication overheads using the same parameters as Fig. \ref{Fig:PriSrvPerf}(a). Bars indicate PriSrv+ broadcast and the service provider's/client's authentication overheads ($|\textsf{Broadcast}|$, $|\textsf{Server}|$, $|\textsf{Client}|$), with PriSrv shown as symbols. The server's communication cost during the authentication phase is constant and identical for both protocols at $|\textsf{Server}|$ = 0.82 KB. PriSrv+ has average communication costs of $|\textsf{Broadcast}|$ = 79.623 KB and $|\textsf{Client}|$ = 83.64 KB, while PriSrv averages $|\textsf{Broadcast}|$ = 677.488 KB and $|\textsf{Client}|$ = 681.505 KB. PriSrv+ reduces these costs by 88.25$\%$ and 87.73$\%$, respectively, compared to PriSrv.

The evaluation shows that PriSrv+ is notably more efficient than PriSrv in computation and communication costs.

% \blue{\textbf{Optimized Efficiency and Scalability}.
% PriSrv+ maintains computational efficiency by ensuring that its overhead scales proportionally with the number of attributes, a fundamental characteristic of well-designed ABE schemes. Unlike conventional ABE systems, where increased attribute flexibility often leads to significant performance degradation, PriSrv+ is specifically optimized to handle an unrestricted attribute universe without incurring excessive computational costs. Empirical results demonstrate up to 7.17$\times$ faster service broadcasts and 3.32$\times$ faster mutual authentication than PriSrv, even with complex attribute sets. These improvements stem from efficient encryption and decryption optimizations, allowing PriSrv+ to provide expressive attribute matching while preserving practical performance for real-time service discovery in wireless networks.} 

\textbf{Optimized Efficiency and Scalability}.
PriSrv+ ensures computational efficiency with overhead growing proportionally to the attribute number, consistent with standard ABE schemes. Unlike conventional ABE systems that suffer from performance degradation with flexible attributes, PriSrv+ supports an unrestricted attribute universe with minimal overhead. Experiments show up to 7.17$\times$ faster service broadcasts and 3.32$\times$ faster mutual authentication than PriSrv, even with complex policies. These gains result from encryption and decryption optimizations, enabling expressive attribute matching while maintaining real-time performance in wireless networks.

\subsection{Interoperability}

PriSrv+ is an enhanced version of PriSrv, improving expressiveness, privacy, and usability. Like PriSrv, it is interoperable with existing SD protocols like mDNS, BLE, EAP, AirDrop.
% aligning with the evolving security needs of modern wireless networks.

\textbf{Integration Approaches.} PriSrv+ can be integrated through two approaches. The first places PriSrv+ at the application layer, allowing it to work with existing lower-layer protocols. If the payload exceeds protocol limits, the lower layers handle segmentation and reassembly without altering PriSrv+’s logic. The second replaces lower-layer protocols with PriSrv+, requiring specific adaptations. We illustrate the first approach using mDNS and BLE, and the second with EAP and AirDrop.

% Long Version
% \textbf{Integration Approaches}. PriSrv+ can be integrated into wireless protocols through two approaches. The first positions PriSrv+ at the application layer, allowing it to work with existing lower-layer protocols. If payload size exceeds protocol limits, lower layers manage segmentation and reassembly without modifying PriSrv+’s logic. The second approach replaces lower-layer protocols with PriSrv+, requiring specific adaptations. Below we use mDNS and BLE as examples to illustrate the first approach and use EAP and AirDrop as examples to illustrate the second.

\textbf{Privacy-Enhanced mDNS and BLE.} PriSrv+ can integrate into the Vanadium framework \cite{Vanadium} to develop privacy-enhanced mDNS and BLE. Vanadium provides service discovery APIs for protocols like mDNS \cite{mDNS} and BLE \cite{BLE}. mDNS, often paired with DNS-SD \cite{DNS-SD}, supports additional attributes in TXT records (up to 65,535 bytes). PriSrv+ broadcasts $(bid,\textsf{CT}_B)$ using 64,622 bytes on BN256 within a single TXT record, whereas PriSrv requires nine TXT records (to transmit 531,996 bytes), reducing the mDNS packet size by 88.89\%.

% Long Version
% \textbf{Privacy-Enhanced mDNS and BLE}. 
% PriSrv+ can integrate into the Vanadium framework \cite{Vanadium} to develop privacy-enhanced mDNS and BLE. Vanadium provides service discovery APIs for protocols like mDNS \cite{mDNS} and BLE \cite{BLE}. mDNS, often paired with DNS Service Discovery (DNS-SD)\cite{DNS-SD}, carries additional attributes in TXT records of up to 65,535 bytes. PriSrv+ broadcasts $(bid,\textsf{CT}_B)$ using 64,622 bytes on BN256, fitting within a single TXT record. In contrast, PriSrv's broadcast of the same message takes 531,996 bytes on BN256 and thus requires nine TXT records. PriSrv+ reduces the mDNS packet size by 88.89$\%$, demonstrating its efficiency in minimizing transmission size.

BLE’s standard 31-byte payload challenges the transmission of large ciphertexts. Using the BLE Attribute Protocol (ATT) and PDU Segmentation, if the payload exceeds BLE’s limit, it is segmented into multiple PDUs and reassembled by the receiver. PriSrv+ achieves an average 87.73\% reduction for clients compared to PriSrv, enabling efficient segmentation and limited fragmentation.

% Long Version
% BLE’s standard 31-byte broadcast payload challenges transmitting large ciphertexts. To support privacy-enhanced BLE with PriSrv+, the BLE Attribute Protocol (ATT) and PDU Segmentation can increase payload capacity. If the payload exceeds BLE’s standard packet size, ATT segments the data into multiple PDUs for sequential transmission, which are reassembled by the receiver. PriSrv+ achieves an average 87.73$\%$ reduction for clients compared to PriSrv, enabling efficient segmentation and minimal fragmentation.

\textbf{Privacy-Enhanced EAP.} PriSrv+ enhances privacy in EAP by integrating encrypted broadcasts and responses. An access point (AP) acts as a pass-through for interactions between the client and service provider. The provider broadcasts privacy-preserving service information, including a broadcast identifier and FEME-generated ciphertext, matching EAP’s initial authentication request. A successful decryption by the client prompts an encrypted reply, which the provider decrypts and then responds with an authentication tag containing DH shares and a MAC. The client verifies the tag, computes a secret session key, and signals success, while the provider computes the same key for a secure session. EAP messages are then encapsulated in EAPOL frames and sent as RADIUS packets. Compared to PriSrv, PriSrv+ reduces communication costs by 86.64\% during mutual authentication.

% Long Version
% \textbf{Privacy-Enhanced EAP}. PriSrv+ enhances privacy in the Extensible Authentication Protocol (EAP) by integrating encrypted broadcasts and responses. An access point (AP) facilitates interactions between the client and service provider, acting as a pass-through for the authentication server. The service provider broadcasts privacy-preserving service information, including a broadcast identifier and ciphertext generated using FEME, aligning with EAP’s initial authentication request. If the client successfully decrypts a broadcast message, it responds with an encrypted reply. The service provider decrypts the reply and, upon success, sends an authentication tag to the client, containing DH shares and a MAC. The client verifies the authentication tag, computes a secret session key, and signals success to the server. The service provider computes the same secret session key to establish a secure session. After this, EAP messages are encapsulated in EAPOL frames and sent as RADIUS packets. Compared to PriSrv, PriSrv+ reduces communication costs by 86.64$\%$ during mutual authentication, making it more suitable for privacy-preserving Wi-Fi connections in resource-constrained environments.

\textbf{Privacy-Enhanced Apple AirDrop.} 
AirDrop uses BLE to broadcast a hashed service provider identity for detecting nearby clients, followed by a TLS handshake that exposes identities via cleartext certificate exchange. Using the PrivateDrop mechanism \cite{heinrich2021privatedrop}, PriSrv+ improves privacy by preventing the transmission of service provider identifiers during BLE advertising and encrypting both parties' certificates with FEME at the start of the TLS handshake. Apple may act as a key generation center, producing the necessary secret keys alongside existing iCloud certificates.

% Long Version
% \textbf{Privacy-Enhanced Apple AirDrop}. 
% AirDrop uses BLE to broadcast the hashed identity of a service provider to detect nearby potential clients. If a match is found, a TLS handshake follows, exchanging their certificates in cleartext and exposing their identities to potential tracking attacks. Using the PrivateDrop mechanism \cite{heinrich2021privatedrop}, PriSrv+ enhances AirDrop's privacy by preventing the transmission of service provider identifiers during BLE advertising and encrypting both parties' certificates with FEME at the start of the TLS handshake. In this setting, Apple may act as a
% key generation center to generate necessary secret keys for the involving parties
% alongside existing iCloud certificates.

\section{Limitations and Future Work} 
\label{sec:Limitations}

While PriSrv+ introduces notable advancements in private service discovery, certain limitations remain.
% These include reliance on a trusted authority, lack of a built-in revocation mechanism, and absence of client abuse mitigation. Addressing these challenges in future work can enhance the robustness and adaptability of PriSrv+.

\textbf{Trusted Authority for Attribute Assignment}.
Although PriSrv+ eliminates credential issuance and revocation burdens of PriSrv, it still depends on a \textit{trusted authority} for attribute assignment during registration—a common but centralized assumption in ABE systems. \textit{Future work} could explore decentralized alternatives to reduce this trust dependency.

\textbf{Revocation Mechanism}.
PriSrv+ does not include an integrated revocation scheme, but it can adopt existing ABE-based approaches such as \textit{key updates}, \textit{attribute expiration}, \textit{proxy re-encryption}, and \textit{server-assisted revocation}. \textit{Future research} may focus on lightweight, scalable revocation strategies suited to dynamic service discovery settings.

\textbf{Mitigating Client Abuse}.
Strong anonymity in PriSrv+ may allow clients to misuse services. Although not addressed natively, PriSrv+ can incorporate ABE-based \textit{traitor tracing} techniques, such as embedding user-specific identifiers in keys or using traceable ciphertexts. \textit{Key-insulated architectures} and dynamic revocation can further mitigate abuse. \textit{Future work} should explore integrating such accountability mechanisms without compromising user privacy.

\section{Conclusion} 
\label{sec:Conclusion}

PriSrv+ significantly advances privacy-preserving service discovery in wireless networks by introducing Fast and Expressive Matchmaking Encryption (FEME). It overcomes the limitations of prior schemes by enabling expressive bilateral access control while enhancing efficiency, security, privacy, and usability. Evaluations demonstrate notable gains in performance and reduced communication overhead, making it well-suited for resource-constrained devices. With formal security guarantees and compatibility with existing wireless protocols, PriSrv+ effectively meets the privacy and security demands of wireless service discovery environments.

% Long Version
% PriSrv+ marks a significant advance in privacy-preserving service discovery protocols for wireless networks. By introducing Fast and Expressive Matchmaking Encryption (FEME), PriSrv+ addresses the limitations of prior schemes, supporting expressive access control policies while improving efficiency, security and privacy, and usability. Comprehensive evaluations show substantial improvements in encryption and decryption performance, with reduced ciphertext size and communication overhead, making PriSrv+ suitable for resource-constrained devices. The interoperability of PriSrv+ with existing wireless service discovery protocols, coupled with formal security proofs, ensures that PriSrv+ fulfills the security and privacy requirements of modern wireless networks.

\section*{Acknowledgment}

The authors thank the anonymous reviewers for their valuable comments and insightful suggestions.

This research is supported by the National Research Foundation, Singapore and Infocomm Media Development Authority under its Trust Tech Funding Initiative, AXA Research Fund, National Natural Science Foundation of China (No. 62372110, 62332007, U22B2028), Fujian Provincial Natural Science of Foundation (No. 2023J02008), Lee Kong Chian Chair Professorship, University of Oregon School of Law, Consumer Protection Research Grant 2025-2026 (under Grant No. 4236D0), National Science Foundation (No. 2112471), Werner Siemens-Stiftung (WSS) as part of the Centre for Cyber Trust (CEYT),  Science and Technology Major Project of Tibetan Autonomous Region of China (No. XZ202201ZD0006G), Open Research Fund of Machine Learning and Cyber Security Interdiscipline Research Engineering Center of Jiangsu Province (No. SDGC2131), National Joint Engineering Research Center of Network Security Detection and Protection Technology, Guangdong Key Laboratory of Data Security and Privacy Preserving, Guangdong Hong Kong Joint Laboratory for Data Security and Privacy Protection, and Engineering Research Center of Trustworthy AI, Ministry of Education. Any opinions, findings and conclusions or recommendations expressed in this material are those of the author(s) and do not reflect the views of National Research Foundation, Singapore and Infocomm Media Development Authority.

% \clearpage
\section*{Ethics Considerations}

This work presents PriSrv+, a cryptographic protocol for privacy-preserving service discovery. It does not involve human subjects, real-world deployments, or personal data—only synthetic datasets and simulated attributes were used. PriSrv+ is designed to enhance user privacy and prevent tracking and profiling attacks, without introducing or exploiting system vulnerabilities. All cryptographic techniques follow established models, and no responsible disclosure was required. We conducted this research in accordance with ethical principles outlined in the Menlo Report and believe it contributes positively to secure communications.

% Long Version
% This work focuses on the design and evaluation of PriSrv+, a cryptographic protocol for privacy-preserving service discovery in wireless networks. Our research does not involve human subjects, personal data, or any user-identifiable information. All experiments and evaluations were conducted using synthetic datasets and simulated service attributes. No real-world deployments, user studies, or data collection from individuals were performed. 
% We believe that PriSrv+ poses no direct ethical risk. On the contrary, it is specifically designed to enhance user privacy, mitigate identity exposure, and prevent tracking and profiling attacks in service discovery protocols. The cryptographic techniques used in PriSrv+ follow well-established models and assumptions in the privacy and security community. 
% This research does not introduce or exploit vulnerabilities in existing systems and does not require responsible disclosure. We followed ethical research principles throughout, in line with the Menlo Report guidelines. We believe the publication of this work contributes positively to the development of secure and privacy-preserving communication systems.

\bibliographystyle{plain}
\bibliography{ref}

% \clearpage
\appendices

\section{Security Model of FEME}
\label{app:SecModel}

The security models for FEME outline its confidentiality, anonymity, and authenticity properties. \textit{Confidentiality} ensures that no probabilistic polynomial-time (PPT) adversary can distinguish between two challenge messages encrypted under a target attribute set and policy, even with access to all the key generation oracles, with a restriction that the decryption keys for the target attribute set and policy have not been queried. \textit{Anonymity} guarantees that no PPT adversary, who outputs two target attribute sets and policies, can distinguish which attribute set or policy was used by the challenger to create a ciphertext,
even with access to all the key generation oracles, with a restriction that the decryption keys for the two target attribute sets and policies have not been queried.  \textit{Authenticity} ensures that an adversary cannot forge a valid ciphertext capable of passing decryption without possessing an attribute encryption key with attributes that satisfy the target access policy, even with access to all the key generation oracles. 

% \begin{definition}
% \label{ACME:correctness}
% A matchmaking encryption encryption scheme $\mathcal{ME}$ is \textit{correct} if for all security parameter $\lambda$,
% \small 
% $$\textsf{Pr}\left[
% 	\begin{array}{l}	
% 	(\mpk,\msk)\xleftarrow{\$}\Setup(1^{\lambda});\\
% 	\textsf{EK}_{\mathcal{S}_{\snd}}\xleftarrow{\$}\EKGen(\msk,\mathcal{S}_{\snd});\\
% 	\textsf{DK}_{\mathcal{S}_{\rcv}}\xleftarrow{\$}\DKGen(\msk,\mathcal{S}_{\rcv});\\
% 	\textsf{SK}_{\mathbb{A}_{\rcv}}\xleftarrow{\$}\PolGen(\msk,\mathbb{A}_{\rcv});\\	
% 	\CT_{\snd}\xleftarrow{\$}\Enc(\textsf{EK}_{\mathcal{S}_{\snd}},\mathbb{A}_{\snd},\msg):\\
% 	 \mathcal{S}_{\snd}\models\mathbb{A}_{\rcv} ~\bigwedge~ \mathcal{S}_{\rcv}\models\mathbb{A}_{\snd} ~\bigwedge\\
% 	\Dec(\textsf{DK}_{\mathcal{S}_{\rcv}},\textsf{SK}_{\mathbb{A}_{\rcv}},\CT_{\snd})=\bot
% 	\end{array}
% 	\right]$$
% \normalsize
% $\leq\nu(\lambda)$, where $\nu$ is a negligible function.
% \end{definition}

\begin{definition}
\label{ME:privacy}
A FEME scheme $\mathcal{FEME}$ satisfies \textit{confidentiality} if for any PPT adversary $\mathcal{A}=(\mathcal{A}_1,\mathcal{A}_2)$, there exists a negligible function $\nu$ such that $\textsf{Adv}_{\mathcal{FEME}}^{\textsf{conf}}(\lambda)\overset{\textsf{def}}{=}$
\small
$$\textsf{Pr}\left[b'=b~
\begin{array}{|l}
	(\mpk,\msk)\leftarrow\Setup(1^{\lambda}),b\xleftarrow{\$}\{0,1\}\\
	(\msg_0^*,\msg_1^*,\mathcal{S}_{\snd}^*, \mathbb{A}_{\snd}^*)\\
	\quad\quad\quad\quad\quad\quad\quad\quad\quad\quad\leftarrow\mathcal{A}_1^{\mathcal{O}_1,\mathcal{O}_2,\mathcal{O}_3}(\mpk)\\	    \CT_{\snd}^*\xleftarrow{\$}\Enc(\EK_{\mathcal{S}_{\snd}^*},\mathbb{A}_{\snd}^*,\msg_b)\\
	b'\leftarrow\mathcal{A}_2^{\mathcal{O}_1,\mathcal{O}_2,\mathcal{O}_3}(\CT_{\snd}^*)\\
\end{array}
\right]$$
\normalsize
$\leq\nu(\lambda),$ where oracles $\mathcal{O}_1$, $\mathcal{O}_2$, $\mathcal{O}_3$ are implemented by $\EKGen$ $(\msk,\cdot)$, $\DKGen(\msk,\cdot)$, $\PolGen(\msk,\cdot)$, respectively. 
It is required that $\mathcal{O}_2$ and $\mathcal{O}_3$  are not queried for attributes and policies that can satisfy $(\mathcal{S}_{\snd}^*, \mathbb{A}_{\snd}^*)$. 
\end{definition}

% This model only captures security under chosen plaintext attacks (CPA). We can extend the above definition by introducing another decryption oracle $\mathcal{O}_4$ which can decrypt ciphertexts except the challenge ciphertext $\CT^*$ to capture security under chosen-ciphertext attacks (CCA). 

\begin{definition}
\label{ME:Anon}
A FEME scheme $\mathcal{FEME}$ satisfies \textit{anonymity} if for any PPT adversary $\mathcal{A}=(\mathcal{A}_1,\mathcal{A}_2)$, there exists a negligible function $\nu$ such that $\textsf{Adv}_{\mathcal{FEME}}^{\textsf{anon}}(\lambda)\overset{\textsf{def}}{=}$
\small
$$\textsf{Pr}\left[b'=b~
\begin{array}{|l}
	(\mpk,\msk)\leftarrow\Setup(1^{\lambda}),b\xleftarrow{\$}\{0,1\}\\
	(\msg^*,\mathcal{S}_{\snd_0}^*, \mathcal{S}_{\snd_1}^*,\mathbb{A}_{\snd_0}^*,\mathbb{A}_{\snd_1}^*)\\
	\quad\quad\quad\quad\quad\quad\quad\quad\quad\quad\leftarrow\mathcal{A}_1^{\mathcal{O}_1,\mathcal{O}_2,\mathcal{O}_3}(\mpk)\\	    \CT_{\snd}^*\xleftarrow{\$}\Enc(\EK_{\mathcal{S}_{\snd}^*},\mathbb{A}_{\snd}^*,\msg_b)\\
	b'\leftarrow\mathcal{A}_2^{\mathcal{O}_1,\mathcal{O}_2,\mathcal{O}_3}(\CT_{\snd}^*)\\
\end{array}
\right]$$
\normalsize
$\leq\nu(\lambda),$ where oracles $\mathcal{O}_1$, $\mathcal{O}_2$, $\mathcal{O}_3$ are implemented by $\EKGen$ $(\msk,\cdot)$, $\DKGen(\msk,\cdot)$, $\PolGen(\msk,\cdot)$, respectively. 
It is required that $\mathcal{O}_2$ and $\mathcal{O}_3$  are not queried for attributes and policies that can satisfy $(\mathcal{S}_{\snd_0}^*, \mathbb{A}_{\snd_0}^*)$ or $(\mathcal{S}_{\snd_1}^*, \mathbb{A}_{\snd_1}^*)$, where $\mathcal{S}_{\snd_0}^*=\{ n_i,v_{i,0}\}_{i\in[\ell_1]}$, $\mathcal{S}_{\snd_1}^*=\{ n_i,v_{i,1}\}_{i\in[\ell_1]}$, $\mathbb{A}_{\snd_0}^*=(\M,\pi,\Psi_{\pi(i),0}=\{ n_{\pi(i)},v_{\pi(i),0}\}_{i\in[m_1]})$, and $\mathbb{A}_{\snd_1}^*=(\M,\pi,\Psi_{\pi(i),1}=\{ n_{\pi(i)},v_{\pi(i),1}\}_{i\in[m_1]})$. 
\end{definition}

\begin{definition}
\label{ME:authenticity}
A FEME scheme $\mathcal{FEME}$ satisfies \textit{authenticity} if for any PPT adversary $\mathcal{A}$, there exists a negligible function $\nu$ such that $\textsf{Adv}_{\mathcal{FEME}}^{\textsf{auth}}(\lambda)\overset{\textsf{def}}{=}$
\small
$$\textsf{Pr}\left[
\begin{array}{l}	    
    (\mpk,\msk)\leftarrow\Setup(1^{\lambda})\\
	(\CT_{\snd},\mathbb{A}_{\rcv},\mathcal{S}_{\rcv})\leftarrow\mathcal{A}^{\mathcal{O}_1,\mathcal{O}_2,\mathcal{O}_3}(\mpk)\\ 
    \textsf{DK}_{\mathcal{S}_{\rcv}}\leftarrow\DKGen(\msk,\mathcal{S}_{\rcv})\\
 
	\textsf{SK}_{\mathbb{A}_{\rcv}}\leftarrow\PolGen(\msk,\mathbb{A}_{\rcv})\\
 
	\msg=\Dec(\textsf{DK}_{\mathcal{S}_{\rcv}},\textsf{SK}_{\mathbb{A}_{\rcv}},\CT_{\snd})\\
	\forall\mathcal{S}_{\snd}\in \mathcal{Q}_{\mathcal{O}_1}:(\mathcal{S}_{\snd}\not\models \mathbb{A}_{\rcv})\bigwedge (\msg\neq\bot)\\
\end{array}
\right]$$
\normalsize
$\leq\nu(\lambda),$ where oracles $\mathcal{O}_1$, $\mathcal{O}_2$, $\mathcal{O}_3$ are implemented by 
$\EKGen(\msk,\cdot)$, $\DKGen(\msk,\cdot)$, $\PolGen(\msk,\cdot)$, and $\mathcal{S}\not\models \mathbb{A}$ denotes that attributes $\mathcal{S}$ does not satisfy $\mathbb{A}$.
\end{definition}
\section{Security Proofs of FEME}
\label{app:proof-FEME}

\textbf{Generic Group Model (GGM)}. We use an extended GGM for bilinear groups, as outlined in \cite{shoup1997lower}. This model includes three random encodings $\sigma_1$, $\sigma_2$, and $\sigma_T$ for the additive group $\mathbb{Z}_q$. These encodings are injective mappings $\sigma_1, \sigma_2, \sigma_T: \mathbb{Z}_q \rightarrow \{0,1\}^m$, where $m > 3\log(q)$. The probability that an adversary $\mathcal{A}$ can guess an element within the image of $\sigma_1$, $\sigma_2$, or $\sigma_T$ is negligible. For $i=1,2,T$, we define the sets $\mathbb{G}_i = {\sigma_i(x) : x \in \mathbb{Z}_p}$. The model provides oracles to compute group operations on $\mathbb{G}_1$, $\mathbb{G}_2$, and $\mathbb{G}_T$, as well as an oracle for a bilinear map $e: \mathbb{G}_1 \times \mathbb{G}_2 \rightarrow \mathbb{G}_T$.

\textbf{Random Oracle}. The challenger $\mathcal{C}$ maintains a list $\mathcal{L}_H$ containing entries of the form $(u_i, h_i, t_i)$, which starts out empty. When the adversary $\mathcal{A}$ queries the oracle with an attribute $u_i = (n_i, v_i)$, the challenger checks whether a tuple with $u_i$ already exists in $\mathcal{L}_H$. 
If such a tuple is found, $\mathcal{C}$ responds with $H(u_i) = h_i \in \mathbb{G}_1$. 
If no match is found, $\mathcal{C}$ selects a random value $t_i\stackrel{\$}{\leftarrow}\mathbb{Z}_p^*$, computes 
$h_i = g_1^{t_i} \in \mathbb{G}_1$, returns $H(u_i) = h_i$, and adds the tuple $(u_i, h_i, t_i)$ to $\mathcal{L}_H$.

% In cryptography, the \textit{Generic Group Model (GGM)} \cite{shoup1997lower} is a theoretical framework used to analyze the security of cryptographic protocols that rely on group operations. In the GGM, the group is treated as a ``black box," meaning that no specific properties of the group elements, other than those arising from the group axioms, are assumed to be known by the adversary. Instead, the adversary can only perform group operations (such as multiplication, inversion, and testing for equality) via an oracle, without direct access to the representation of the group elements.

% Formally, in the \textit{Generic Group Model (GGM)} \cite{shoup1997lower},

% 1. A cyclic group $\mathbb{G}$ of order $q$ is represented abstractly, where group elements are encoded as random labels.

% 2. The adversary is given access to an oracle that allows them to:

%     $-$ Perform the group operation. Given two encodings $\sigma(a)$ and $\sigma(b)$, the oracle returns the encoding $\sigma(ab)$.

%     $-$ Compute inverse. Give an encoding $\sigma(a)$, the oracle returns $\sigma(a^{-1})$.

%     $-$ Test equality. Given two encodings $\sigma(a)$ and $\sigma(b)$, the oracle checks if $a=b$.

% 3. The adversary cannot exploit any structure of the group elements beyond these oracle operations. In particular, they cannot ``deduce" the specific internal representation of group elements or their mathematical structure beyond group axioms.

\subsection{Proof of Theorem \ref{theo:FEME-privacy} (FEME: Confidentiality)}
\label{app:proof-FEME-privacy}

\begin{table*}[thbp]
\tiny
\centering 
\setlength{\tabcolsep}{0.001mm}
\renewcommand\arraystretch{1.2}
     \begin{tabular}{|c|c|c|c|c|c|c|c|c|c|c|c|c|c|c|}     
        \hline  
        $1$ & $\kappa$ & $t_i$ & $t_i\tau$ & $t_ir$ & \violet{$t_is_2$} & $x\!+\!\kappa\tau$ &  \violet{$\kappa\lambda_i\!+\!t_{\pi}s_3$} & \violet{$t_i\tau_1s_3$} & $(\!x\!+\!\kappa\tau_1\!)s_3$ & $\frac{1}{b_1}(\!\lambda_i\!+\!t_{\rho}r'\!)$ & $\frac{1}{b_2}(\!\lambda_i\!+\!t_{\rho}r'\!)$ & $\frac{1}{b_1}(\!\kappa\psi_i\!+\!t_{\rho}r'\!)$ & $\frac{1}{b_2}(\!\kappa\psi_i\!+\!t_{\rho}r'\!)$ \\
        \hline
        $\mu$ & $\mu\kappa$ & $\mu t_i$ & $\mu t_i\tau$ & $\mu t_ir$ & \violet{$\mu t_is_2$} & $\mu(\!x\!+\!\kappa\tau\!)$ &  \violet{$\mu(\!\kappa\lambda_i\!+\!t_{\pi}s_3\!)$} & \violet{$\mu t_i\tau_1s_3$} & $\mu(\!x\!+\!\kappa\tau_1\!)s_3$ & $\frac{\mu}{b_1}(\!\lambda_i\!+\!t_{\rho}r'\!)$ & $\frac{\mu}{b_2}(\!\lambda_i\!+\!t_{\rho}r'\!)$ & $\frac{\mu}{b_1}(\!\kappa\psi_i\!+\!t_{\rho}r'\!)$ & $\frac{\mu}{b_2}(\!\kappa\psi_i\!+\!t_{\rho}r'\!)$ \\
        \hline 
        $b_1$ & $b_1\kappa$ & $b_1t_i$ & $b_1t_i\tau$ & $b_1t_ir$ & \violet{$b_1t_is_2$} & $b_1(\!x\!+\!\kappa\tau\!)$ &  \violet{$b_1(\!\kappa\lambda_i\!+\!t_{\pi}s_3\!)$} & \violet{$b_1t_i\tau_1s_3$} & $b_1(\!x\!+\!\kappa\tau_1\!)s_3$ & $\lambda_i\!+\!t_{\rho}r'$ & $\frac{b_1}{b_2}(\!\lambda_i\!+\!t_{\rho}r'\!)$ & $(\!\kappa\psi_i\!+\!t_{\rho}r'\!)$ & $\frac{b_1}{b_2}(\!\kappa\psi_i\!+\!t_{\rho}r'\!)$ \\
        \hline
        $b_2$ & $b_2\kappa$ & $b_2t_i$ & $b_2t_i\tau$ & $b_2t_ir$ & \violet{$b_2t_is_2$} & $b_2(\!x\!+\!\kappa\tau\!)$ &  \violet{$b_2(\!\kappa\lambda_i\!+\!t_{\pi}s_3\!)$} & \violet{$b_2t_i\tau_1s_3$} & $b_2(\!x\!+\!\kappa\tau_1\!)s_3$ & $\frac{b_2}{b_1}(\!\lambda_i\!+\!t_{\rho}r'\!)$ & -- & $\frac{b_2}{b_1}(\!\kappa\psi_i\!+\!t_{\rho}r'\!)$ & -- \\ 
        \hline
        $r$ & $r\kappa$ & $rt_i$ & $rt_i\tau$ & $t_ir^2$ & \violet{$rt_is_2$} & $r(\!x\!+\!\kappa\tau\!)$ &  \violet{$r(\!\kappa\lambda_i\!+\!t_{\pi}s_3\!)$} & \violet{$rt_i\tau_1s_3$} & $r(\!x\!+\!\kappa\tau_1\!)s_3$ & $\frac{r}{b_1}(\!\lambda_i\!+\!t_{\rho}r'\!)$ & $\frac{r}{b_2}(\!\lambda_i\!+\!t_{\rho}r'\!)$ & $\frac{r}{b_1}(\!\kappa\psi_i\!+\!t_{\rho}r'\!)$ & $\frac{r}{b_2}(\!\kappa\psi_i\!+\!t_{\rho}r'\!)$ \\
        \hline 
        $r'$ & $r'\kappa$ & $r't_i$ & $r't_i\tau$ & $t_irr'$ & \violet{$r't_is_2$} & $r'(\!x\!+\!\kappa\tau\!)$ &  \violet{$r'(\!\kappa\lambda_i\!+\!t_{\pi}s_3\!)$} & \violet{$r't_i\tau_1s_3$} & $r'(\!x\!+\!\kappa\tau_1\!)s_3$ & $\frac{r'}{b_1}(\!\lambda_i\!+\!t_{\rho}r'\!)$ & $\frac{r'}{b_2}(\!\lambda_i\!+\!t_{\rho}r'\!)$ & $\frac{r'}{b_1}(\!\kappa\psi_i\!+\!t_{\rho}r'\!)$ & $\frac{r'}{b_2}(\!\kappa\psi_i\!+\!t_{\rho}r'\!)$ \\
        \hline
        $s_1$ & $s_1\kappa$ & $s_1t_i$ & $s_1t_i\tau$ & $s_1t_ir$ & \violet{$s_1t_is_2$} & $s_1(\!x\!+\!\kappa\tau\!)$ &  \violet{$s_1(\!\kappa\lambda_i\!+\!t_{\pi}s_3\!)$} & \violet{$s_1t_i\tau_1s_3$} & $s_1(\!x\!+\!\kappa\tau_1\!)s_3$ & $\frac{s_1}{b_1}(\!\lambda_i\!+\!t_{\rho}r'\!)$ & $\frac{s_1}{b_2}(\!\lambda_i\!+\!t_{\rho}r'\!)$ & $\frac{s_1}{b_1}(\!\kappa\psi_i\!+\!t_{\rho}r'\!)$ & $\frac{s_1}{b_2}(\!\kappa\psi_i\!+\!t_{\rho}r'\!)$ \\
        \hline
        $s_3$ & $s_3\kappa$ & $s_3t_i$ & $s_3t_i\tau$ & $s_3t_ir$ & \violet{$s_3t_is_2$} & $s_3(\!x\!+\!\kappa\tau\!)$ &  \violet{$s_3(\!\kappa\lambda_i\!+\!t_{\pi}s_3\!)$} & \violet{$t_i\tau_1s_3^2$} & $(\!x\!+\!\kappa\tau_1\!)s_3^2$ & $\frac{s_3}{b_1}(\!\lambda_i\!+\!t_{\rho}r'\!)$ & $\frac{s_3}{b_2}(\!\lambda_i\!+\!t_{\rho}r'\!)$ & $\frac{s_3}{b_1}(\!\kappa\psi_i\!+\!t_{\rho}r'\!)$ & $\frac{s_3}{b_2}(\!\kappa\psi_i\!+\!t_{\rho}r'\!)$ \\ 
        \hline
        $b_1\tau$ & $b_1\tau\kappa$ & $b_1\tau t_i$ & $b_1\tau t_i\tau$ & $b_1\tau t_ir$ & \violet{$b_1\tau t_is_2$} & $(\!x\!+\!\kappa\tau\!)$ &  \violet{$b_1\tau(\!\kappa\lambda_i\!+\!t_{\pi}s_3\!)$} & \violet{$b_1\tau t_i\tau_1s_3$} & $b_1\tau(\!x\!+\!\kappa\tau_1\!)s_3$ & $\tau(\!\lambda_i\!+\!t_{\rho}r'\!)$ & $\frac{b_1\tau}{b_2}(\!\lambda_i\!+\!t_{\rho}r'\!)$ & $\tau(\!\kappa\psi_i\!+\!t_{\rho}r'\!)$ & $\frac{b_1\tau}{b_2}(\!\kappa\psi_i\!+\!t_{\rho}r'\!)$ \\
        \hline
        $b_2\tau$ & $b_2\tau\kappa$ & $b_2\tau t_i$ & $b_2\tau t_i\tau$ & $b_2\tau t_ir$ & \violet{$b_2\tau t_is_2$} & $(\!x\!+\!\kappa\tau\!)$ &  \violet{$b_2\tau(\!\kappa\lambda_i\!+\!t_{\pi}s_3\!)$} & \violet{$b_2\tau t_i\tau_1s_3$} & $b_2\tau(\!x\!+\!\kappa\tau_1\!)s_3$ & $\frac{b_2\tau}{b_1}(\!\lambda_i\!+\!t_{\rho}r'\!)$ & -- & $\frac{b_2\tau}{b_1}(\!\kappa\psi_i\!+\!t_{\rho}r'\!)$ & -- \\  
        \hline
        $b_1s_2'$ & $b_1s_2'\kappa$ & $b_1s_2't_i$ & $b_1s_2't_i\tau$ & $b_1s_2't_ir$ & \violet{$b_1s_2't_is_2$} & $b_1s_2'(\!x\!+\!\kappa\tau\!)$ &  \violet{$b_1s_2'(\!\kappa\lambda_i\!+\!t_{\pi}s_3\!)$} & \violet{$b_1s_2't_i\tau_1s_3$} & $b_1s_2'(\!x\!+\!\kappa\tau_1\!)s_3$ & $s_2'(\!\lambda_i\!+\!t_{\rho}r'\!)$ & $\frac{b_1s_2'}{b_2}(\!\lambda_i\!+\!t_{\rho}r'\!)$ & $s_2'(\!\kappa\psi_i\!+\!t_{\rho}r'\!)$ & $\frac{b_1s_2'}{b_2}(\!\kappa\psi_i\!+\!t_{\rho}r'\!)$ \\ 
        \hline
        $b_2s_2''$ & $b_2s_2''\kappa$ & $b_2s_2''t_i$ & $b_2s_2''t_i\tau$ & $b_2s_2''t_ir$ & \violet{$b_2s_2''t_is_2$} & $b_2s_2''(\!x\!+\!\kappa\tau\!)$ &  \violet{$b_2s_2''(\!\kappa\lambda_i\!+\!t_{\pi}s_3\!)$} & \violet{$b_2s_2''t_i\tau_1s_3$} & $b_2s_2''(\!x\!+\!\kappa\tau_1\!)s_3$ & $\frac{b_2s_2''}{b_1}(\!\lambda_i\!+\!t_{\rho}r'\!)$ & $s_2''(\!\lambda_i\!+\!t_{\rho}r'\!)$ & $\frac{b_2s_2''}{b_1}(\!\kappa\psi_i\!+\!t_{\rho}r'\!)$ & $s_2''(\!\kappa\psi_i\!+\!t_{\rho}r'\!)$ \\  
        \hline
        \multirow{1}{*}{$\alpha\!+\!\kappa r$} & \multicolumn{2}{c|}{$(\alpha\!+\!\kappa r)\kappa$}  & \multicolumn{2}{c|}{$(\alpha\!+\!\kappa r)t_i\tau$} & \multicolumn{2}{c|}{\violet{$(\alpha\!+\!\kappa r)t_is_2$}} &  \multicolumn{2}{c|}{\violet{$(\alpha\!+\!\kappa r)(\!\kappa\lambda_i\!+\!t_{\pi}s_3\!)$}} & $(\alpha\!+\!\kappa r)\cdot$ & \multicolumn{2}{c|}{$\frac{\alpha\!+\!\kappa r}{b_1}(\!\lambda_i\!+\!t_{\rho}r'\!)$} & \multicolumn{2}{c|}{$\frac{\alpha\!+\!\kappa r}{b_1}(\!\kappa\psi_i\!+\!t_{\rho}r'\!)$} \\    
        \cline{1-9}\cline{11-14}
        $\alpha$ & \multicolumn{2}{c|}{$(\alpha\!+\!\kappa r)t_i$} & \multicolumn{2}{c|}{$(\alpha\!+\!\kappa r)t_ir$} & \multicolumn{2}{c|}{$(\alpha\!+\!\kappa r)(\!x\!+\!\kappa\tau\!)$} &  \multicolumn{2}{c|}{\violet{$(\alpha\!+\!\kappa r)t_i\tau_1s_3$}} & $(\!x\!+\!\kappa\tau_1\!)s_3$ & \multicolumn{2}{c|}{$\frac{\alpha\!+\!\kappa r}{b_2}(\!\lambda_i\!+\!t_{\rho}r'\!)$} & \multicolumn{2}{c|}{$\frac{\alpha\!+\!\kappa r}{b_2}(\!\kappa\psi_i\!+\!t_{\rho}r'\!)$} \\ 
        \hline
        \multirow{1}{*}{$b_1\tau_1s_3'$} & \multicolumn{2}{c|}{$b_1\tau_1s_3'\kappa$}  & \multicolumn{2}{c|}{$b_1\tau_1s_3't_i\tau$} & \multicolumn{2}{c|}{\violet{$b_1\tau_1s_3't_is_2$}} &  \multicolumn{2}{c|}{\violet{$b_1\tau_1s_3'(\!\kappa\lambda_i\!+\!t_{\pi}s_3\!)$}} & $b_1\tau_1s_3'\cdot$ & \multicolumn{2}{c|}{$\tau_1s_3'(\!\lambda_i\!+\!t_{\rho}r'\!)$} & \multicolumn{2}{c|}{$\tau_1s_3'(\!\kappa\psi_i\!+\!t_{\rho}r'\!)$} \\    
        \cline{1-9}\cline{11-14}
        $x\mu$ & \multicolumn{2}{c|}{$b_1\tau_1s_3't_i$} & \multicolumn{2}{c|}{$b_1\tau_1s_3't_ir$} & \multicolumn{2}{c|}{$b_1\tau_1s_3'(\!x\!+\!\kappa\tau\!)$} &  \multicolumn{2}{c|}{\violet{$b_1\tau_1^2s_3't_is_3$}} & $(\!x\!+\!\kappa\tau_1\!)s_3$ & \multicolumn{2}{c|}{$\frac{b_1\tau_1s_3'}{b_2}(\!\lambda_i\!+\!t_{\rho}r'\!)$} & \multicolumn{2}{c|}{$\frac{b_1\tau_1s_3'}{b_2}(\!\kappa\psi_i\!+\!t_{\rho}r'\!)$} \\ 
        \hline
        \multirow{1}{*}{$b_2\tau_1s_3''$} & \multicolumn{2}{c|}{$b_2\tau_1s_3''\kappa$}  & \multicolumn{2}{c|}{$b_2\tau_1s_3''t_i\tau$} & \multicolumn{2}{c|}{\violet{$b_2\tau_1s_3''t_is_2$}} &  \multicolumn{2}{c|}{\violet{$b_2\tau_1s_3''(\!\kappa\lambda_i\!+\!t_{\pi}s_3\!)$}} & $b_2\tau_1s_3''\cdot$ & \multicolumn{2}{c|}{$\frac{b_2\tau_1s_3''}{b_1}(\!\lambda_i\!+\!t_{\rho}r'\!)$} & \multicolumn{2}{c|}{$\frac{b_2\tau_1s_3''}{b_1}(\!\kappa\psi_i\!+\!t_{\rho}r'\!)$} \\    
        \cline{1-9}\cline{11-14}
        \blue{$\alpha s\!+\!x\mu s_3$} & \multicolumn{2}{c|}{$b_2\tau_1s_3''t_i$} & \multicolumn{2}{c|}{$b_2\tau_1s_3''t_ir$} & \multicolumn{2}{c|}{$b_2\tau_1s_3''(\!x\!+\!\kappa\tau\!)$} &  \multicolumn{2}{c|}{\violet{$b_2\tau_1s_3''t_i\tau_1s_3$}} & $(\!x\!+\!\kappa\tau_1\!)s_3$ & \multicolumn{2}{c|}{$\tau_1s_3''(\!\lambda_i\!+\!t_{\rho}r'\!)$} & \multicolumn{2}{c|}{$\tau_1s_3''(\!\kappa\psi_i\!+\!t_{\rho}r'\!)$} \\
        \hline
     \end{tabular}
    \caption{Pairing elements in $\mathbb{G}_T$ for the confidentiality (anonymity) proof of FEME \\{\footnotesize($t_{\pi}$ denotes $t_{\pi(i)}$, and $t_{\rho}$ denotes $t_{\rho(i)}$)}}
    \label{Tab:FEME-privacy}
\end{table*}

\noindent\textit{Proof}. Our proof close follows the proof structure in \cite{bethencourt2007ciphertext}. We start with a standard observation derived from a basic hybrid argument. 
    In the confidentiality game, $\mathcal{C}$ generates a challenge ciphertext with a component $\textsf{ct}_0$, which is either $\hat{H}(V)\oplus\phi(\textsf{msg}_0^*)$ or $\hat{H}(V)\oplus\phi(\textsf{msg}_1^*)$, where $V=e(g_1,g_2)^{\alpha s+x\mu s_3}$ and $s=s_1+s_2$. Alternatively, We consider a modified game where $V$ is either $e(g_1,g_2)^{\alpha s+x\mu s_3}$ or $e(g_1,g_2)^{\theta}$, with $\theta$ randomly chosen from $\mathbb{Z}_p^*$. We show that the adversary $\mathcal{A}_1$ in the original game can be reduced to an adversary $\mathcal{A}_2$ in the modified game. Since no $\mathcal{A}_2$ has a non-negligible advantage, it implies that $\mathcal{A}_1$ cannot either.

    Given that $\mathcal{A}_1$ has an advantage $\epsilon$ in the original game, we can construct $\mathcal{A}_2$ as follows. During the challenge phase, after receiving $\msg_0$ and $\msg_1$ from $\mathcal{A}_1$ and $V$ (which is either $V^{(1)}=e(g_1,g_2)^{\alpha s+x\mu s_3}$ or $V^{(2)}=e(g_1,g_2)^{\theta}$ from the challenger $\mathcal{C}$, $\mathcal{A}_2$ flips a coin $\beta\in\{0,1\}$ and sends $\hat{H}(V)\oplus\phi(\msg_\beta)$ to $\mathcal{A}_1$. Once $\mathcal{A}_1$ outputs a bit $\beta'$, $\mathcal{A}_2$ outputs 1 if $\beta'=\beta$, or 0 otherwise. 
    
    If $V=V^{(1)}=e(g_1,g_2)^{\alpha s+x\mu s_3}$, the challenge is a well-formed FEME ciphertext, and $\mathcal{A}_1$ has an advantage $\epsilon$ in correctly guessing $\beta'=\beta$. If $V=V^{(2)}=e(g_1,g_2)^{\theta}$, the challenge becomes independent of $\msg_0$ and $\msg_1$, giving $\mathcal{A}_2$ an advantage of 0. Consequently, we have
    {\small
    \begin{eqnarray*}
    \textsf{Pr}[\mathcal{A}_2\textsf{~win}]
    &=&\textsf{Pr}[V=V^{(1)}]\cdot\textsf{Pr}[\beta'=\beta|V^{(1)}]\\
    &&+\textsf{Pr}[V=V^{(2)}]\cdot\textsf{Pr}[\beta'=\beta|V=V^{(2)}]\\
    &\leq&1/2\cdot(1/2+\epsilon)+1/2\cdot 1/2=1/2+\epsilon/2,
    \end{eqnarray*}}
    and the overall advantage of $\mathcal{A}_2$ is $\epsilon/2$. The existence of any successful $\mathcal{A}_1$ implies the existence of a corresponding $\mathcal{A}_2$ with a non-negligible advantage. 
    
    Finally, we prove that no such $\mathcal{A}_2$ can distinguish between $e(g_1,g_2)^{\alpha s+x\mu s_3}$ and $e(g_1,g_2)^{\theta}$ in polynomial time. The combination of these results shows that no $\mathcal{A}_1$ can have a non-negligible advantage. \vspace{1mm}

    % and the adversary should decide which is the case. Any adversary that has advantage $\epsilon$ in the privacy game can be transformed into an adversary that has an advantage of at least $\epsilon/2$ in the modified privacy game. The reason is illustrated below. Consider two hybrids: one in which the adversary must distinguish between $e(g_1,g_2)^{\alpha s+x\mu s_3}\cdot\textsf{msg}_0^*$ and $e(g_1,g_2)^{\theta}$; another in which it must distinguish between $e(g_1,g_2)^{\theta}$ and $e(g_1,g_2)^{\alpha s+x\mu s_3}\cdot\textsf{msg}_1^*$. Both of these are equivalent to the modified game above.

    \noindent\textbf{Simulation of the Modified Game}. Let $g_1=\sigma_1(1)$, $g_2=\sigma_2(1)$ and $e(g_1,g_2)=\sigma_T(1)$. We write $g_1^x$ to denote $\sigma_1(x)$, $g_2^y$ to denote $\sigma_2(y)$ and $e(g_1,g_2)^z$ to denote $\sigma_T(z)$.

    \noindent{\scriptsize$\bullet$} \textbf{\textit{Setup}}. The challenger $\mathcal{C}$ chooses $\alpha,x,\mu,b_1,b_2,\kappa\stackrel{\$}{\leftarrow}\mathbb{Z}_p^*$, and calculates $Z=e(g_1,g_2)^{\alpha}$, $Y=e(g_1,g_2)^{x\mu}$, $\delta_0=g_2^{\mu}$, $\delta_1=g_2^{b_1}$, $\delta_2=g_2^{b_2}$ and $h=g_1^{\kappa}$. $\mathcal{C}$ sends the master public key $\mpk=(Z,Y,h,\delta_0,\delta_1,\delta_2)$ to $\mathcal{A}$.

    \noindent{\scriptsize$\bullet$} \textbf{\textit{Phase 1}}. In phase 1, $\mathcal{A}$ can make oracle queries to the random oracle and a key generation oracle as follows.

    \noindent$-$ \textit{Random oracle} $(\mathcal{O}_{H})$. Same as defined above.

    \noindent$-$ \textit{Attribute encryption key generation oracle} $(\mathcal{O}_{\EKGen})$. When $\mathcal{A}$ makes a key query for an attribute set $\mathcal{S}_{\textsf{snd}}$, $\mathcal{C}$ picks $\tau\stackrel{\$}{\leftarrow}\mathbb{Z}_p^*$. Then, $\mathcal{C}$ computes  
    {\small    $$\textsf{ek}_{1,i}=g_1^{t_i\tau},~~\textsf{ek}_2=g_2^{b_1\tau},~~\textsf{ek}_3=g_2^{b_2\tau},~~\textsf{ek}_4=g_1^{x+\kappa\tau}.$$}
    Then, $\mathcal{C}$ sends to $\mathcal{A}$ the attribute encryption key $\textsf{EK}_{\mathcal{S}_{\textsf{snd}}}=(\{n_i\}_{i\in[\ell_1]},\{\textsf{ek}_{1,i}\}_{i\in[\ell_1]},\textsf{ek}_2,\textsf{ek}_3,\textsf{ek}_4)$.   

    \noindent$-$ \textit{Attribute decryption key generation oracle} $(\mathcal{O}_{\DKGen})$. When $\mathcal{A}$ makes a key query for an attribute set $\mathcal{S}_{\textsf{rcv}}$, $\mathcal{C}$ picks $r\stackrel{\$}{\leftarrow}\mathbb{Z}_p^*$. Then, $\mathcal{C}$ generates the attribute decryption key as
    {\small    $$\textsf{dk}_1=g_1^{\alpha+\kappa r},~~\textsf{dk}_{2,i}=g_1^{t_ir},~~\textsf{dk}_3=g_2^{r}.$$}
    Then, $\mathcal{C}$ sends to $\mathcal{A}$ the attribute decryption key $\textsf{DK}_{\mathcal{S}_{\textsf{rcv}}}=(\{n_i\}_{i\in[\ell_2]},\textsf{dk}_1,\{\textsf{dk}_{2,i}\}_{i\in[\ell_2]},\textsf{dk}_3)$.

    \noindent$-$ \textit{Policy decryption key generation oracle} $(\mathcal{O}_{\PolGen})$. When $\mathcal{A}$ makes a key query for a policy $\mathbb{A}_{\textsf{rcv}}=(\textbf{A},\rho,\{\Psi_{\rho(i)}\}_{i\in[m_2]})$, $\mathcal{C}$ picks $r'\stackrel{\$}{\leftarrow}\mathbb{Z}_p^*$ and a vector $\textbf{y}\stackrel{\$}{\leftarrow}\mathbb{Z}_p^{m_2-1}$. Let $\lambda_i=\textbf{A}_i(\alpha||\textbf{y})^{\top}$ and $\psi_i=\textbf{A}_i(\mu||\textbf{y})^{\top}$. Note that the $\lambda_i$ (resp. $\psi_i$) are chosen uniformly and independently at random from $\mathbb{Z}_p^*$ subject to the random distribution of $\alpha$ (resp. $\mu$) and $\textbf{y}$. Then, $\mathcal{C}$ generates the policy decryption key as   
    {\small    $$\textsf{sk}_1=g_2^{r'},~~\textsf{sk}_{2,i}=g_1^{\frac{1}{b_1}(\lambda_i+t_{\rho(i)}r')},~~\textsf{sk}_{3,i}=g_1^{\frac{1}{b_2}(\lambda_i+t_{\rho(i)}r')},$$    $$\textsf{sk}_{4,i}=g_1^{\frac{1}{b_1}(\kappa\psi_i+t_{\rho(i)}r')},~~\textsf{sk}_{5,i}=g_1^{\frac{1}{b_2}(\kappa\psi_i+t_{\rho(i)}r')}.$$}
    Then, $\mathcal{C}$ sends to $\mathcal{A}$ the policy decryption key $\textsf{SK}_{\mathbb{A}_{\textsf{rcv}}}=((\textbf{A},\rho,\{n_{\rho(i)}\}_{i\in[m_2]}),\textsf{sk}_1,\{\textsf{sk}_{2,i},\textsf{sk}_{3,i},\textsf{sk}_{4,i},\textsf{sk}_{5,i}\}_{i\in[m_2]})$.   

    \noindent{\scriptsize$\bullet$} \textbf{\textit{Challenge}}. $\mathcal{A}$ outputs the sender's attribute sets $\mathcal{S}_{\textsf{snd}}^*$, a policy $\mathbb{A}_{\textsf{snd}}^*$ and two messages $\textsf{msg}_0^*$, $\textsf{msg}_1^*$ of equal length that it intends to challenge. $\mathcal{C}$ checks if $\mathcal{S}_{\textsf{snd}}^*$ satisfies any of the access policy $\mathbb{A}_{\textsf{rcv}}$ queried in Phase 1. If yes, $\mathcal{C}$ aborts. Otherwise, $\mathcal{C}$ chooses $s_1,s_2',s_2'',s_3',s_3'',\tau'\stackrel{\$}{\leftarrow}\mathbb{Z}_p^*$ and sets $s_2=s_2'+s_2''$, $s_3=s_3'+s_3''$, $s=s_1+s_2$. Then $\mathcal{C}$ selects $\lambda_1,\cdots,\lambda_{m_1}\stackrel{\$}{\leftarrow}\mathbb{Z}_p^*$ for encryption of $\mathcal{S}_{\textsf{snd}}^*$. 
    The challenger selects $\theta\stackrel{\$}{\leftarrow}\mathbb{Z}_p^*$. The challenger flips random coin $b\in\{0,1\}$ to encrypt $\msg_b^*$, and random coin $\beta\in\{0,1\}$ to determine which of the following ciphertext should be created. 
    
    If $\beta=0$, it generates the challenge ciphertext as follows:  
    {\small    $$V=e(g_1,g_2)^{\alpha s+x\mu s_3},~\ct_0=\phi(\msg_b^*)\oplus\hat{H}(V),$$ $$\textsf{ct}_1=g_2^{s_1},~\textsf{ct}_2=g_2^{s_3},$$    $$\textsf{ct}_{3,i}=g_1^{\kappa\lambda_i+t_{\pi(i)}s_3},~\textsf{ct}_{4,1}=g_2^{b_1s_2'},~\textsf{ct}_{4,2}=g_2^{b_2s_2''},$$
    $$\textsf{ct}_{5,i}=g_1^{t_is_2},~\textsf{ct}_{6,i}=g_1^{t_i\tau_1s_3},~\textsf{ct}_7=g_2^{b_1\tau_1s_3'},$$
    $$\textsf{ct}_8=g_2^{b_2\tau_1s_3''},~\textsf{ct}_9=g_1^{(x+\kappa\tau_1)s_3}, \text{~where~}\tau_1=\tau+\tau'.$$}
    
    Otherwise, it generates $V=e(g_1,g_2)^{\theta}$, and the other ciphertext components are kept the same. 
    
    Then, $\mathcal{C}$ sends to adversary $\mathcal{A}$ the ciphertext $\textsf{CT}_{\snd}=((\textbf{M},\pi,\{n_{\pi(i)}\}_{i\in[m_1]}),\textsf{ct}_0,\textsf{ct}_1,\textsf{ct}_2,(\textsf{ct}_{3,i})_{i\in[m_1]},\textsf{ct}_{4,1},\textsf{ct}_{4,2},$ $\{\textsf{ct}_{5,i},$ $\textsf{ct}_{6,i}\}_{i\in[\ell_1]},\textsf{ct}_7,\textsf{ct}_8,\textsf{ct}_9)$.

    \noindent{\scriptsize$\bullet$} \textbf{\textit{Phase 2}}. It is the same as in Phase 1 with the restriction that any input access policy $\mathbb{A}$ are not allowed to satisfy the challenge attribute sets $\mathcal{S}_{\textsf{snd}}^*$.

    \noindent{\scriptsize$\bullet$} \textbf{\textit{Guess}}. $\mathcal{A}$ outputs a bit as a guess.

    \textbf{Analysis of $\mathcal{A}$'s Success Probability}. To demonstrate that no PPT adversary $\mathcal{A}$ can distinguish between $\ct_0$ in the aforementioned game, we assume the contrary. The only way $\mathcal{A}$'s views could differ is if there exist two distinct terms yielding the same result when $\theta=\delta(\alpha s+x\mu s_3)$ but producing different results when $\theta$ is sampled randomly. Let $\theta_1$ and $\theta_2$ be two such terms. Since $\theta$ only occurs in $e(g_1,g_2)^{\theta}$, which cannot be paired, $\mathcal{A}$ can only create queries where $\theta$ is an additive term. Thus, $\theta_1$ and $\theta_2$ can be written as $\theta_1=\delta\theta+\theta_1'$ and $\theta_2=\delta\theta+\theta_2'$ for some $\theta_1'$, $\theta_2'$ that do not contain $\theta$. Based on the assumption that $\theta_1=\theta_2$ when $\theta=\delta(\alpha s+x\mu s_3)$, we have $\delta_1(\alpha s+x\mu s_3)+\theta_1'=\delta_2(\alpha s+x\mu s_3)+\theta_2'$. Rearranging this equation gives $\theta_1'-\theta_2'=(\delta_2-\delta_1)(\alpha s+x\mu s_3)$ , implying that $\mathcal{A}$ can algebraically construct $e(g_1,g_2)^{\delta(\alpha s+x\mu s_3)}$ for some $\delta\in\mathbb{Z}_q$ using the oracle outputs it has already queried. Below, we show that constructing such an expression is computationally infeasible, giving $\mathcal{A}$ only a negligible advantage in winning the confidentiality game.

    % If $\mathcal{A}$ can construct $e(g_1,g_2)^{\delta(\alpha s+x\mu s_3)}$ for some $\delta\in\mathbb{Z}_p^*$ that can be combined from the oracle outputs he has already queried, then $\mathcal{A}$ can use it to distinguish $e(g_1,g_2)^{\alpha s+x\mu s_3}$ from $e(g_1,g_2)^{\theta}$. Therefore, we show that $\mathcal{A}$ can construct $e(g_1,g_2)^{\delta(\alpha s+x\mu s_3)}$ for some $\delta$ with only negligible probability, and cannot gain a non-negligible advantage in the privacy game.   

    To calculate the probability of $\mathcal{A}$ constructing $e(g_1,g_2)^{\delta(\alpha s+x\mu s_3)}$ for some $\delta\in\mathbb{Z}_q$, we perform a case analysis based on the information $\mathcal{A}$ receives from the simulation. For completeness, we first summarize the exponent elements available to $\mathcal{A}$ in groups $\mathbb{G}_1$, $\mathbb{G}_2$, $\mathbb{G}_T$.\vspace{1mm}

    \noindent$-$ $\mathbb{G}_1$ elements: $1$, $\kappa$, $t_i$, $t_i\tau$, $t_ir$, $t_is_2$, $x+\kappa\tau$,  $\kappa\lambda_i+t_{\pi(i)}s_3$, $t_i\tau_1s_3$, $(x+\kappa\tau_1)s_3$, $\frac{1}{b_1}(\lambda_i+t_{\rho(i)}r')$, $\frac{1}{b_2}(\lambda_i+t_{\rho(i)}r')$, $\frac{1}{b_1}(\kappa\psi_i+t_{\rho(i)}r')$, $\frac{1}{b_2}(\kappa\psi_i+t_{\rho(i)}r')$.\vspace{1mm}

    \noindent$-$ $\mathbb{G}_2$ elements: 1, $\mu$, $b_1$, $b_2$, $r$, $r'$, $s_1$, $s_3$, $b_1\tau$, $b_2\tau$, $b_1s_2'$, $b_2s_2''$, $\alpha+\kappa r$,  $b_1\tau_1s_3'$, $b_2\tau_1s_3''$. \vspace{1mm}

    \noindent$-$ $\mathbb{G}_T$ elements: 1, $\alpha$, $x\mu$.\vspace{1mm}

    %listed in sequence 
    % \noindent$-$ $\mathbb{G}_1$ elements: $1$, $\kappa$, $t_i$, $t_i\tau$,  $x+\kappa\tau$, $t_ir$, $\frac{1}{b_1}(\lambda_i+t_ir')$, $\frac{1}{b_2}(\lambda_i+t_ir')$, $\frac{1}{b_1}(\kappa\lambda_i+t_ir')$, $\frac{1}{b_2}(\kappa\lambda_i+t_ir')$, $\kappa\lambda_i+t_is_3$, $t_is_2$, $t_i(\tau+\tau')s_3$, $(x+\kappa(\tau+\tau'))s_3$.\vspace{1mm}

    % \noindent$-$ $\mathbb{G}_2$ elements: 1, $\mu$, $b_1$, $b_2$, $b_1\tau$, $b_2\tau$, $\alpha+\kappa r$, $r$, $r'$, $s_1$, $s_3$, $b_1s_2'$, $b_2s_2''$, $b_1(\tau+\tau')s_3'$, $b_2(\tau+\tau')s_3''$.\vspace{1mm}

    % \noindent$-$ $\mathbb{G}_T$ elements: 1, $\alpha$, $x\mu$.\vspace{1mm}

    We now enumerate all possible queries into $\mathbb{G}_T$ using the bilinear map and the group elements available to $\mathcal{A}$, as shown in Table \ref{Tab:FEME-privacy}. Note that the blue element $\alpha s+x\mu s_3$ in Table \ref{Tab:FEME-privacy} is not used in this proof.
    
    $\mathcal{A}$ can compute arbitrary linear combinations of these terms, and we will demonstrate that none can take the form $\delta(\alpha s+x\mu s_3)$, where $s=s_1+s_2$.

    %----------------------------------------------------

    (1) Consider how to construct $e(g_1,g_2)^{\delta\alpha s_1}$ for some $\delta$. From Table \ref{Tab:FEME-privacy}, the only possibility is that $\mathcal{A}$ could create $\lambda_is_1$ using terms such as $s_1(\kappa\lambda_i+t_{\pi(i)}s_3)$, $s_1\kappa$, $s_1t_i$, $s_1t_i\tau$, $s_1t_ir$, $s_1t_is_2$, $s_1(x+\kappa\tau)$, $s_1t_i\tau_1s_3$, and $s_1(x+\kappa\tau_1)s_3$. However, $\lambda_is_1$ cannot be combined through addition or subtraction with any other elements in $\mathbb{G}_1$, $\mathbb{G}_2$, or $\mathbb{G}_T$. Hence, constructing $\delta\alpha s_1$ in $\mathbb{G}_T$ is impossible for $\mathcal{A}$.

    %----------------------------------------------------

    (2) Now, consider constructing $e(g_1,g_2)^{\delta\alpha s_2}$ for some $\delta$, where $s_2=s_2'+s_2''$. According to Table \ref{Tab:FEME-privacy}, $\mathcal{A}$ cannot combine $s_2=s_2'+s_2''$ via simple addition or subtraction within the elements of $\mathbb{G}_1$, $\mathbb{G}_2$, or $\mathbb{G}_T$. The only way $\mathcal{A}$ could generate a term involving $\alpha s_2$ is by pairing $\frac{1}{b_1}\cdot(\lambda_i+t_{\rho(i)}r')$ with $b_1\tau_1s_3'$ and pairing $\frac{1}{b_2}\cdot(\lambda_i+t_{\rho(i)}r')$ with $b_2\tau_1s_3''$, producing $(\lambda_i+t_{\rho(i)}r')s_2'$ and $(\lambda_i+t_{\rho(i)}r')s_2''$, which combine to form $(\lambda_i+t_{\rho(i)}r')(s_2'+s_2'')=\lambda_is_2+t_{\rho(i)}r's_2$ in $\mathbb{G}_T$.

    For $\mathcal{A}$ to construct $\delta\alpha s_2$ in $\mathbb{G}_T$, it must first construct $t_{\rho(i)}r's_2$ and then cancel out the term of $\lambda_is_2$. From Table \ref{Tab:FEME-privacy}, $\mathcal{A}$ can deduce $t_{\rho(i)}r's_2$ by pairing $t_{\rho(i)}s_2$ (derived from $t_is_2$) with $r'$. However, 
    canceling $\lambda_i s_2$ by reconstructing $\lambda_i$ as $\alpha$ is impossible. The reason is that the input access policy $\mathbb{A}$ cannot be satisfied by the attribute sets $\mathcal{S}^*$. Therefore, constructing $\delta\alpha s_2$ in $\mathbb{G}_T$ is infeasible for $\mathcal{A}$.
    
    %----------------------------------------------------

    (3) Finally, consider constructing $e(g_1, g_2)^{\delta x\mu s_3}$ for some $\delta$, where $s_3 = s_3' + s_3''$. As shown in Table \ref{Tab:FEME-privacy}, combining $s_3 = s_3' + s_3''$ through addition or subtraction within the elements of $\mathbb{G}_1$, $\mathbb{G}_2$, or $\mathbb{G}_T$ is impossible. The only way $\mathcal{A}$ could generate a term involving $xs_3$ is by pairing $\frac{1}{b_1} \cdot (\kappa\psi_i + t_{\rho(i)}r')$ with $b_1\tau_1s_3'$ and pairing $\frac{1}{b_2} \cdot (\kappa\psi_i + t_{\rho(i)}r')$ with $b_2\tau_1s_3''$, yielding $(\kappa\psi_i + t_{\rho(i)}r')\tau_1s_3'$ and $(\kappa\psi_i + t_{\rho(i)}r')\tau_1s_3''$, which combine to produce $(\kappa\psi_i + t_{\rho(i)}r')\tau_1s_3$ in $\mathbb{G}_T$.
    
    Thus, $\mathcal{A}$ can only create $x\psi_i s_3$ from terms such as $(\kappa\psi_i + t_{\rho(i)}r')\tau_1s_3$, $s_3\kappa$, $s_3t_i$, and $s_3(x + \kappa\tau)$. However, combining $x\psi_i s_3$ via addition or subtraction within the elements of $\mathbb{G}_1$, $\mathbb{G}_2$, or $\mathbb{G}_T$ is impossible. Therefore, constructing $\delta x\mu s_3$ in $\mathbb{G}_T$ is infeasible for $\mathcal{A}$.

    %----------------------------------------------------

    \textbf{Analysis of Simulation Failure}. Let $n$ represent the total number of group elements $\mathcal{A}$ receives from its oracle queries to the hash function, groups $\mathbb{G}_1$, $\mathbb{G}_2$, $\mathbb{G}_T$, the bilinear map $e$, and its interaction with the confidentiality game. We show that the views of $\mathcal{A}$ for $\beta = 0$ and $\beta = 1$ are identically distributed, except with probability $O(n^2/q)$, based on the randomness in the variable values chosen during the simulation. This probability arises from an accidental collision, where two distinct polynomials in the GGM evaluate to the same value. As indicated in Table \ref{Tab:FEME-privacy}, the polynomial's maximum degree is 5. By the Schwartz-Zippel lemma    \cite{schwartz1980fast,zippel1979probabilistic}, the probability of such a collision occurring is $O(1/q)$. Using a union bound, the probability of a collision across all $n$ queries is at most $O(n^2/q)$, which is negligible when $q$ is exponentially large in the secret parameter $\kappa$.

    In conclusion, $\mathcal{A}$ only holds a negligible advantage in the modified game, implying that it also has a negligible advantage in the confidentiality game. 
    
    This completes the proof of Theorem \ref{theo:FEME-privacy}. $\hfill\blacksquare$

\subsection{Proof of Theorem \ref{theo:FEME-anon} (FEME: Anonymity)}
\label{app:proof-FEME-anon}

\noindent\textit{Proof}. The logic flow of this proof is similar to the one for the Theorem \ref{theo:FEME-privacy}. Firstly, we consider a modified security game that implies anonymity. Next, we bound $\mathcal{A}$'s success probability and analyze the simulation failure. 
    % In the anonymity experiment (defined in \red{ABC}), the challenger randomly samples a bit $b\in\{0,1\}$ and creates $\textsf{ct}_{5,i}$ of a challenge ciphertext either 
    
    In the anonymity game, the ciphertext components related to the two challenged attribute sets and the two challenged access policies are     $\textsf{ct}_{3,i}=h^{\textbf{M}_i(s_1||\textbf{v})^{\top}}\cdot H(\Psi_{\pi(i)})^{s_3}$, $\textsf{ct}_{5,i}=H(u_i)^{s_2}$ and $\textsf{ct}_{6,i}=H(u_i)^{\tau_1}$, where $\tau_1=\tau+\tau'$. Similarly, we can simulate them as $\textsf{ct}_{3,i}=h^{\textbf{M}_i(s_1||\textbf{v})^{\top}}\cdot H(\Psi_{\pi(i)})^{s_3}=g_1^{\kappa\lambda_i+t_{\pi(i)}s_3}$, $\textsf{ct}_{5,i}=H(u_i)^{s_2}=g_1^{t_is_2}$ and $\textsf{ct}_{6,i}=H(u_i)^{\tau_1}=g_1^{t_i\tau_1s_3}$, where $\lambda_i$ are chosen uniformly and independently at random from $\mathbb{Z}_p^*$ subject to the random distribution of $s_1$ and $\textbf{v}$. 

   The challenger produces a challenge ciphertext element $\textsf{ct}_{3,i}$ either $g_1^{\kappa\lambda_i+t_{\pi(i),0}s_3}$ or $g_1^{\kappa\lambda_i+t_{\pi(i),1}s_3}$ (given two challenge tag sets $\{t_{\pi(i),0}\}_i$, $\{t_{\pi(i),1}\}_i$, using $g_1^{t_{\pi(i),0}}$, $g_1^{t_{\pi(i),1}}$ as the querying results of $H(\Psi_{\pi(i)})$ from the random oracle, respectively), generates $\textsf{ct}_{5,i}$ either $g_1^{t_{i,0}\cdot s_2}$ or $g_1^{t_{i,1}\cdot s_2}$, creates $\textsf{ct}_{6,i}$ either $g_1^{t_{i,0}\cdot \tau_1s_3}$ or $g_1^{t_{i,1}\cdot \tau_1s_3}$ (given two challenge tag sets $\{t_{i,0}\}_i$, $\{t_{i,1}\}_i$, using $g_1^{t_{i,0}}$, $g_1^{t_{i,1}}$ as the querying results of $H(u_i)$ from the random oracle, respectively).

    We can instead consider a modified experiment in which the challenger randomly samples a bit $\beta\stackrel{\$}{\leftarrow}\{0,1\}$ and $\theta_1,\theta_2,\theta_3\stackrel{\$}{\leftarrow}\mathbb{Z}_p^*$. The generated challenge ciphertext element $\textsf{ct}_{3,i}$ is either $g_1^{\kappa\lambda_i+t_{\pi(i),\beta}s_3}$ or $g_1^{\theta_1}$, the element $\textsf{ct}_{5,i}$ is either $g_1^{t_{i,\beta}\cdot s_2}$ or $g_1^{\theta_2}$, and $\textsf{ct}_{6,i}$ is either $g_1^{t_{i,\beta}\cdot \tau_1s_3}$ or $g_1^{\theta_3}$.

    Therefore, assume that $\mathcal{A}$ has an advantage $\epsilon$ in winning the anonymity game. Then, $\mathcal{A}$ has an advantage $\epsilon/2$ in distinguishing $g_1^{\kappa\lambda_i+t_{\pi(i),\beta}s_3}$ from $g_1^{\theta_1}$, an advantage $\epsilon/2$ in distinguishing $g_1^{t_{i,\beta}\cdot s_2}$ from $g_1^{\theta_2}$, and an advantage $\epsilon/2$ in distinguishing $g_1^{t_{i,\beta}\cdot \tau_1s_3}$ from $g_1^{\theta_3}$. For instance, the probability of distinguishing  $g_1^{t_{i,0}\cdot \tau_1s_3}$ from $g_1^{\theta_3}$ is equal to that of distinguishing $g_1^{t_{i,1}\cdot \tau_1s_3}$ from $g_1^{\theta_3}$. We call them the three tuples for simplicity in the following proof.\vspace{1mm}

    % For simplicity, we denote $g_1^{t_{i,\beta}$ and $g_1^{t_{i,\beta}\cdot \theta}$ as $g_1^{t_i s_2}$ and $g_1^{t_i\theta}$ respectively in all further paragraphs. The modified game is simulated as follows.

    \textbf{Simulation of the Modified Game}. The adversary $\mathcal{A}$ aims to distinguish the three tuples.
    
    \noindent{\scriptsize$\bullet$} \textbf{Setup}. Same as defined in the proof of Theorem \ref{theo:FEME-privacy}.

    \noindent{\scriptsize$\bullet$} \textbf{Phase 1}. In phase 1, $\mathcal{A}$ can make oracle queries to the random oracle and a key generation oracle as follows.

    \noindent$-$ \textit{Random oracle} $(\mathcal{O}_{H})$. Same as defined above.

    \noindent$-$ \textit{Key generation oracles} $(\mathcal{O}_{\EKGen},\mathcal{O}_{\DKGen},\mathcal{O}_{\PolGen})$. Same as defined in the proof of Theorem \ref{theo:FEME-privacy}.    

    \noindent{\scriptsize$\bullet$} \textbf{Challenge}. $\mathcal{A}$ outputs $\mathcal{S}_0^*=\{u_{i,0}\}_{i\in[m]}=\{n_i,v_{i,0}\}_{i\in[\ell_1]}$, $\mathcal{S}_1^*=\{u_{i,1}\}_{i\in[\ell_1]}=\{n_i,v_{i,1}\}_{i\in[\ell_1]}$, and two access policies $\mathbb{A}_{\snd_0}^*=(\M,\pi,\Psi_{\pi(i),0}=\{n_{\pi(i)},v_{\pi(i),0}\}_{i\in[m_1]})$, $\mathbb{A}_{\snd_1}^*=(\M,\pi,\Psi_{\pi(i),1}=\{n_{\pi(i)},v_{\pi(i),1}\}_{i\in[m_1]})$ that it intends to attack. Note that $\mathcal{S}_0^*,\mathcal{S}_1^*$ have the same attribute names $\{n_i\}_{i\in[\ell_1]}$, and $\mathbb{A}_{\snd_0}^*$, $\mathbb{A}_{\snd_1}^*$ have the same attribute names $\{n_{\pi(i)}\}_{i\in[m_1]}$. $\mathcal{C}$ checks whether $\mathcal{S}_0^*$ or $\mathcal{S}_1^*$ satisfies any of the access policy $\mathbb{A}$ queried in Phase 1. If yes, $\mathcal{C}$ rejects $\mathcal{S}_0^*$, $\mathcal{S}_1^*$. $\mathcal{C}$ also checks whether any of the attribute set $\mathcal{S}$ queried in Phase 1 satisfies $\mathbb{A}_{\snd_0}^*$ or $\mathbb{A}_{\snd_1}^*$. If yes, $\mathcal{C}$ rejects $\mathbb{A}_{\snd_0}^*$, $\mathbb{A}_{\snd_1}^*$. Otherwise, $\mathcal{C}$ chooses $\theta,s_1,s_2',s_2'',s_3',s_3'',\tau'\stackrel{\$}{\leftarrow}\mathbb{Z}_p^*$ and sets $s_2=s_2'+s_2''$, $s_3=s_3'+s_3''$, $s=s_1+s_2$. Then $\mathcal{C}$ selects $\beta\stackrel{\$}{\leftarrow}\{0,1\}$ for encryption one set of attributes, and flips a coin $b\in\{0,1\}$. 
    
    If $b=0$, it generates the challenge ciphertext as follows:    
    {\small     $$V=e(g_1,g_2)^{\alpha s+x\mu s_3},~\ct_0=\phi(\msg)\oplus\hat{H}(V),~\textsf{ct}_1=g_2^{s_1},$$    $$\textsf{ct}_2=g_2^{s_3},~\textsf{ct}_{3,i}=g_1^{\kappa\lambda_i+t_{\pi(i),\beta}s_3},~\textsf{ct}_{4,1}=g_2^{b_1s_2'},~\textsf{ct}_{4,2}=g_2^{b_2s_2''},$$
    $$\textsf{ct}_{5,i}=g_1^{t_{i,\beta}s_2},~\textsf{ct}_{6,i}=g_1^{t_{i,\beta}\tau_1s_3},~\textsf{ct}_7=g_2^{b_1\tau_1s_3'},$$
    $$\textsf{ct}_8=g_2^{b_2\tau_1s_3''},~\textsf{ct}_9=g_1^{(x+\kappa\tau_1)s_3}, \text{~where~}\tau_1=\tau+\tau'.$$
    }
    
    Otherwise, it generates $\textsf{ct}_{3,i}=g_1^{\theta_1}$, $\textsf{ct}_{5,i}=g_1^{\theta_2}$, $\textsf{ct}_{6,i}=g_1^{\theta_3}$, and the other ciphertext components are kept the same.
    
    Then, $\mathcal{C}$ sends to adversary $\mathcal{A}$ the ciphertext $\textsf{CT}_{\snd}=((\textbf{M},\pi,\{n_{\pi(i)}\}_{i\in[m_1]}),\textsf{ct}_0,\textsf{ct}_1,\textsf{ct}_2,(\textsf{ct}_{3,i})_{i\in[m_1]},\textsf{ct}_{4,1},\textsf{ct}_{4,2},$ $\{\textsf{ct}_{5,i},$ $\textsf{ct}_{6,i}\}_{i\in[\ell_1]},\textsf{ct}_7,\textsf{ct}_8,\textsf{ct}_9)$.

    \noindent{\scriptsize$\bullet$} \textbf{Phase 2}. It is the same as in Phase 1 with the restriction that any input access policy $\mathbb{A}$ are not allowed to satisfy the challenge attribute sets $\mathcal{S}_0^*$ and $\mathcal{S}_1^*$.

    \noindent{\scriptsize$\bullet$} \textbf{\textit{Guess}}. $\mathcal{A}$ outputs a bit as a guess.

    \textbf{Analysis of $\mathcal{A}$'s Success Probability}. For simplicity, we denote $t_{\pi(i),\beta}$ as $t_{\pi(i)}$, and $t_{i,\beta}$ as $t_i$ in all further paragraphs.
    Suppose that $\mathcal{A}$ can algebraically construct $e(g_1,g_2)^{\delta_1(\kappa\lambda_i+t_{\pi(i)}s_3)}$, $e(g_1,g_2)^{\delta_2t_is_2}$, $e(g_1,g_2)^{\delta_3t_i\tau_1s_3}$ for some $\delta_1,\delta_2,\delta_3\in\mathbb{Z}_q$ using all oracle outputs it has already queried, then it can use them to distinguish the three tuples.

    To calculate the probability of $\mathcal{A}$ constructing $e(g_1,g_2)^{\delta_1(\kappa\lambda_i+t_{\pi(i)}s_3)}$, $e(g_1,g_2)^{\delta_2 t_is_2}$, $e(g_1,g_2)^{\delta_3 t_i\tau_1s_3}$ for some $\delta_1,\delta_2,\delta_3\in\mathbb{Z}_q$, we perform a case analysis based on the information $\mathcal{A}$ receives from the simulation. 
    % For completeness, 
    % we first summarize the exponent elements available to $\mathcal{A}$ in groups $\mathbb{G}_1$, $\mathbb{G}_2$, and $\mathbb{G}_T$.\vspace{1mm}

    % \noindent$-$ $\mathbb{G}_1$ elements: $1$, $\kappa$, $t_i$, $t_i\tau$, $t_ir$, 
    % % $t_is_2$, 
    % $x+\kappa\tau$,  $\kappa\lambda_i+t_is_3$, $t_i\tau_1s_3$, $(x+\kappa\tau_1)s_3$, $\frac{1}{b_1}(\lambda_i+t_ir')$, $\frac{1}{b_2}(\lambda_i+t_ir')$, $\frac{1}{b_1}(\kappa\psi_i+t_ir')$, $\frac{1}{b_2}(\kappa\psi_i+t_ir')$.\vspace{1mm}

    % \noindent$-$ $\mathbb{G}_2$ elements: 1, $\mu$, $b_1$, $b_2$, $r$, $r'$, $s_1$, $s_3$, $b_1\tau$, $b_2\tau$, $b_1s_2'$, $b_2s_2''$, $\alpha+\kappa r$,  $b_1\tau_1s_3'$, $b_2\tau_1s_3''$. \vspace{1mm}

    % \noindent$-$ $\mathbb{G}_T$ elements: 1, $\alpha$, $x\mu$, 
    % $\alpha s+x\mu s_3$.\vspace{1mm}
    
    For completeness, we list all possible queries into $\mathbb{G}_T$ using the bilinear map and group elements available to $\mathcal{A}$, as shown in Table \ref{Tab:FEME-privacy}. The violet elements in Table \ref{Tab:FEME-privacy} should be excluded for this anonymity proof, they are used in the confidentiality proof. $\mathcal{A}$ can compute arbitrary linear combinations of these terms. Below, we show that constructing these expressions is computationally infeasible, giving $\mathcal{A}$ only a negligible advantage in winning the anonymity game.

    %----------------------------------------------------

    (1) Consider how to construct $e(g_1,g_2)^{\delta_1(\kappa\lambda_i+t_{\pi(i)}s_3)}$ for some $\delta_1$. According to Table \ref{Tab:FEME-privacy}, $\mathcal{A}$ cannot combine $s_3=s_3'+s_3''$ via simple addition or subtraction within the elements of $\mathbb{G}_1$, $\mathbb{G}_2$, or $\mathbb{G}_T$. 
    
    To generate a term involving $\kappa\lambda_i$ is by pairing $\frac{1}{b_1} \cdot (\lambda_i + t_{\rho(i)}r')$ with $b_1\tau_1s_3'$ and pairing $\frac{1}{b_2} \cdot (\lambda_i + t_{\rho(i)}r')$ with $b_2\tau_1s_3''$, yielding $(\lambda_i + t_{\rho(i)}r')\tau_1s_3'$ and $(\lambda_i + t_{\rho(i)}r')\tau_1s_3''$, which combine to produce $(\lambda_i + t_{\rho(i)}r')\tau_1s_3$ in $\mathbb{G}_T$.
    
    To generate a term involving $t_is_3$ is by pairing $\frac{1}{b_1} \cdot (\kappa\psi_i + t_{\rho(i)}r')$ with $b_1\tau_1s_3'$ and pairing $\frac{1}{b_2} \cdot (\kappa\psi_i + t_{\rho(i)}r')$ with $b_2\tau_1s_3''$, yielding $(\kappa\psi_i + t_{\rho(i)}r')\tau_1s_3'$ and $(\kappa\psi_i + t_{\rho(i)}r')\tau_1s_3''$, which combine to produce $(\kappa\psi_i + t_{\rho(i)}r')\tau_1s_3$ in $\mathbb{G}_T$. 
    
    However, $\kappa\lambda_i+t_{\pi(i)}s_3$ cannot be derived through addition or subtraction of $(\lambda_i + t_{\rho(i)}r')\tau_1s_3$, $(\kappa\psi_i + t_{\rho(i)}r')\tau_1s_3$ with any other elements in $\mathbb{G}_1$, $\mathbb{G}_2$, or $\mathbb{G}_T$. Hence, constructing $\delta_1(\kappa\lambda_i+t_{\pi(i)}s_3)$ in $\mathbb{G}_T$ is impossible for $\mathcal{A}$.    \vspace{1mm}

    %----------------------------------------------------

    (2) Now, consider constructing $e(g_1,g_2)^{\delta_2 t_is_2}$ for some $\delta_2$. According to Table \ref{Tab:FEME-privacy}, $\mathcal{A}$ cannot combine $s_2=s_2'+s_2''$ via simple addition or subtraction within the elements of $\mathbb{G}_1$, $\mathbb{G}_2$, or $\mathbb{G}_T$. 
    
    The first way $\mathcal{A}$ could generate a term involving $t_is_2$ is by pairing $\frac{1}{b_1}\cdot(\lambda_i+t_{\rho(i)}r')$ with $b_1s_2'$ and pairing $\frac{1}{b_2}\cdot(\lambda_i+t_{\rho(i)}r')$ with $b_2s_2''$, producing $(\lambda_i+t_{\rho(i)}r')s_2'$ and $(\lambda_i+t_{\rho(i)}r')s_2''$, which combine to $(\lambda_i+t_{\rho(i)}r')(s_2'+s_2'')=\lambda_is_2+r't_{\rho(i)}s_2$ in $\mathbb{G}_T$.
    However, $t_is_2$ cannot be derived through addition or subtraction of $\lambda_is_2+r't_{\rho(i)}s_2$ with any other elements in $\mathbb{G}_1$, $\mathbb{G}_2$, or $\mathbb{G}_T$. Hence, constructing $\delta_2 t_is_2$ in $\mathbb{G}_T$ is impossible for $\mathcal{A}$.

    Another way $\mathcal{A}$ could generate a term involving $t_is_2$ is by pairing $\frac{1}{b_1}\cdot(\kappa\psi_i+t_{\rho(i)}r')$ with $b_1s_2'$ and pairing $\frac{1}{b_2}\cdot(\kappa\psi_i+t_{\rho(i)}r')$ with $b_2s_2''$, producing $(\kappa\psi_i+t_{\rho(i)}r')s_2'$ and $(\kappa\psi_i+t_{\rho(i)}r')s_2''$, which combine to $(\kappa\psi_i+t_{\rho(i)}r')(s_2'+s_2'')=\kappa\psi_is_2+r't_{\rho(i)}s_2$ in $\mathbb{G}_T$.
    However, $t_is_2$ cannot be derived through addition or subtraction of $\kappa\psi_is_2+r't_{\rho(i)}s_2$ with any other elements in $\mathbb{G}_1$, $\mathbb{G}_2$, or $\mathbb{G}_T$. Hence, constructing $\delta_2 t_is_2$ in $\mathbb{G}_T$ is impossible for $\mathcal{A}$.

    %----------------------------------------------------

    (3) Finally, consider constructing $e(g_1,g_2)^{\delta_3t_i\tau_1s_3}$ for some $\delta_3$. According to Table \ref{Tab:FEME-privacy}, $\mathcal{A}$ cannot combine $s_3=s_3'+s_3''$ via simple addition or subtraction within the elements of $\mathbb{G}_1$, $\mathbb{G}_2$, or $\mathbb{G}_T$. 
    
    The first way $\mathcal{A}$ could generate a term involving $t_i\tau_1s_3$ is by pairing $\frac{1}{b_1} \cdot (\lambda_i + t_{\rho(i)}r')$ with $b_1\tau_1s_3'$ and pairing $\frac{1}{b_2} \cdot (\lambda_i + t_{\rho(i)}r')$ with $b_2\tau_1s_3''$, yielding $(\lambda_i + t_{\rho(i)}r')\tau_1s_3'$ and $(\lambda_i + t_{\rho(i)}r')\tau_1s_3''$, which combine to produce $(\lambda_i + t_{\rho(i)}r')\tau_1s_3=\lambda_i\cdot\tau_1s_3 + r'\cdot t_i\tau_1s_3$ in $\mathbb{G}_T$.     
    However, $t_i\tau_1s_3$ cannot be derived through addition or subtraction of $\lambda_i\cdot\tau_1s_3 + r'\cdot t_{\rho(i)}\tau_1s_3$ with any other elements in $\mathbb{G}_1$, $\mathbb{G}_2$, or $\mathbb{G}_T$. Hence, constructing $\delta_3t_i\tau_1s_3$ in $\mathbb{G}_T$ is impossible for $\mathcal{A}$.
    
    Another way $\mathcal{A}$ could generate a term involving $t_i\tau_1s_3$ is by pairing $\frac{1}{b_1} \cdot (\kappa\psi_i + t_{\rho(i)}r')$ with $b_1\tau_1s_3'$ and pairing $\frac{1}{b_2} \cdot (\kappa\psi_i + t_{\rho(i)}r')$ with $b_2\tau_1s_3''$, yielding $(\kappa\psi_i + t_{\rho(i)}r')\tau_1s_3'$ and $(\kappa\psi_i + t_{\rho(i)}r')\tau_1s_3''$, which combine to produce $(\kappa\psi_i + t_{\rho(i)}r')\tau_1s_3$ in $\mathbb{G}_T$.     
    However, $t_i\tau_1s_3$ cannot be derived through addition or subtraction of $(\kappa\psi_i + t_{\rho(i)}r')\tau_1s_3$ with any other elements in $\mathbb{G}_1$, $\mathbb{G}_2$, or $\mathbb{G}_T$. Hence, constructing $\delta_3t_i\tau_1s_3$ in $\mathbb{G}_T$ is impossible for $\mathcal{A}$.
    
    \textbf{Analysis of Simulation Failure}. Same as that in the proof for Theorem \ref{theo:FEME-privacy}.

    In conclusion, $\mathcal{A}$ only holds a negligible advantage in the modified game, implying that it also has a negligible advantage in the anonymity game. 
    
    This completes the proof of Theorem \ref{theo:FEME-anon}. $\hfill\blacksquare$

\subsection{Proof of Theorem \ref{theo:FEME-authn} (FEME: Authenticity)}
\label{app:proof-FEME-authn}

% \noindent\textit{Proof Sketch}. 
To prove the authenticity of FEME, we show that an adversary cannot create a valid ciphertext that passes decryption without the proper attribute encryption keys. The challenger sets up the system, providing the adversary with master public keys and simulating key generation oracles. The adversary can query these oracles to obtain encryption and decryption keys for chosen attribute sets and policies, except for those matching the final target policy. The proof relies on a contradiction: assuming the adversary produces a valid ciphertext under conditions not allowed by the oracles, we demonstrate this would imply the adversary must have queried the sender's encryption key with attributes matching the target policy, violating the security model's constraints. This involves showing that certain polynomial relations between group elements must hold, which leads to the contradiction, proving that any valid forgery is computationally infeasible.
% The concrete proof is deferred to the full version\cite{PriSrv+} due to length limitation.

\noindent\textit{Proof}. The interaction between the adversary $\mathcal{A}$ and challenger $\mathcal{C}$ proceeds as follows.

    \noindent{\scriptsize$\bullet$} \textbf{\textit{Setup}}. The challenger $\mathcal{C}$ chooses $\alpha,x,\mu,b_1,b_2,\kappa\stackrel{\$}{\leftarrow}\mathbb{Z}_p^*$, and calculates $Z=e(g_1,g_2)^{\alpha}$, $Y=e(g_1,g_2)^{x\mu}$, $\delta_0=g_2^{\mu}$, $\delta_1=g_2^{b_1}$, $\delta_2=g_2^{b_2}$ and $h=g_1^{\kappa}$. $\mathcal{C}$ sends the master public key $\mpk=(Z,Y,h,\delta_0,\delta_1,\delta_2)$ to $\mathcal{A}$.

    \noindent{\scriptsize$\bullet$} \textbf{\textit{Phase 1}}. In phase 1, $\mathcal{A}$ can make oracle queries to the random oracle and a key generation oracle as follows.

    \noindent$-$ \textit{Random oracle} $(\mathcal{O}_{H})$. Same as defined above.

    \noindent$-$ \textit{Attribute encryption key generation oracle} $(\mathcal{O}_{\EKGen})$. When $\mathcal{A}$ makes a encryption key query for an attribute set $\mathcal{S}_{\textsf{snd}}$, $\mathcal{C}$ picks $\tau_j\stackrel{\$}{\leftarrow}\mathbb{Z}_p^*$ for the $j$-th query. Then, $\mathcal{C}$ generates the attribute encryption key as     $$\textsf{ek}_{1,i}=h_i^{\tau_j},~~\textsf{ek}_2=\delta_1^{\tau_j},~~\textsf{ek}_3=\delta_2^{\tau_j},~~\textsf{ek}_4=g_1^{x}\cdot h^{\tau_j}.$$
    Then, $\mathcal{C}$ sends to $\mathcal{A}$ the attribute encryption key $\textsf{EK}_{\mathcal{S}_{\textsf{snd}}}=(\{n_i\}_{i\in[\ell_1]},\{\textsf{ek}_{1,i}\}_{i\in[\ell_1]},\textsf{ek}_2,\textsf{ek}_3,\textsf{ek}_4)$.   

    % \noindent$-$ \textit{Attribute decryption key generation oracle} $(\mathcal{O}_{\DKGen})$. Same as defined in the proof of Theorem \ref{theo:FEME-privacy}.
    
    % \noindent$-$ \textit{Policy decryption key generation oracle} $(\mathcal{O}_{\PolGen})$. Same as defined in the proof of Theorem \ref{theo:FEME-privacy}.

    \noindent$-$ \textit{Attribute decryption key generation oracle} $(\mathcal{O}_{\DKGen})$. When $\mathcal{A}$ makes a key query for an attribute set $\mathcal{S}_{\textsf{rcv}}$, $\mathcal{C}$ picks $r\stackrel{\$}{\leftarrow}\mathbb{Z}_p^*$. Then, $\mathcal{C}$ generates the attribute decryption key as  $$\textsf{dk}_1=g_1^{\alpha}\cdot h^{r},~~\textsf{dk}_{2,i}=h_i^{r},~~\textsf{dk}_3=g_2^{r}.$$
    Then, $\mathcal{C}$ sends to $\mathcal{A}$ the attribute decryption key $\textsf{DK}_{\mathcal{S}_{\textsf{rcv}}}=(\{n_i\}_{i\in[\ell_2]},\textsf{dk}_1,\{\textsf{dk}_{2,i}\}_{i\in[\ell_2]},\textsf{dk}_3)$.

    \noindent$-$ \textit{Policy decryption key generation oracle} $(\mathcal{O}_{\PolGen})$. When $\mathcal{A}$ makes a key query for a policy $\mathbb{A}_{\textsf{rcv}}=(\textbf{A},\rho,\{\Psi_{\rho(i)}\}_{i\in[m_2]})$, $\mathcal{C}$ picks $r'\stackrel{\$}{\leftarrow}\mathbb{Z}_p^*$ and a vector $\textbf{y}\stackrel{\$}{\leftarrow}\mathbb{Z}_p^{m_2-1}$. Let $\lambda_i=\textbf{A}_i(\alpha||\textbf{y})^{\top}$ and $\psi_i=\textbf{A}_i(\mu||\textbf{y})^{\top}$. Note that the $\lambda_i$ (resp. $\psi_i$) are chosen uniformly and independently at random from $\mathbb{Z}_p^*$ subject to the random distribution of $\alpha$ (resp. $\mu$) and $\textbf{y}$. Then, $\mathcal{C}$ generates the policy decryption key as   $$\textsf{sk}_1=g_2^{r'},~~\textsf{sk}_{2,i}=(g_1^{\lambda_i}\cdot h_{\rho(i)}^{r'})^{\frac{1}{b_1}},~~\textsf{sk}_{3,i}=(g_1^{\lambda_i}\cdot h_{\rho(i)}^{r'})^{\frac{1}{b_2}},$$    $$\textsf{sk}_{4,i}=(h^{\psi_i}\cdot h_{\rho(i)}^{r'})^{\frac{1}{b_1}},~~\textsf{sk}_{5,i}=(h^{\psi_i}\cdot h_{\rho(i)}^{r'})^{\frac{1}{b_2}}.$$
    Then, $\mathcal{C}$ sends to $\mathcal{A}$ the policy decryption key $\textsf{SK}_{\mathbb{A}_{\textsf{rcv}}}=((\textbf{A},\rho,\{n_{\rho(i)}\}_{i\in[m_2]}),\textsf{sk}_1,\{\textsf{sk}_{2,i},\textsf{sk}_{3,i},\textsf{sk}_{4,i},\textsf{sk}_{5,i}\}_{i\in[m_2]})$.   

    \noindent{\scriptsize$\bullet$} \textbf{\textit{Forgery}}. The attacker $\mathcal{A}$ sends $(\CT_{\snd}^*,\mathcal{S}_{\snd}^*,\mathbb{A}_{\rcv}^*)$ to $\mathcal{C}$, where $\mathcal{S}_{\snd}^*\models\mathbb{A}_{\rcv}^*$. The restriction is that the attribute encryption key of $\mathcal{S}_{\snd}^*$ has not been queried in $\mathcal{O}_{\EKGen}$, and the policy decryption key of $\mathbb{A}_{\rcv}^*$ and any other $\mathbb{A}_{\rcv}$ satisfying $\mathcal{S}_{\snd}^*\models\mathbb{A}_{\rcv}$ have not been queried in $\mathcal{O}_{\PolGen}$.
    We must prove that under the restrictions, the adversary $\mathcal{A}$ cannot forge a valid ciphertext $\CT_{\snd}^*$. 
    
    We use proof by contradiction to assume that $\CT_{\snd}^*$ is a valid ciphertext. The adversary $\mathcal{A}$ has access to the group elements provided in the master public key $\textsf{mpk}=(Z,Y,h,\delta_0,\delta_1,\delta_2)$, $h_i=H(u_i)$ obtained from $\mathcal{O}_{H}$, and $\EK_{\mathcal{S}_{\snd}}$ obtained from $\mathcal{O}_{\EKGen}$, as they are the input of $\Enc$ for $\CT_{\snd}^*$ generation. Let $\textsf{EK}_{\mathcal{S}_{\textsf{snd}}}^{(j)}=(\{n_i^{(j)}\}_{i\in[\ell_1]},\{\textsf{ek}_{1,i}^{(j)}\}_{i\in[\ell_1]},\textsf{ek}_2^{(j)},\textsf{ek}_3^{(j)},\textsf{ek}_4^{(j)})$ be the output of the $j$-th query to $\mathcal{O}_{\EKGen}$. Denote $q_e$ as the query time to $\mathcal{O}_{\EKGen}$. In GGM, the only way for the adversary to generate new group elements is to use the existing exponent elements available to $\mathcal{A}$ in groups $\mathbb{G}_1$, $\mathbb{G}_2$ and $\mathbb{G}_T$. This means that there are known scalars $\big(a,b,b',b'',c,c',c'',\{d_i\}_{i\in[m_1]},e,\{f_{1,j},f_{2,j},f_{3,j},f_{4,j}\}_{j\in[q_e]}\big)$ such that:
    $$V=Z^{a+b}\cdot Y^{c},~\ct_0=\phi(\msg)\oplus\hat{H}(V),$$ $$\textsf{ct}_1^*=g_2^{a},~\textsf{ct}_2^*=g_2^{c},~\textsf{ct}_{3,i}^*=h^{d_i}\cdot (h_{\pi(i)})^{c},$$    $$\textsf{ct}_{4,1}^*=\delta_1^{b'},~\textsf{ct}_{4,2}^*=\delta_2^{b''},~\textsf{ct}_{5,i}^*=h_i^{b},$$
    \begin{eqnarray*}               \textsf{ct}_{6,i}^*&=&\prod\nolimits_{j=1}^{q_e}(\ek_{1,i}^{(j)})^{f_{1,j}\cdot c}\cdot h_i^{e\cdot c},\\        \textsf{ct}_7^*&=&\prod\nolimits_{j=1}^{q_e}(\ek_2^{(j)})^{f_{2,j}\cdot c'}\cdot\delta_1^{e\cdot c'}, \\    \textsf{ct}_8^*&=&\prod\nolimits_{j=1}^{q_e}(\ek_3^{(j)})^{f_{3,j}\cdot c''}\cdot\delta_2^{e\cdot c''},\\    \textsf{ct}_9^*&=&\prod\nolimits_{j=1}^{q_e}(\ek_4^{(j)})^{f_{4,j}\cdot c}\cdot h^{e\cdot c}.
    \end{eqnarray*}
    In the ciphertext $\CT_{\snd}^*$ construction, it implicit sets that $a=s_1$, $b=s_2$, $c=s_3$, $b'=s_2'$, $b''=s_2''$, $c'=s_3'$, $c''=s_3''$, where $b=b'+b''$ and $c=c'+c''$. It also implies that $d_i=\textbf{M}_{i}(s_1||\textbf{v})^{\top}$ and $e=\tau'$.

    Recall that $\CT_{\snd}^*$ is a valid ciphertext on $\mathcal{S}_{\snd}^*$ and $\mathcal{S}_{\snd}^*\models\mathbb{A}_{\rcv}^*$. According to the correctness proof (shown in the full version \cite{PriSrv+}), we know that:
    $$\frac{e(\prod_{i\in I_2}(\sk_{4,\rho(i)})^{\omega_i},\ct_7)e(\prod_{i\in I_2}(\sk_{5,\rho(i)})^{\omega_i},\ct_8)}{e(\ct_9,\delta_0)e(\prod\nolimits_{i\in I_2}(\ct_{6,\rho(i)})^{\omega_i},\sk_1)}=Y^{-s_3},$$
    where $\sum_{i\in I_2}(\A_i(\mu||\sfy)^{\top}\cdot\omega_i)=\mu$.

    Based on this equation, we derive that
    {\footnotesize
    \begin{eqnarray*}
        &&\frac{e(\prod_{i\in I_2}(\sk_{4,\rho(i)})^{\omega_i},\ct_7)e(\prod_{i\in I_2}(\sk_{5,\rho(i)})^{\omega_i},\ct_8)}{e(\ct_9,\delta_0)e(\prod\nolimits_{i\in I_2}(\ct_{6,\rho(i)})^{\omega_i},\sk_1)}\\
        &=&\frac{e(\prod_{i\in I_2}((h^{\psi_i}\cdot h_{\rho(i)}^{r'})^{\frac{1}{b_1}})^{\omega_i},\prod\nolimits_{j=1}^{q_e}(\ek_2^{(j)})^{f_{2,j}c'}\cdot\delta_1^{e\cdot c'})}{e(\prod\nolimits_{j=1}^{q_e}(\ek_4^{(j)})^{f_{4,j}c}\cdot h^{e\cdot c},\delta_0)}\\
        &&\cdot \frac{e(\prod_{i\in I_2}((h^{\psi_i}\cdot h_{\rho(i)}^{r'})^{\frac{1}{b_2}})^{\omega_i},\prod\nolimits_{j=1}^{q_e}(\ek_3^{(j)})^{f_{3,j}c''}\cdot\delta_2^{e\cdot c''})}{e(\prod\nolimits_{i\in I_2}(\prod\nolimits_{j=1}^{q_e}(\ek_{1,\rho(i)}^{(j)})^{f_{1,j}c}\cdot h_{\rho(i)}^{e\cdot c})^{\omega_i},g_2^{r'})}\\         
        &=&\frac{e(\prod_{i\in I_2}(h^{\psi_i}\cdot h_{\rho(i)}^{r'})^{\frac{\omega_i}{b_1}},\prod\nolimits_{j=1}^{q_e}(\delta_1^{\tau_j})^{f_{2,j}c'}\cdot\delta_1^{e\cdot c'})}{e(\prod\nolimits_{j=1}^{q_e}(g_1^{x}\cdot h^{\tau_j})^{f_{4,j}c}\cdot h^{e\cdot c},g_2^{\mu})}\\
        &&\cdot \frac{e(\prod_{i\in I_2}(h^{\psi_i}\cdot h_{\rho(i)}^{r'})^{\frac{\omega_i}{b_2}},\prod\nolimits_{j=1}^{q_e}(\delta_2^{\tau_j})^{f_{3,j}c''}\cdot\delta_2^{e\cdot c''})}{e(\prod\nolimits_{i\in I_2}(\prod\nolimits_{j=1}^{q_e}(h_{\rho(i)}^{\tau_j})^{f_{1,j}c}\cdot h_{\rho(i)}^{e\cdot c})^{\omega_i},g_2^{r'})}\\  
        &=&\frac{e\big(g_1^{\frac{\kappa}{b_1}\sum_{i\in I_2}\psi_i\omega_i}\cdot g_1^{\frac{r'}{b_1}\sum_{i\in I_2}\omega_it_{\rho(i)}},g_2^{c'b_1(\sum\nolimits_{j=1}^{q_e}f_{2,j}\tau_j+e)}\big)}{e(g_1^{cx\cdot\sum\nolimits_{j=1}^{q_e}f_{4,j}}\cdot g_1^{c(\sum\nolimits_{j=1}^{q_e}f_{4,j}\tau_j+e)\kappa},g_2^{\mu})}\\
        &&\cdot \frac{e\big(g_1^{\frac{\kappa}{b_2}\sum_{i\in I_2}\psi_i\omega_i}\cdot g_1^{\frac{r'}{b_2}\sum_{i\in I_2}\omega_it_{\rho(i)}},g_2^{c''b_2(\sum\nolimits_{j=1}^{q_e}f_{3,j}\tau_j+e)}\big)}{e(g_1^{c(\sum\nolimits_{j=1}^{q_e}f_{1,j}\tau_j+e)\sum\nolimits_{i\in I_2}\omega_it_{\rho(i)}},g_2^{r'})}\\
        &=&e(g_1,g_2)^{-x\mu\cdot c}=Y^{-c},
    \end{eqnarray*}}
    where $\sum_{i\in I_2}\psi_i\cdot\omega_i=\mu$, and $\psi_i=\textbf{A}_i(\mu||\textbf{y})^{\top}$.

    Then, we obtain the following polynomial relation:
    {\footnotesize
    \begin{eqnarray*}
        &&\big(\frac{\kappa}{b_1}\sum_{i\in I_2}\psi_i\omega_i+\frac{r'}{b_1}\sum_{i\in I_2}\omega_it_{\rho(i)}\big)\cdot \big(c'b_1(\sum\nolimits_{j=1}^{q_e}f_{2,j}\tau_j+e)\big)\\
        &&+\big(\frac{\kappa}{b_2}\sum_{i\in I_2}\psi_i\omega_i+\frac{r'}{b_2}\sum_{i\in I_2}\omega_it_{\rho(i)}\big)\cdot\big(c''b_2(\sum\nolimits_{j=1}^{q_e}f_{3,j}\tau_j+e)\big)\\
        &&- \big(cx\cdot \sum\nolimits_{j=1}^{q_e}f_{4,j}+c(\sum\nolimits_{j=1}^{q_e}f_{4,j}\tau_j+e)\kappa\big)\cdot\mu\\
        &&-\big(c(\sum\nolimits_{j=1}^{q_e}f_{1,j}\tau_j+e)\sum\nolimits_{i\in I_2}\omega_it_{\rho(i)}\big)\cdot r'\\
        &=&-x\mu c
    \end{eqnarray*}}

    Let us rearrange the left side of this equation as follows.
    {\footnotesize
    \begin{eqnarray*}
        &&\big(\frac{\kappa}{b_1}\sum_{i\in I_2}\psi_i\omega_i+\frac{r'}{b_1}\sum_{i\in I_2}\omega_it_{\rho(i)}\big)\cdot \big(c'b_1(\sum\nolimits_{j=1}^{q_e}f_{2,j}\tau_j+e)\big)\\
        &&+\big(\frac{\kappa}{b_2}\sum_{i\in I_2}\psi_i\omega_i+\frac{r'}{b_2}\sum_{i\in I_2}\omega_it_{\rho(i)}\big)\cdot\big(c''b_2(\sum\nolimits_{j=1}^{q_e}f_{3,j}\tau_j+e)\big)\\
        &&- \big(cx\cdot \sum\nolimits_{j=1}^{q_e}f_{4,j}+c(\sum\nolimits_{j=1}^{q_e}f_{4,j}\tau_j+e)\kappa\big)\cdot\mu\\
        &&-\big(c(\sum\nolimits_{j=1}^{q_e}f_{1,j}\tau_j+e)\sum\nolimits_{i\in I_2}\omega_it_{\rho(i)}\big)\cdot r'\\
        &=&\big(\kappa c'\sum\nolimits_{j=1}^{q_e}f_{2,j}\tau_j(\sum_{i\in I_2}\psi_i\omega_i)\\        &&~~~~+r'c'\sum\nolimits_{j=1}^{q_e}f_{2,j}\tau_j(\sum_{i\in I_2}\omega_it_{\rho(i)})\\
        &&~~~~+\kappa c'e\sum_{i\in I_2}\psi_i\omega_i+r'c'e\sum_{i\in I_2}\omega_it_{\rho(i)}\big)\\
        &&+\big(\kappa c''\sum\nolimits_{j=1}^{q_e}f_{3,j}\tau_j(\sum_{i\in I_2}\psi_i\omega_i)\\        &&~~~~+r'c''\sum\nolimits_{j=1}^{q_e}f_{3,j}\tau_j(\sum_{i\in I_2}\omega_it_{\rho(i)})\\
        &&~~~~+\kappa c''e\sum_{i\in I_2}\psi_i\omega_i+r'c''e\sum_{i\in I_2}\omega_it_{\rho(i)}\big)\\
        &&- \big(\sum\nolimits_{j=1}^{q_e}f_{4,j}\cdot x\mu c+c(\sum\nolimits_{j=1}^{q_e}f_{4,j}\tau_j+e)\kappa\cdot\mu\big)\\
        &&-\big(cr'(\sum\nolimits_{j=1}^{q_e}f_{1,j}\tau_j+e)\sum\nolimits_{i\in I_2}\omega_it_{\rho(i)}\big) \\
        &=&\kappa (ce+c'\sum\nolimits_{j=1}^{q_e}f_{2,j}\tau_j+c''\sum\nolimits_{j=1}^{q_e}f_{3,j}\tau_j)(\sum_{i\in I_2}\psi_i\omega_i)\\        &&+r'(ce+c'\sum\nolimits_{j=1}^{q_e}f_{2,j}\tau_j+c''\sum\nolimits_{j=1}^{q_e}f_{3,j}\tau_j)(\sum_{i\in I_2}\omega_it_{\rho(i)})\\
        &&- \big(\sum\nolimits_{j=1}^{q_e}f_{4,j}\cdot x\mu c+\kappa\cdot c\mu(\sum\nolimits_{j=1}^{q_e}f_{4,j}\tau_j+e)\big)\\
        &&-\big(r'(c\sum\nolimits_{j=1}^{q_e}f_{1,j}\tau_j+ce)\sum\nolimits_{i\in I_2}\omega_it_{\rho(i)}\big) \\
        &=&\kappa [(ce+c'\sum\nolimits_{j=1}^{q_e}f_{2,j}\tau_j+c''\sum\nolimits_{j=1}^{q_e}f_{3,j}\tau_j)\mu\\ 
        &&~~~~~~~~-c\mu(\sum\nolimits_{j=1}^{q_e}f_{4,j}\tau_j+e)\big)]\\
        &&+r'(-c\sum\limits_{j=1}^{q_e}f_{1,j}\tau_j+c'\sum\limits_{j=1}^{q_e}f_{2,j}\tau_j+c''\sum\limits_{j=1}^{q_e}f_{3,j}\tau_j)\\
        &&~~~~~\cdot(\sum_{i\in I_2}\omega_it_{\rho(i)})\\
        &&- \sum\nolimits_{j=1}^{q_e}f_{4,j}\cdot x\mu c\\
        &=& (-c\sum\nolimits_{j=1}^{q_e}f_{4,j}\tau_j+c'\sum\nolimits_{j=1}^{q_e}f_{2,j}\tau_j+c''\sum\nolimits_{j=1}^{q_e}f_{3,j}\tau_j)\kappa\mu\\ 
        &&+r'(-c\sum\limits_{j=1}^{q_e}f_{1,j}\tau_j+c'\sum\limits_{j=1}^{q_e}f_{2,j}\tau_j+c''\sum\limits_{j=1}^{q_e}f_{3,j}\tau_j)\\
        &&~~~~~\cdot(\sum_{i\in I_2}\omega_it_{\rho(i)})\\
        &&- \sum\nolimits_{j=1}^{q_e}f_{4,j}\cdot x\mu c
    \end{eqnarray*}}

    Since the right side of the equation is $-x\mu c$, we must have 
    \begin{eqnarray}
        -c\sum\limits_{j=1}^{q_e}f_{4,j}\tau_j+c'\sum\limits_{j=1}^{q_e}f_{2,j}\tau_j+c''\sum\limits_{j=1}^{q_e}f_{3,j}\tau_j=0,\\
        -c\sum\limits_{j=1}^{q_e}f_{1,j}\tau_j+c'\sum\limits_{j=1}^{q_e}f_{2,j}\tau_j+c''\sum\limits_{j=1}^{q_e}f_{3,j}\tau_j=0,\\
        \sum\nolimits_{j=1}^{q_e}f_{4,j}=1.
    \end{eqnarray}

    Combining equations (1)-(3), we derive that $$\sum\nolimits_{j=1}^{q_e}f_{1,j}=\sum\nolimits_{j=1}^{q_e}f_{4,j}=1.$$
   
    We arrange equation (1) as
    $$\sum\nolimits_{j=1}^{q_e}(-cf_{4,j}+c'f_{2,j}+c''f_{3,j})\tau_j=0.$$
    According to linear algebraic theory, we deduce that 
    $-cf_{4,j}+c'f_{2,j}+c''f_{3,j}=0$ for all $j\in[q_e]$.     
    As $c=c'+c''$, we have 
    \begin{eqnarray*}
        -(c'+c'')f_{4,j}+c'f_{2,j}+c''f_{3,j}&=&0\\
        c'(f_{2,j}-f_{4,j})+c''(f_{3,j}-f_{4,j})&=&0
    \end{eqnarray*}
    Then, we deduce that $f_{2,j}=f_{3,j}=f_{4,j}$ for all $j\in[q_e]$.
    Similarly, we deduce that $f_{1,j}=f_{2,j}=f_{3,j}$ for all $j\in[q_e]$ from equation (2). 
    Then, $f_{1,j}=f_{2,j}=f_{3,j}=f_{4,j}$ for all $j\in[q_e]$.
    Under this situation, we have  $$\textsf{EK}_{\mathcal{S}_{\textsf{snd}}}^{(j)}=(\{n_i^{(j)}\}_{i\in[\ell_1]},\{\textsf{ek}_{1,i}^{(j)}\}_{i\in[\ell_1]},\textsf{ek}_2^{(j)},\textsf{ek}_3^{(j)},\textsf{ek}_4^{(j)})$$ are created for the same $\mathcal{S}_{\textsf{snd}}^*$ for all $j\in[q_e]$, where $\mathcal{S}_{\textsf{snd}}^*\models\mathbb{A}_{\textsf{rcv}}^*$ (the deduction of this step is deterred to the end of the proof).     
    It indicates that the ciphertext $\CT_{\snd}^*$ must be constructed as    
    \begin{eqnarray*}                   \textsf{ct}_{6,i}^*&=&\prod\limits_{j=1}^{q_e}(H(u_i)^{\tau_j})^{f_{j}\cdot c}\cdot h_i^{e\cdot c}=(H(u_i)^{\tau_0}\cdot h_i^e)^c,\\    \textsf{ct}_7^*&=&\prod\limits_{j=1}^{q_e}(\delta_1^{\tau_j})^{f_{j}\cdot c'}\cdot\delta_1^{e\cdot c'}=\big(\delta_1^{\tau_0}\cdot\delta_1^{e}\big)^{c'},\\
    \textsf{ct}_8^*&=&\prod\limits_{j=1}^{q_e}(\delta_2^{\tau_j})^{f_{j}\cdot c''}\cdot\delta_2^{e\cdot c''}=\big(\delta_2^{\tau_0}\cdot\delta_2^{e}\big)^{c''},\\
    \textsf{ct}_9^*&=&\prod\limits_{j=1}^{q_e}(g_1^xh^{\tau_j})^{f_{j}\cdot c}\cdot h^{e\cdot c}=\big((g_1^xh^{\tau_0})\cdot h^{e}\big)^{c},
    \end{eqnarray*}
    where $f_j=f_{1,j}=f_{2,j}=f_{3,j}=f_{4,j}$ and $\tau_0=\sum\nolimits_{j=1}^{q_e}\tau_jf_j$. 
    
    We can deduce that 
    $$\textsf{ct}_{6,i}^*=(\ek_{1,i}^*\cdot h_i^{e})^c,~\textsf{ct}_7^*=(\ek_2^*\cdot\delta_1^{e})^{c'},$$
    $$\textsf{ct}_8^*=(\ek_3^*\cdot\delta_2^{e})^{c''},~\textsf{ct}_9^*=(\ek_4^*\cdot h^{e})^c,$$
    where $\ek_{1,i}^*= H(u_i)^{\tau_0}$, $\ek_2^*=\delta_1^{\tau_0}$, $\ek_3^*=\delta_2^{\tau_0}$, $\ek_4^*=g_1^xh^{\tau_0}$.    
    Since the scalers $\big(c,c',c'',e)$ are selected by the adversary, $\mathcal{A}$ can easily calculate $\EK_{\snd}^*=(\{n_i^*\}_{i\in[\ell_1]},\{\textsf{ek}_{1,i}^*\}_{i\in[\ell_1]},\textsf{ek}_2^*,\textsf{ek}_3^*,\textsf{ek}_4^*)$ for $\mathcal{S}_{\textsf{snd}}^*$ and $\mathcal{S}_{\textsf{snd}}^*\models\mathbb{A}_{\textsf{rcv}}^*$, which contradicts with the constraint.

    Therefore, the adversary cannot symbolically produce a valid forgery, and the \textit{authenticity} of FEME is proved. $\hfill\qedsymbol$\\
    
    In the following, we proved that $\textsf{EK}_{\mathcal{S}_{\textsf{snd}}}^{(j)}$ are created for the same $\mathcal{S}_{\textsf{snd}}^*$. 
    % based on $f_{1,j}=f_{2,j}=f_{3,j}=f_{4,j}=f_j$ and $\sum\nolimits_{j=1}^{q_e}f_j=1$, for all $j\in[q_e]$.

    \noindent(1) Let $q_e=1$. We have $f_{1,j}=f_{2,j}=f_{3,j}=f_{4,j}=f_j=1$ and 
    $$\textsf{ct}_{6,i}^*=(\ek_{1,i})^{c}\cdot h_i^{e\cdot c},~\textsf{ct}_7^*=(\ek_2)^{c'}\cdot\delta_1^{e\cdot c'},$$
    $$\textsf{ct}_8^*=(\ek_3)^{c''}\cdot\delta_2^{e\cdot c''},~\textsf{ct}_9^*=(\ek_4)^{c}\cdot h^{e\cdot c}.$$
    If the challenge ciphertext $\CT_{\snd}^*$ is valid, the adversary can directly calculate $\EK_{\snd}^*$ for $\mathcal{S}_{\textsf{snd}}^*$ and $\mathcal{S}_{\textsf{snd}}^*\models\mathbb{A}_{\textsf{rcv}}^*$, which contradicts with the constraint.
    \vspace{1.5mm}

    \noindent(2) Let $q_e=2$. We have $f_{1,j}=f_{2,j}=f_{3,j}=f_{4,j}=f_j$ and $\sum\nolimits_{j=1}^{2}f_j=f_1+f_2=1$. \vspace{1.5mm}
    
    \noindent(2.1) We firstly prove that $\mathcal{S}_{\snd}^{(1)}=\{u_i^{(1)}\}_{i\in[\ell_1^{(1)}]}$ and $\mathcal{S}_{\snd}^{(2)}=\{u_i^{(2)}\}_{i\in[\ell_1^{(2)}]}$ must have the same number of attributes, i.e., $\ell_1^{(1)}=\ell_1^{(2)}$. 
    
    Suppose that $\mathcal{S}_{\snd}^{(1)}$ and $\mathcal{S}_{\snd}^{(2)}$ has different number of attributes and $\ell_1^{(1)}>\ell_1^{(2)}$. We have 
    $$\textsf{ct}_{6,i}^*=\big(H(u_i^{(1)})^{\tau_1f_1}\cdot h_i^e\big)^c,\text{~for~}i=\ell_1^{(1)}.$$
    Therefore, $\tau_0=\tau_1f_1$. Since $\tau_0=\sum\nolimits_{j=1}^{2}\tau_jf_j$, it can be deduced that $f_2=0$. Since $f_1+f_2=1$, we have $f_1=1$. Then,
    $$\textsf{ct}_{6,i}^*=(\ek_{1,i}^*\cdot h_i^{e})^c,~\textsf{ct}_7^*=(\ek_2^*\cdot\delta_1^{e})^{c'},$$
    $$\textsf{ct}_8^*=(\ek_3^*\cdot\delta_2^{e})^{c''},~\textsf{ct}_9^*=(\ek_4^*\cdot h^{e})^c.$$
    $\mathcal{A}$ can easily calculate $\EK_{\snd}^*$ for $\mathcal{S}_{\textsf{snd}}^*$ and $\mathcal{S}_{\textsf{snd}}^*\models\mathbb{A}_{\textsf{rcv}}^*$, which contradicts with the constraint.\vspace{1.5mm}
    
    \noindent(2.2) Next, we prove that $\mathcal{S}_{\textsf{snd}}^{(1)}=\mathcal{S}_{\textsf{snd}}^{(2)}$.
    If $\mathcal{S}_{\textsf{snd}}^{(1)}\neq\mathcal{S}_{\textsf{snd}}^{(2)}$, there must exist some certain attribute $u_i^{(1)}\neq u_i^{(2)}$, where $i\in[\ell_1^{(2)}]$. Put $u_i^{(1)}$ and $u_i^{(2)}$ into the random oracle to get the hash values $h_i^{(1)}=H(u_i^{(1)})=g_1^{t_i^{(1)}}$ and $h_i^{(2)}=H(u_i^{(2)})=g_1^{t_i^{(2)}}$, where $t_i^{(1)},t_i^{(2)}\stackrel{\$}{\leftarrow}\mathbb{Z}_p^*$ and $t_i^{(1)}\neq t_i^{(2)}$.

    The ciphertext $\CT_{\snd}^*$ is constructed as    
    \begin{eqnarray*}                   \textsf{ct}_{6,i}^*&=&\prod\limits_{j=1}^{2}(H(u_i^{(j)})^{\tau_j})^{f_{j}\cdot c}\cdot h_i^{e\cdot c},\\    \textsf{ct}_7^*&=&\prod\limits_{j=1}^{2}(\delta_1^{\tau_j})^{f_{j}\cdot c'}\cdot\delta_1^{e\cdot c'}=\big(\delta_1^{\tau_0}\cdot\delta_1^{e}\big)^{c'},\\
    \textsf{ct}_8^*&=&\prod\limits_{j=1}^{2}(\delta_2^{\tau_j})^{f_{j}\cdot c''}\cdot\delta_2^{e\cdot c''}=\big(\delta_2^{\tau_0}\cdot\delta_2^{e}\big)^{c''},\\
    \textsf{ct}_9^*&=&\prod\limits_{j=1}^{2}(g_1^xh^{\tau_j})^{f_{j}\cdot c}\cdot h^{e\cdot c}=\big((g_1^xh^{\tau_0})\cdot h^{e}\big)^{c},
    \end{eqnarray*}
    where $\tau_0=\sum\nolimits_{j=1}^{2}\tau_jf_j$.

    If the ciphertext $\CT_{\snd}^*$ is valid, it must have
    $$\textsf{ct}_{6,i}^*=\prod\limits_{j=1}^{2}(H(u_i^{(j)})^{\tau_j})^{f_{j}\cdot c}\cdot h_i^{e\cdot c}=(H(\bar{u}_i)^{\tau_0}\cdot h_i^e)^c$$
    for some attribute $\bar{u}_i$. Let $\bar{h}_i=H(\bar{u}_i)=g_1^{\bar{t}_i}$, where $\bar{t}_i\stackrel{\$}{\leftarrow}\mathbb{Z}_p^*$. We have
    \begin{eqnarray*}
        \textsf{ct}_{6,i}^*&=&\prod\nolimits_{j=1}^{2}(H(u_i^{(j)})^{\tau_j})^{f_{j}\cdot c}\cdot h_i^{e\cdot c}\\
        &=&\big(g_1^{\sum\nolimits_{j=1}^{2}t^{(j)}\cdot \tau_jf_{j}}\cdot h_i^{e}\big)^c\\
        &=&(g^{\bar{t}_i\tau_0}\cdot h_i^e)^c=(H(\bar{u}_i)^{\tau_0}\cdot h_i^e)^c.
    \end{eqnarray*}
    Since $\tau_0=\sum\nolimits_{j=1}^{2}\tau_jf_j$ and $\sum\nolimits_{j=1}^{2}f_j=f_1+f_2=1$, we can deduce that 
    \begin{eqnarray*}
        &&\sum\nolimits_{j=1}^{2}t_i^{(j)}\cdot f_{j}\tau_j=\bar{t}_i\tau_0\\        &\Rightarrow&t_i^{(1)}\tau_1f_1+t_i^{(2)}\tau_2f_2=\bar{t}_i(\tau_1f_1+\tau_2f_2)\\        
        &\Rightarrow&(t_i^{(1)}-\bar{t}_i)f_1\tau_1+(t_i^{(2)}-\bar{t}_i)f_2\tau_2=0.       
        % &\Rightarrow&(t_i^{(1)}-\bar{t}_i)\tau_1f_1+(t^{(2)}-\bar{t}_i)\tau_2(1-f_1)=0\\           
        % &\Rightarrow&[(t_i^{(1)}-\bar{t})\tau_1-(t_i^{(2)}-\bar{t}_i)\tau_2]f_1= (\bar{t}_i-t_i^{(2)})\tau_2\\
    \end{eqnarray*}

    As $\tau_1$, $\tau_2$ are random numbers in $\mathbb{Z}_p^*$, we have $(t_i^{(1)}-\bar{t}_i)f_1=0$ and $(t_i^{(2)}-\bar{t}_i)f_2=0$. Since $f_1+f_2=1$, we have $t_i^{(1)}=\bar{t}_i$ and $f_1=1$, $f_2=0$, or $f_1=0$, $f_2=1$ and $t_i^{(2)}=\bar{t}_i$.
    In either way, if the challenge ciphertext $\CT_{\snd}^*$ is valid, the adversary can directly calculate $\EK_{\snd}^*$ for $\mathcal{S}_{\textsf{snd}}^*$ and $\mathcal{S}_{\textsf{snd}}^*\models\mathbb{A}_{\textsf{rcv}}^*$, which contradicts with the constraint.
    \vspace{1.5mm}

    \noindent(3) Let $q_e\geq 3$. We have $f_{1,j}=f_{2,j}=f_{3,j}=f_{4,j}=f_j$ and $\sum\nolimits_{j=1}^{q_e}f_j=1$. \vspace{1.5mm}
    
    \noindent(3.1) We firstly prove that $\mathcal{S}_{\snd}^{(j)}$ must have the same number of attributes, where $j\in[q_3]$. 
    
    Suppose that $\mathcal{S}_{\snd}^{(1)}$, $\mathcal{S}_{\snd}^{(2)}$, $...$ $\mathcal{S}_{\snd}^{(q_e)}$ has different number of attributes and $\ell_1^{(1)}>\ell_1^{(2)}>\cdots>\ell_1^{(q_e)}$. We deduce that 
    \begin{eqnarray*}                   \textsf{ct}_{6,\ell_1^{(1)}}^*&=&\big(H(u^{(1)})^{\tau_1f_1}\cdot h_{\ell_1^{(1)}}^e\big)^c=\big(H(u^{(1)})^{\tau_0}\cdot h_{\ell_1^{(1)}}^e\big)^c,\\  
    &&\text{~~~~~~~~~~~~~~~~~~~~~~~~~where~} h_{\ell_1^{(1)}}=H(u^{(1)}),\\    \textsf{ct}_{6,\ell_1^{(2)}}^*&=&\big(H(\hat{u}^{(1)})^{\tau_1f_1}\cdot H(\hat{u}^{(2)})^{\tau_2f_2}\cdot h_{\ell_1^{(2)}}^e\big)^c\\    
   &=&\big(H(\bar{u})^{\tau_0}\cdot h_{\ell_1^{(2)}}^e\big)^c,\text{~~~~~~~where~} h_{\ell_1^{(2)}}=H(\bar{u}),\\
    &&\cdots\\    \textsf{ct}_{6,\ell_1^{(q_e)}}^*&=&\big(H(\tilde{u}^{(1)})^{\tau_1f_1}\cdots H(\check{u}^{(q_e)})^{\tau_{q_e}f_{q_e}}\cdot h_{\ell_1^{(3)}}^e\big)^c,\\   
    &=&\big(H(\check{u})^{\tau_0}\cdot h_{\ell_1^{(3)}}^e\big)^c,\text{~~~~~~~where~} h_{\ell_1^{(3)}}=H(\check{u}),\\
    \textsf{ct}_7^*&=&\prod\limits_{j=1}^{q_e}(\delta_1^{\tau_j})^{f_{j}\cdot c'}\cdot\delta_1^{e\cdot c'}=\big(\delta_1^{\tau_0}\cdot\delta_1^{e}\big)^{c'},\\
    \textsf{ct}_8^*&=&\prod\limits_{j=1}^{q_e}(\delta_2^{\tau_j})^{f_{j}\cdot c''}\cdot\delta_2^{e\cdot c''}=\big(\delta_2^{\tau_0}\cdot\delta_2^{e}\big)^{c''},\\
    \textsf{ct}_9^*&=&\prod\limits_{j=1}^{q_e}(g_1^xh^{\tau_j})^{f_{j}\cdot c}\cdot h^{e\cdot c}=\big((g_1^xh^{\tau_0})\cdot h^{e}\big)^{c},
    \end{eqnarray*}
    where $f_j=f_{1,j}=f_{2,j}=f_{3,j}=f_{4,j}$.

    Therefore, we have $\tau_0=\tau_1f_1$ and $\tau_0=\sum\nolimits_{j=1}^{q_e}\tau_jf_j$. It can be deduced that $f_2=f_3=\cdots=f_{q_e}=0$. Since $\sum\nolimits_{j=1}^{q_e}f_j=1$, we have $f_1=1$. Then,
    $$\textsf{ct}_{6,i}^*=(\ek_{1,i}^*\cdot h_i^{e})^c,~\textsf{ct}_7^*=(\ek_2^*\cdot\delta_1^{e})^{c'},$$
    $$\textsf{ct}_8^*=(\ek_3^*\cdot\delta_2^{e})^{c''},~\textsf{ct}_9^*=(\ek_4^*\cdot h^{e})^c.$$
    $\mathcal{A}$ can easily calculate $\EK_{\snd}^*$ for $\mathcal{S}_{\textsf{snd}}^*$ and $\mathcal{S}_{\textsf{snd}}^*\models\mathbb{A}_{\textsf{rcv}}^*$, which contradicts with the constraint.\vspace{1.5mm}

    \noindent(3.2) Next, we prove that $\mathcal{S}_{\textsf{snd}}^{(1)}=\cdots=\mathcal{S}_{\textsf{snd}}^{(q_e)}$.
    If $\mathcal{S}_{\textsf{snd}}^{(1)}\neq\mathcal{S}_{\textsf{snd}}^{(2)}$, there must exist some certain attribute $u_i^{(1)}\neq\cdots\neq u_i^{(q_e)}$, where $i\in[\ell_1^{(3)}]$. Put $u_i^{(1)}$, $\cdots$, $u_i^{(q_e)}$ into the random oracle to get the hash values $h_i^{(1)}=H(u_i^{(1)})=g_1^{t_i^{(1)}}$, $\cdots$, $h_i^{(q_e)}=H(u_i^{(q_e)})=g_1^{t_i^{(q_e)}}$, where $t_i^{(1)},\cdots,t_i^{(q_e)}\stackrel{\$}{\leftarrow}\mathbb{Z}_p^*$ and $t_i^{(1)}\neq \cdots\neq t_i^{(q_e)}$.

    The ciphertext $\CT_{\snd}^*$ is constructed as    
    \begin{eqnarray*}                   \textsf{ct}_{6,i}^*&=&\prod\limits_{j=1}^{q_e}(H(u_i^{(j)})^{\tau_j})^{f_{j}\cdot c}\cdot h_i^{e\cdot c},\\    \textsf{ct}_7^*&=&\prod\limits_{j=1}^{q_e}(\delta_1^{\tau_j})^{f_{j}\cdot c'}\cdot\delta_1^{e\cdot c'}=\big(\delta_1^{\tau_0}\cdot\delta_1^{e}\big)^{c'},\\
    \textsf{ct}_8^*&=&\prod\limits_{j=1}^{q_e}(\delta_2^{\tau_j})^{f_{j}\cdot c''}\cdot\delta_2^{e\cdot c''}=\big(\delta_2^{\tau_0}\cdot\delta_2^{e}\big)^{c''},\\
    \textsf{ct}_9^*&=&\prod\limits_{j=1}^{q_e}(g_1^xh^{\tau_j})^{f_{j}\cdot c}\cdot h^{e\cdot c}=\big((g_1^xh^{\tau_0})\cdot h^{e}\big)^{c},
    \end{eqnarray*}
    where $\tau_0=\sum\nolimits_{j=1}^{q_e}\tau_jf_j$.

    If the ciphertext $\CT_{\snd}^*$ is valid, it must have
    $$\textsf{ct}_{6,i}^*=\prod\limits_{j=1}^{q_e}(H(u_i^{(j)})^{\tau_j})^{f_{j}\cdot c}\cdot h_i^{e\cdot c}=(H(\bar{u}_i)^{\tau_0}\cdot h_i^e)^c$$
    for some attribute $\bar{u}_i$. Let $\bar{h}_i=H(\bar{u}_i)=g_1^{\bar{t}_i}$, where $\bar{t}_i\stackrel{\$}{\leftarrow}\mathbb{Z}_p^*$. We have
    \begin{eqnarray*}
        \textsf{ct}_{6,i}^*&=&\prod\nolimits_{j=1}^{q_e}(H(u_i^{(j)})^{\tau_j})^{f_{j}\cdot c}\cdot h_i^{e\cdot c}\\
        &=&\big(g_1^{\sum\nolimits_{j=1}^{q_e}t^{(j)}\cdot \tau_jf_{j}}\cdot h_i^{e}\big)^c\\
        &=&(g^{\bar{t}_i\tau_0}\cdot h_i^e)^c=(H(\bar{u}_i)^{\tau_0}\cdot h_i^e)^c.
    \end{eqnarray*}
    Since $\tau_0=\sum\nolimits_{j=1}^{q_e}\tau_jf_j$ and $\sum\nolimits_{j=1}^{q_e}f_j=1$, we can deduce that 
    \begin{eqnarray*}
        &&\sum\nolimits_{j=1}^{q_e}t_i^{(j)}f_{j}\tau_j=\bar{t}_i\tau_0=\sum\nolimits_{j=1}^{q_e}\bar{t}_i\tau_jf_j\\
        &\Rightarrow&\sum\nolimits_{j=1}^{q_e}(t_i^{(j)}-\bar{t}_i)f_{j}\tau_j=0.
    \end{eqnarray*}

    As $\tau_1$, $\cdots$, $\tau_{q_e}$ are random numbers in $\mathbb{Z}_p^*$, we have $(t_i^{(j)}-\bar{t}_i)f_j=0$ for $j\in[q_e]$. Since $\sum\nolimits_{j=1}^{q_e}f_j=1$, we have $t_i^{(j)}=\bar{t}_i$, $f_j=1$ for some $j\in[q_e]$, and $f_{j'}=0$, for the $j'\in[q_e]$ and $j'\neq j$.
    In either way, if the challenge ciphertext $\CT_{\snd}^*$ is valid, the adversary can directly calculate $\EK_{\snd}^*$ for $\mathcal{S}_{\textsf{snd}}^*$ and $\mathcal{S}_{\textsf{snd}}^*\models\mathbb{A}_{\textsf{rcv}}^*$, which contradicts with the constraint.\vspace{2.5mm}    
    
    Therefore, we prove that  $\textsf{EK}$ are created for the same $\mathcal{S}_{\textsf{snd}}^*$ for all $j\in[q_e]$, where $\mathcal{S}_{\textsf{snd}}^*\models\mathbb{A}_{\textsf{rcv}}^*$.\\

    This completes the proof of Theorem \ref{theo:FEME-authn}.$\hfill\blacksquare$

\end{document}